%% file: RaR_IT_multigroup.tex
\documentclass[twoside,11pt,draftcls,onecolumn]{IEEEtran}
\usepackage{amsfonts}
\usepackage{amssymb}
\usepackage{mathrsfs}
\usepackage{verbatim}
\usepackage{epsfig}
\usepackage{subfigure}
\usepackage[centertags]{amsmath}
\usepackage{graphicx}

\newtheorem{thm}{\bf Theorem}

\newtheorem{lem}{{ Lemma}}
\newtheorem{prop}[thm]{{\it Proposition}}
\newtheorem{defn}{Definition}
\newtheorem{remark}[thm]{Remark}

\newtheorem{eg}{Example}[section]

\title{Multi-group ML Decodable Collocated and Distributed Space Time Block Codes}
\author{G. Susinder Rajan and B. Sundar Rajan
\thanks{This work was supported partly by the DRDO-IISc program on Advanced Research in Mathematical Engineering, and partly by the Council of Scientific \& Industrial Research (CSIR, India) Research Grant (22(0365)/04/EMR-II). The material in this paper was presented in part at the IEEE Information Theory Workshop, Chengdu, October 22-26, 2006 \cite{RaR1,RaR2}, in part at the IEEE Wireless Communications and Networking Conference, Hong Kong, March 11-15, 2007 \cite{RTR}, in part at the IEEE International Symposium on Information Theory, Nice, June 24-29, 2007 \cite{RaR3,RaR4} and also in part at the IEEE International Conference on Communications 2008 \cite{RaR5}, Beijing, May 2008. The authors are with the Department of Electrical Communication Engineering, Indian Institute of Science, Bangalore-560012, India. Email:\{susinder,bsrajan\}@ece.iisc.ernet.in.}}

\begin{document}
\maketitle
\begin{abstract}
In this paper, collocated and distributed space-time block codes (DSTBCs) which admit multi-group maximum likelihood (ML) decoding are studied. First the collocated case is considered and the problem of constructing space-time block codes (STBCs) which optimally tradeoff rate and ML decoding complexity is posed. Recently, sufficient conditions for multi-group ML decodability have been provided in the literature and codes meeting these sufficient conditions were called Clifford Unitary Weight (CUW) STBCs. An algebraic framework based on extended Clifford algebras is proposed to study CUW STBCs and using this framework, the optimal tradeoff between rate and ML decoding complexity of CUW STBCs is obtained for few specific cases. Code constructions meeting this tradeoff optimally are also provided. The paper then focuses on multi-group ML decodable DSTBCs for application in synchronous wireless relay networks and three constructions of four-group ML decodable DSTBCs are provided. Finally, the OFDM based Alamouti space-time coded scheme proposed by Li-Xia for a $2$ relay asynchronous relay network is extended to a more general transmission scheme that can achieve full asynchronous cooperative diversity for arbitrary number of relays. It is then shown how differential encoding at the source can be combined with the proposed transmission scheme to arrive at a new transmission scheme that can achieve full cooperative diversity in asynchronous wireless relay networks with no channel information and also no timing error knowledge at the destination node. Four-group decodable DSTBCs applicable in the proposed OFDM based transmission scheme are also given.
\end{abstract}
\section{Introduction}
\label{sec1}
Space-Time coding for Multiple Input Multiple Output (MIMO) systems has seen a lot of progress in the last decade. Starting from orthogonal designs \cite{TJC,TiH1,Lia} and quasi orthogonal designs \cite{Jaf,TBH,YGT,WWX,DYTGT}, several space-time block code (STBC) constructions have been proposed in the literature including the recently proposed space-time block codes from division algebras \cite{SRS1}, crossed product algebras \cite{SRS2}, co-ordinate interleaved orthogonal designs \cite{KhR} and Clifford algebras \cite{KaR1,KaR2,KaR3,KaR4}. Several aspects of space-time block codes (STBCs) have been studied in the literature. In the high SNR regime, two main aspects which dictate the error performance are diversity gain and coding gain. Of these two aspects, diversity gain has been well studied and presently many high rate, full diversity STBC constructions are available in the literature. An important class of such codes is the ones from division algebras \cite{SRS1}. Coding gain has remained an open problem not only for MIMO channels but also for Single Input Single Output channels and the AWGN channel. Later few more aspects such as the information lossless property \cite{HaH} and the diversity-multiplexing gain tradeoff \cite{ZhT,TaV} were introduced. Explicit STBCs satisfying these additional requirements were also obtained from division algebras \cite{KiR1,ORBV,EKPK}. However, there are still other important issues that need to be addressed. One such important issue is the Maximum Likelihood (ML) decoding complexity of STBCs. The lattice decoder or sphere decoder \cite{ViB,DCB} is known to be an efficient ML decoder. However, the complexity of a sphere decoder \cite{HaV,JaO} is also prohibitively large for high rate STBCs such as those from division algebras. For example, decoding a $4\times 4$ STBC from cyclic division algebras is equivalent to decoding a $32$ dimensional real lattice and performing a simulation to obtain an error performance curve can easily take several weeks. Thus it is not practically feasible to implement ML decoding for the `good' performing codes in the literature. It is well known \cite{YGT,WWX,KhR} that STBCs obtained from orthogonal designs (ODs) using QAM constellation admit single real symbol decoding and give full diversity. But for $4$ Tx antennas, an OD which provides a transmission rate of $1$ complex symbol per channel use does not exist \cite{TJC,TiH1,Lia}. However, it was shown in \cite{YGT,WWX,KhR,KaR1} that a single complex symbol decodable ($2$ real symbol decodable) full diversity STBC for $4$ transmit antennas can be constructed. Later in \cite{DYTGT,KaR2,KaR3,KaR4}, the general framework of multi-symbol decodable or multi-group decodable STBCs was introduced to improve the transmission rate. Multi-symbol or multi-group decodable STBCs admit ML decoding to be done separately for groups of symbols rather than all the symbols together thus reducing the ML decoding complexity. The class of STBCs from ODs correspond to the case of one real symbol per group. Thus it is clear that there is a tradeoff involving rate, ML decoding complexity and number of transmit antennas for full diversity STBCs. In the first part of this paper, measures of rate and ML decoding complexity are given and the problem of optimally trading off rate for ML decoding complexity within the framework of multi-group decodable STBCs is formally posed. A partial solution to this general problem is provided by characterizing this tradeoff for a certain specific class of STBCs called as Clifford Unitary Weight (CUW) STBCs \cite{KaR1,KaR2,KaR3,KaR4}. An algebraic framework based on extended Clifford algebras is introduced to study CUW STBCs. This framework is used to obtain the optimal rate-ML decoding complexity tradeoff and also to construct CUW STBCs meeting this tradeoff optimally. Recently in \cite{BHV,PGA,SaF}, a $2\times 2$ high rate, information lossless STBC with low ML decoding complexity and non-vanishing determinants has been discussed. This $2\times 2$ STBC is not a multi-group decodable STBC and such STBCs are not considered in this paper.

The second part of the paper focuses on constructing distributed space-time block codes (DSTBCs) with low ML decoding complexity for the Jing-Hassibi protocol \cite{JiH}. Distributed space-time coding \cite{LaW,JiH} is a coding technique for exploiting cooperative diversity in wireless relay networks wherein each relay is made to transmit a column of a space-time code thereby imitating a multiple antenna system. There are mainly two types of processing at the relay nodes that are widely discussed in the literature: (1) amplify and forward and (2) decode and forward. Throughout this paper, we focus only on amplify and forward based protocols for three reasons: (1) relay nodes are not required to decode and re-encode, (2) relay nodes do not require the channel knowledge for processing (this feature can permit a possible extension of the protocol to a completely non-coherent strategy) and (3) simpler processing at the relay nodes. In \cite{JiH}, Jing and Hassibi have proposed an amplify and forward based two phase transmission protocol for achieving cooperative diversity in wireless relay networks. This protocol essentially employs STBCs satisfying certain additional conditions to take care of the distributed nature. We call such codes satisfying certain additional conditions as DSTBCs to distinguish them from collocated STBCs. Analogous to the case of collocated STBCs, for large number of relays, the ML decoding complexity of DSTBCs becomes too prohibitive at the destination and thus is an important issue that needs to be addressed. Most of the previous works on DSTBCs \cite{OgH,EOK,EVAK} fail to address this issue. In \cite{KiR2}, full diversity, two-group ML decodable DSTBCs were constructed using division algebras. In \cite{JiJ1}, quasi-orthogonal STBCs were proposed for use as DSTBCs for the specific case of $4$ relays. In the second part of this paper, using the algebraic framework of extended Clifford algebras introduced in the first part, three new classes of four-group decodable full diversity DSTBCs for any number of relays are constructed.

The Jing-Hassibi protocol assumes that there is perfect symbol synchronization amongst the relay nodes and that the signals transmitted from the relays arrive at the same time at the destination. But achieving symbol synchronization among the geographically distributed relay nodes is a challenging and difficult task in practice. Several works in the literature \cite{LiX,GuX,WGV,XLi,YLX,ShX,MHSD,LZX} have recognized this as a major bottleneck and have proposed many coding and transmission techniques to mitigate the effects of symbol asynchronism. Most of the works based on amplify and forward propose methods to achieve full cooperative diversity in asynchronous wireless relay networks, but however fail to address the ML decoding complexity issues. In \cite{LiX}, a OFDM based Alamouti transmission scheme is proposed to combat the effects of symbol asynchronism. The Li-Xia transmission scheme is particularly interesting because of its associated simplicity and low decoding complexity. In this scheme, OFDM is implemented at the source node and time reversal/conjugation is performed at the relay nodes on the received OFDM symbols from the source node. The received signals at the destination after OFDM demodulation are shown to have the Alamouti STBC structure and hence single symbol maximum likelihood (ML) decoding can be performed. However, the Alamouti code is applicable only for the case of two relay nodes and for larger number of relays, the authors of \cite{LiX} propose to cluster the relay nodes and employ Alamouti code in each cluster. But this clustering technique provides diversity order of only two and fails to exploit the diversity available in the network. Motivated by the results of \cite{LiX}, in the third part of this paper it is shown that the DSTBCs proposed in this paper can be used along with OFDM to achieve full asynchronous cooperative diversity for any number of relays along with low ML decoding complexity.

Finally it is shown how differential encoding at the source node can be combined with the proposed OFDM based transmission scheme to arrive at a new transmission scheme that provides full cooperative diversity in asynchronous relay networks with no channel information and also no timing error knowledge at any of the nodes.

The main contributions of this paper can be summarized as follows:
\begin{itemize}
\item A new measure of rate of an STBC is defined and the problem of optimal tradeoff between rate and ML decoding complexity within the framework of multi-group ML decodable STBCs is posed. An algebraic framework based on extended Clifford algebras is introduced for studying CUW STBCs. Using this algebraic framework and tools from representation theory of groups, the optimal tradeoff between rate and ML decoding complexity of CUW STBCs is characterized for certain specific cases.
\item Constructions of CUW STBCs meeting this optimal tradeoff for the specific cases are also provided. The ABBA construction first proposed in \cite{TBH} is shown to be a certain specific matrix representation of extended Clifford algebras and hence they fall under the class of CUW STBCs. The contributions on multi-group ML decodable collocated STBCs are described in Section \ref{sec2}.
\item The Jing-Hassibi protocol \cite{JiH} is generalized to allow non-unitary matrices at the relays. The necessary and sufficient conditions needed for DSTBCs to be multi-group ML decodable are identified and three new classes of four group ML decodable DSTBCs which achieve full cooperative diversity for any number of relays are also provided. To the knowledge of the authors, these are the first known DSTBCs that achieve the least possible ML decoding complexity compared to all other DSTBC constructions, having the same transmission rate in complex symbols per channel use in the literature. This contribution is detailed in Section \ref{sec3}.
\item The OFDM based Alamouti transmission scheme for $2$ relays in \cite{LiX} is extended to a more general transmission scheme that can achieve full asynchronous cooperative diversity for any number of relays. Sufficient conditions for a DSTBC to be compatible with the requirements of this OFDM based transmission scheme are provided and the four-group decodable DSTBCs in this paper are shown to satisfy these sufficient conditions.
\item It is shown how differential encoding at the source node can be combined with the proposed OFDM based transmission scheme to arrive at a new transmission scheme that provides full cooperative diversity in asynchronous relay networks with no channel information and also no timing error knowledge at any of the nodes. All the results based on OFDM for asynchronous relay networks comprise Section \ref{sec4}.
\end{itemize}

\subsection{Notation}
Vectors and matrices are denoted by lowercase bold letters and uppercase bold letters respectively. $\mathbf{I_m}$, $\mathbf{0_m}$ denote an $m\times m$ identity matrix and $m\times m$ all zero matrix respectively. $\mathbf{I}$ and $\mathbf{0}$ are used to denote an identity matrix and an all zero matrix respectively having an appropriate size depending on the context. For a set $A$, the cardinality of $A$ is denoted by $|A|$. A null set is denoted by $\phi$. For a matrix, $(.)^T$, $(.)^*$ and $(.)^H$ denote transposition, conjugation and conjugate transpose operations respectively. For a complex matrix $\mathbf{X}$, the matrices $\mathbf{X_I}$ and $\mathbf{X_Q}$ denote the matrices obtained by taking the real and imaginary parts of $\mathbf{X}$ respectively. If $B$ is a module over a base ring $R$, then $End_{R}B$ denotes the set of all $R$ linear maps from $B$ to $B$. For sets $A_1$ and $A_2$, the Cartesian product of $A_1$ and $A_2$ is denoted by $A_1\times A_2$. For groups $G_1$ and $G_2$, the direct product of $G_1$ and $G_2$ is denoted by $G_1\times G_2$. For vector spaces $V_1$ and $V_2$, the tensor product of $V_1$ and $V_2$ is denoted by $V_1\otimes V_2$. For a vector space $V$, $GL(V)$ is used to denote the set of invertible linear maps from $V$ to $V$.

\section{Multi-group ML Decodable Collocated STBCs}
\label{sec2}
In this section, multi-group ML decodable collocated STBCs are discussed. In Subsection \ref{subsec2_1}, a relation between STBCs and linear space-time designs is given and using this relation, measures of rate and ML decoding complexity of STBCs are defined. In Subsection \ref{subsec2_2}, linear space-time designs are classified based upon the classification done in \cite{KaR1} for single complex symbol decodable codes. An algebraic framework based on extended Clifford algebras is introduced to study a class of linear space-time designs called Clifford unitary weight designs. Using this algebraic framework, the optimal tradeoff between rate and ML decoding complexity of STBCs from Clifford unitary weight designs is characterized under some conditions in Subsection \ref{subsec2_3}.

\subsection{STBCs and Linear Space-Time Designs}
\label{subsec2_1}

In this subsection, a connection between STBCs and linear space-time designs is established. Using this relation, measures of rate and ML decoding complexity of a STBC are then defined.

\begin{defn}
\label{defn_stbc}
A STBC $\mathscr{C}$ of size $T \times N_T$  is a finite set of $T\times N_T$ complex matrices.
\end{defn}

Let $N_T$ denote the number of transmit antennas, $N_R$ denote the number of receive antennas and $T$ denote the number of channel uses consumed for transmitting a space-time codeword. Then the rate of transmission in bits per channel use (bpcu) of a STBC as in Definition \ref{defn_stbc} is given by $\frac{\log_2 |\mathscr{C}|}{T}$ bpcu. In this paper, we use a different measure of rate which is motivated by basic concepts of dimension in linear algebra. This measure is also indicative of the coding gain of the STBC and several examples of STBCs in the literature are discussed to illustrate the significance of the new measure of rate introduced in this paper.

Note that the set of all $T\times N_T$ complex matrices is a vector space over the field of real numbers $\mathbb{R}$ and has a dimension of $2TN_T$ over $\mathbb{R}$. Consider the subspace $\langle\mathscr{C}\rangle$ spanned by the codewords, i.e., the elements of $\mathscr{C}$. Let $K$ denote the dimension of $\langle\mathscr{C}\rangle$ over $\mathbb{R}$ and let $\mathbf{A_i}, i=1,\dots,K \in \mathbb{C}^{T\times N_T}$ be a basis for $\langle\mathscr{C}\rangle$. Then every element of $\mathscr{C}$ can be expressed as $\sum_{i=1}^{K}x_i\mathbf{A_i}$ for some $x_i, i=1,\dots,K \in \mathbb{R}$. If we think of the $x_i$'s as real variables and $\mathbf{S}(\mathbf{s}=\left[\begin{array}{cccc}x_1 & x_2 & \dots & x_K\end{array}\right]^T)=\sum_{i=1}^{K}x_i\mathbf{A_i}$ as a matrix whose entries are complex linear functions of the real variables, then the STBC $\mathscr{C}$ can be expressed as
\begin{equation}
\label{eqn_lstd}
\mathscr{C}=\left\{S(\mathbf{s})|\mathbf{s}\in\mathscr{A}\right\}
\end{equation}
for some finite subset $\mathscr{A}\subset\mathbb{R}^K$.

\begin{defn}
\label{defn_lstd}
A linear space-time design (LSTD) $\mathbf{S}(\left[\begin{array}{cccc}x_1 & x_2 & \dots & x_K\end{array}\right]^T)$ of size $T\times N_T$ in real variables $x_1,x_2,\dots,x_K$ is a $T\times N_T$ matrix which can be expressed as $\sum_{i=1}^{K}x_i\mathbf{A_i}$ for some\\ $\mathbf{A_i}, i=1,2,\dots, K\in \mathbb{C}^{T\times N_T}$ which are linearly independent over the field of real numbers.
\end{defn}

The notion of linear independence of weight matrices of a LSTD over $\mathbb{R}$ has not been stressed or mentioned explicitly in most previous works though it has been implicitly assumed.

Notice that \eqref{eqn_lstd} specifies a way to describe STBCs using linear space-time designs (LSTDs) and also explicitly provides a method to encode STBCs. From an encoding perspective, the real variables can be thought of as modulating the matrices $\mathbf{A_i}, i=1,\dots,K$. Hence we call the matrices $\mathbf{A_i}, i=1,\dots,K$ as basis matrices or modulation matrices or weight matrices. The vector of real variables $\mathbf{s}$ takes values from $\mathscr{A}\subset\mathbb{R}^K$. We call $\mathscr{A}$ as the signal set. The connection between STBCs and LSTDs is pictorially depicted in Fig. \ref{fig_stbc_lstd}.

\begin{figure}[htp]
\centering
\input{stbc_lstd.pstex_t}
\caption{STBCs and linear space-time designs}
\label{fig_stbc_lstd}
\end{figure}
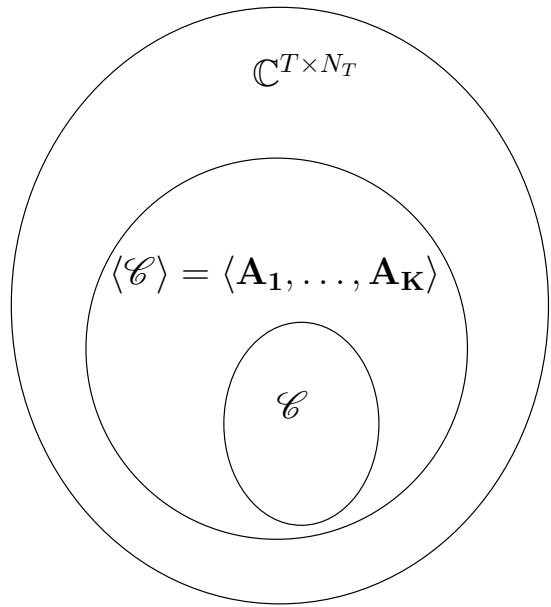

\begin{remark}
Note that for a given STBC $\mathscr{C}$ the set of basis matrices $\mathbf{A_i}, i=1,\dots,K$ along with the associated signal set $\mathscr{A}$ is not unique, i.e., there may exist another set of basis matrices with some other associated signal set that results in the same STBC $\mathscr{C}$. Note also that it is not necessary that the basis matrices have to be codewords. We shall see in the sequel that the choice of basis matrices and signal set controls the encoding as well as decoding complexity. However, it is important to note that $K$ is unique to the STBC $\mathscr{C}$.
\end{remark}

Thus a STBC can be thought of as a subset of a subspace of dimension $K$. Thus designing a STBC can be done in two steps: First choose a subspace of dimension $K$ (choose a LSTD) and then choose a subset of required cardinality (choose the signal set $\mathscr{A}$) within the chosen subspace.

\subsubsection{Measure of Rate}
\label{subsubsec2_1_1}
~\\
In this paper, we use the following definition of rate of a STBC.

\begin{defn}
\label{defn_rate}
Rate of a STBC $\mathscr{C}=\frac{\mathrm{dimension}(\langle\mathscr{C}\rangle)}{T}=\frac{K}{T}$ dimensions per channel use.
\end{defn}

Note that the unit of rate of a STBC according to Definition \ref{defn_rate} is dimensions per channel use (dpcu). Since there are $K$ real variables which modulate $K$ modulation matrices, we can view it as though we are sending $K$ real symbols (one on each dimension) in $T$ channel uses. Alternatively, we can pair two real variables at a time and view it as $\frac{K}{2}$ complex symbols being transmitted in $T$ channel uses. We would like to mention that most of the previous works on STBCs follow the convention of measuring rate in complex symbols per channel use which in our case is $\frac{K}{2T}$ complex symbols per channel use and is simply proportional to rate as per Definition \ref{defn_rate}. Though the terminology of basis matrices and rate have been used previously in the literature (for example see \cite{HTW}), to the knowledge of the authors, rate of a STBC has not been defined explicitly and clearly as in Definition \ref{defn_rate} although many works in the literature may be measuring rate in a similar way. Note that if linear independence of basis matrices is not retained and if rate were to be measured by simply counting the number of complex variables in the LSTD, then one can claim to have any arbitrary rate of transmission which can be quite deceptive at times. The notion of linear independence makes things clear and avoids such confusions. Definition \ref{defn_rate} is particularly useful because it essentially allows to define rate of a LSTD, hence allowing us to separate the study of LSTDs from STBCs.  Also, we argue that rate as per Definition \ref{defn_rate} is a first order indicative of coding gain and hence is a parameter which has to be maximized. Intuitively, the higher the dimension, the more efficiently we can pack codewords in it optimizing some criteria. One of the criteria of interest is to maximize the coding gain which is given by $\min_{\mathbf{C_1},\mathbf{C_2}\in\mathscr{C}}\mathrm{det}\left((\mathbf{C_1}-\mathbf{C_2})^H(\mathbf{C_1}-\mathbf{C_2})\right)$.

Recall that even in the case of classical linear error correcting codes over finite fields, rate was defined as the ratio of the dimension of the subspace spanned by the codewords to the number of channel uses. In the case of classical linear error correcting codes, the code itself is a subspace whereas in the case of STBCs, the code is a subset of a subspace. The following examples of existing STBCs reinforce the statement that rate as per Definition \ref{defn_rate} is a first order indicative of coding gain.

\begin{eg}
\label{eg_al_gc}
Let us consider the Alamouti code \cite{Ala} and the Golden code \cite{BRV} which are given by:
$\left[\begin{array}{cc}
x_1+ix_2 &  -x_3+ix_4\\
x_3+ix_4 &  x_1-ix_2
\end{array}\right]$ and\\
$\left[\begin{array}{cc}
(x_1+ix_2)\alpha+(x_3+ix_4)\alpha\theta & (x_5+ix_6)\alpha+(x_7+ix_8)\alpha\theta\\
i((x_5+ix_6)\bar{\alpha}+(x_7+ix_8)\bar{\alpha}\bar{\theta}) & (x_1+ix_2)\bar{\alpha}+(x_3+ix_4)\bar{\alpha}\bar{\theta}
\end{array}\right]$ respectively where, $\theta=\frac{1+\sqrt{5}}{2}$, $\bar{\theta}=\frac{1-\sqrt{5}}{2}$, $\alpha=1+i(1-\theta)$ and $\bar{\alpha}=1+i(1-\bar{\theta})$. In both cases, the real variables are allowed to take values independently from a finite subset of $\mathbb{Z}$. It can be checked that there are $4$ basis matrices for the Alamouti code and $8$ basis matrices for the Golden code. Thus the rate of Alamouti code and Golden code are $2$ dpcu and $4$ dpcu respectively and it is well known \cite{BRV} that the Golden code outperforms the Alamouti code when they are both compared with the same transmission rate in bpcu.
\end{eg}

\begin{eg}
\label{eg_od_qod}
Let us consider the $4\times 4$ OD and the $4\times 4$ quasi orthogonal design. They are given by:
$\left[\begin{array}{cccc}
x_1+ix_2 &  -x_3+ix_4 & -x_5+ix_6 & 0\\
x_3+ix_4 &  x_1-ix_2   &  0 & -x_5+ix_6\\
x_5+ix_6 & 0 & x_1-ix_2 & x_3-ix_4\\
0 & x_5+ix_6 & -x_3-ix_4 & x_1+ix_2
\end{array}\right]$ and\\ $\left[\begin{array}{cccc}
x_1+ix_2 &  -x_3+ix_4 & x_5+ix_6 &  -x_7+ix_8\\
x_3+ix_4 &  x_1-ix_2 & x_7+ix_8 &  x_5-ix_6\\
x_5+ix_6 &  -x_7+ix_8 & x_1+ix_2 &  -x_3+ix_4\\
x_7+ix_8 &  x_5-ix_6 & x_3+ix_4 &  x_1-ix_2
\end{array}\right]$ respectively. Their respective rates can be verified to be $\frac{3}{2}$ dpcu and $2$ dpcu respectively. STBCs from quasi orthogonal designs are known to outperform STBCs from ODs \cite{Jaf,TBH} for the same transmission rate in bpcu.
\end{eg}

The above examples show that given two STBCs having the same number of codewords, the one having higher rate as per Definition \ref{defn_rate} outperforms the other in most cases, thus providing a good motivation for Definition \ref{defn_rate}.

\subsubsection{Measure of ML decoding complexity}
\label{subsubsec2_1_2}
~\\
Towards defining a measure for ML decoding complexity, let us first define a measure of encoding complexity. If we use \eqref{eqn_lstd} for encoding a STBC using LSTDs, we see that in general one needs to choose an element from $\mathscr{A}$ and then substitute for the real variables $x_1,x_2,\dots,x_K$ in the LSTD. This method of encoding clearly requires a lookup table (memory) with $|\mathscr{A}|$ entries. However, if the signal set $\mathscr{A}$ is a Cartesian product of $g$ smaller signal sets in dimension $\frac{K}{g}$, then the complexity can be reduced. To be precise, if $\mathscr{A}=\mathscr{A}_1\times\mathscr{A}_2\times\dots\times\mathscr{A}_g$ where each $\mathscr{A}_i\subset \mathbb{R}^{\frac{K}{g}}$ with cardinality $|\mathscr{A}|^{\frac{1}{g}}$, then the STBC $\mathscr{C}$ itself decomposes as a sum of $g$ different STBCs, which is shown below. Let $K=g\lambda$. Then by appropriately reordering/relabeling the real variables we can assume without loss of generality \footnote{Here we have assumed that the first $\lambda$ real variables belong to first group and the second $\lambda$ real variables belong to the second group and the last $\lambda$ real variables belong to the $g$-th group. In general, the partitioning of real variables into $g$-groups can be quite arbitrary.} that $\mathbf{S}(\mathbf{s})=\sum_{i=1}^{K}x_i\mathbf{A_i}=\mathbf{S_1}(\mathbf{s_1})+\mathbf{S_2}(\mathbf{s_2})+\dots+\mathbf{S_g}(\mathbf{s_g})$ where, $\mathbf{S_i}(\mathbf{s_i})=\sum_{j=(i-1)\lambda+1}^{i\lambda}x_j\mathbf{A_j}$ and\\ $\mathbf{s_i}=\left[\begin{array}{cccc}x_{(i-1)\lambda+1} & x_{(i-1)\lambda+2} & \dots x_{i\lambda}\end{array}\right]^T$. Hence the STBC decomposes as $\mathscr{C}=\sum_{i=1}^{g}\mathscr{C}_i$ where,
$$
\begin{array}{ccc}
\mathscr{C}_1&=&\left\{\mathbf{S_1}(\mathbf{s_1})|\mathbf{s_1}\in\mathscr{A}_1\right\}\\
\mathscr{C}_2&=&\left\{\mathbf{S_2}(\mathbf{s_2})|\mathbf{s_2}\in\mathscr{A}_2\right\}\\
&\vdots\\
\mathscr{C}_g&=&\left\{\mathbf{S_g}(\mathbf{s_g})|\mathbf{s_g}\in\mathscr{A}_g\right\}
\end{array}
$$

\begin{defn}\cite{RaR6}
\label{defn_g_enc}
A STBC $\mathscr{C}=\left\{S(\mathbf{s})|\mathbf{s}\in\mathscr{A}\subset\mathbb{R}^K\right\}$ is said to $g$-group encodable or $\frac{K}{g}$ real symbol encodable (or $\frac{K}{2g}$ complex symbol decodable) if $\mathscr{A}=\mathscr{A}_1\times\mathscr{A}_2\times\dots\times\mathscr{A}_g$ where each $\mathscr{A}_i\subset \mathbb{R}^{\frac{K}{g}}$ with cardinality $|\mathscr{A}|^{\frac{1}{g}}$.
\end{defn}

The encoding of a $g$-group encodable STBC is pictorially shown in Fig. \ref{fig_g_enc}. Thus the encoding complexity of $g$-group encodable STBCs is $g(|\mathscr{C}|^{\frac{1}{g}})$. Note that in addition if $\mathscr{A}_1=\mathscr{A}_2=\dots=\mathscr{A}_g$ then the memory required for encoding is also minimized.

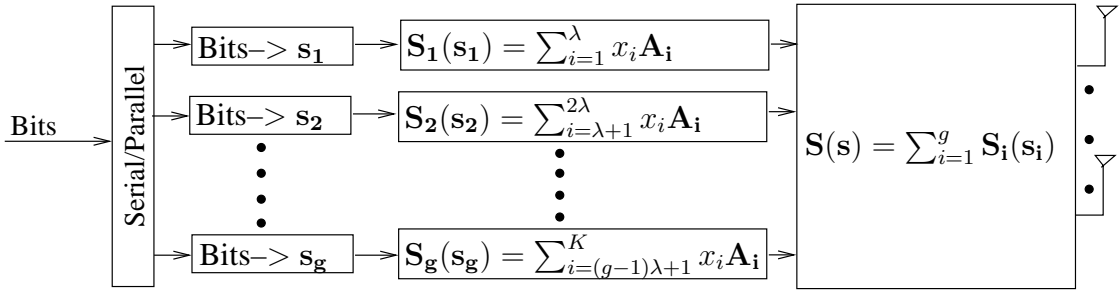
\begin{figure}[h]
\centering
\input{g_enc.pstex_t}
\caption{Encoding for a $g$-group encodable STBC}
\label{fig_g_enc}
\end{figure}

\begin{eg}
Consider the example of the Golden code which was discussed in Example \ref{eg_al_gc}. As per Definition \ref{defn_g_enc}, the Golden code is $8$-group encodable or single real symbol encodable.
\end{eg}

Thus we have seen how a $g$-group encodable STBC $\mathscr{C}$ decomposes into a sum of $g$ STBCs\\ $\mathscr{C}_i, i=1,\dots,g$ and thus admits independent encoding of the $\mathscr{C}_i$'s. A natural question that follows is: Under what conditions does  a $g$-group encodable STBC $\mathscr{C}$ admit independent decoding of the constituent $\mathscr{C}_i$'s? Towards that end, let us look at the ML decoding metric. Let $\mathbf{X}$ be the transmitted codeword of size $T\times N_T$, $\mathbf{H}$ be the $N_T\times N_R$ channel matrix and $\mathbf{Y}$ be the received matrix of size $T\times N_R$. Then, the ML decoder is given by

\begin{equation}
\label{eqn_ml_dec}
\mathbf{\hat{X}}=\arg\min_{\mathbf{X}\in\mathscr{C}}\parallel \mathbf{Y}-\mathbf{X}\mathbf{H}\parallel_F^2.
\end{equation}

For a $g$-group encodable STBC $\mathscr{C}$, $\mathbf{X}=\sum_{i=1}^{g}\mathbf{X_i}$ for some $\mathbf{X_i}\in\mathscr{C}_i$. It can be shown \cite{KhR,KaR1,KaR2} that if the basis matrices $\mathbf{A_i}, i=1,\dots,K$ satisfy the condition

\begin{equation}
\label{eqn_g_dec_cond}
\mathbf{A_i}^H\mathbf{A_j}+\mathbf{A_j}^H\mathbf{A_i}=\mathbf{0}~\mathrm{whenever}~\mathbf{A_i}\in \langle\mathscr{C}_k\rangle, \mathbf{A_j}\in\langle\mathscr{C}_l\rangle,~k\neq l
\end{equation}

\noindent then the ML decoder decomposes as

\begin{equation}
\mathbf{\hat{X}}=\sum_{i=1}^{g}\arg\min_{\mathbf{X_i}\in\mathscr{C}_i}\parallel \mathbf{Y}-\mathbf{X_iH}\parallel_F^2.
\end{equation}

In other words, the component STBCs $\mathscr{C}_i$'s can then be decoded independently. It can also be shown \cite{KhR,KaR1,KaR2} that \eqref{eqn_g_dec_cond} is a necessary condition for this to happen.

\begin{remark}
Note that the subspaces $\langle\mathscr{C}_i\rangle, i=1,\dots,K$ intersect trivially, i.e., $\langle\mathscr{C}_k\rangle\cap\langle\mathscr{C}_l\rangle=0$. Thus $\langle\mathscr{C}\rangle=\langle\mathscr{C}_1\rangle\oplus\langle\mathscr{C}_2\rangle\oplus\dots\oplus\langle\mathscr{C}_g\rangle$. If the condition in \eqref{eqn_g_dec_cond} is satisfied for the basis matrices, then it implies that $\mathbf{A}^H\mathbf{B}+\mathbf{B}^H\mathbf{A}=\mathbf{0}, \forall~ \mathbf{A}\in\langle\mathscr{C}_k\rangle, \mathbf{B}\in\langle\mathscr{C}_l\rangle, k\neq l$. In other words, this becomes a property of the two subspaces $\langle\mathscr{C}_k\rangle$ and $\langle\mathscr{C}_l\rangle$.
\end{remark}

\begin{defn}
\label{defn_g_dec}\cite{RaR6}
A STBC $\mathscr{C}=\left\{S(\mathbf{s})|\mathbf{s}\in\mathscr{A}\subset\mathbb{R}^K\right\}$ is said to $g$-group decodable or $\frac{K}{g}$ real symbol decodable (or $\frac{K}{2g}$ complex symbol decodable) if $\mathscr{C}$ is $g$-group encodable and if the associated basis matrices satisfy \eqref{eqn_g_dec_cond}.
\end{defn}

\begin{eg}
All STBCs obtained from ODs are single real symbol decodable if every real variable in the OD takes values independently from a PAM (Pulse Amplitude Modulation) signal set. As an example, consider the Alamouti code that was previously discussed in Example \ref{eg_al_gc}. The associated basis matrices are $\mathbf{A_1}=\left[\begin{array}{cc}1 & 0\\0 & 1\end{array}\right]$, $\mathbf{A_2}=\left[\begin{array}{cc}i & 0\\0 & -i\end{array}\right]$, $\mathbf{A_3}=\left[\begin{array}{cc}0 & -1\\1 & 0\end{array}\right]$ and $\mathbf{A_4}=\left[\begin{array}{cc}0 & i\\ i & 0\end{array}\right]$. It can be checked that they satisfy the condition in \eqref{eqn_g_dec_cond} for $g=4$. In this case, $\mathbf{S_1}(\mathbf{s_1})=x_1\mathbf{A_1}$, $\mathbf{S_2}(\mathbf{s_2})=x_2\mathbf{A_2}$, $\mathbf{S_3}(\mathbf{s_3})=x_3\mathbf{A_3}$ and
$\mathbf{S_4}(\mathbf{s_4})=x_4\mathbf{A_4}$. Hence the Alamouti code is single real symbol decodable.
\end{eg}

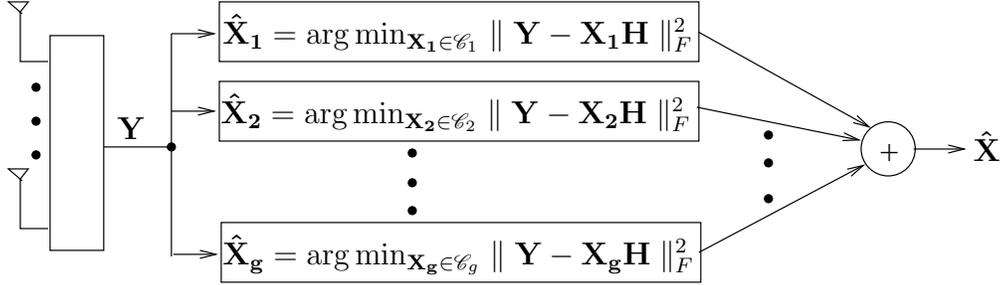
\begin{figure}[htp]
\centering
\input{g_dec.pstex_t}
\caption{ML decoding for a $g$-group ML decodable STBC}
\label{fig_g_dec}
\end{figure}

The ML decoding for a $g$-group decodable code is illustrated pictorially in Fig. \ref{fig_g_dec}. It is clear that the decoding complexity is reduced for $g$-group decodable STBCs from $|\mathscr{C}|$ computations to $g|\mathscr{C}|^{\frac{1}{g}}$ computations. Further, we know that the sphere decoder \cite{ViB,DCB} is an efficient ML decoder if vector $\mathbf{s}$ takes values from a lattice constellation. Moreover, it has been shown  \cite{HaV,JaO} that the average complexity of a sphere decoder depends on the dimension of the equivalent lattice \cite{DCB} and more or less independent of the size of the code. Thus, we can take the dimension of the corresponding equivalent lattice as a measure of the sphere decoder complexity. For a general STBC, this dimension is equal to $K$ whereas for $g$-group ML decodable STBCs, it is $\frac{K}{g}=\lambda$. Thus the expected as well as the worst case ML decoding complexity is lesser for $g$-group ML decodable STBCs.

\subsubsection{Full diversity}
\label{subsubsec2_1_3}
~\\
Apart from rate and ML decoding complexity, yet another important aspect of STBCs is the diversity gain. Diversity gain is a measure of the slope of the error probability versus the SNR when plotted on a log-log scale and this is given by $N_R(\min_{\mathbf{C_1},\mathbf{C_2}\in\mathscr{C}} \mathrm{rank}(\mathbf{C_1}-\mathbf{C_2}))$. Thus full diversity of $N_RN_T$ is achieved by a STBC if the coding gain is not equal to zero.

\subsubsection{Problem statement of optimal Rate-ML decoding complexity tradeoff}
\label{subsubsec2_1_4}
~\\
Having surveyed three important aspects of rate, ML decoding complexity and diversity for a STBC, we can now pose the problem of rate-ML decoding complexity tradeoff. This problem can be formally stated in two equivalent ways which are listed down as given below.

\begin{enumerate}
\item Given $\lambda$, $T$ and $N_t$ what is the maximum rate of any full diversity STBC?
\item Given $g$, $T$ and $N_t$ what is the maximum rate of any full diversity STBC?
\end{enumerate}

If $\lambda=1$ and $N_T=T$, then the solution is precisely the STBCs from square orthogonal designs constructed in \cite{TiH1,Lia} for which the maximum rate is $\frac{\lceil\log_2 N_T\rceil+1}{2^{\lceil\log_2 N_T\rceil-1}}$ dpcu. In this paper, the maximum rate of a certain class of full diversity square STBCs from Clifford unitary weight designs is characterized for $\lambda=2^a$.

The following example illustrates that full diversity and encoding/decoding complexity are related indirectly.

\begin{eg}
\label{eg_4x4_CIOD}
Consider the $4\times 4$ co-ordinate interleaved orthogonal design (CIOD) \cite{KhR} given by $\mathbf{S}=\left[\begin{array}{cccc}
x_1+ix_2 & -x_3+ix_4 & 0 & 0\\
x_3+ix_4 & x_1-ix_2 & 0 & 0\\
0 & 0 & x_5+ix_6 & -x_7+ix_8\\
0 & 0 & x_7+ix_8 & x_5-ix_6
\end{array}\right]$. The weight matrices of the above LSTD satisfy $\mathbf{A_i}^H\mathbf{A_j}+\mathbf{A_j^H}\mathbf{A_i}=\mathbf{0},~\forall~i\neq j$. Let the notation $\Delta$ stand for the codeword difference matrix. Note that $\det\left(\mathbf{\Delta S}^H\mathbf{\Delta S}\right)=\left(\sum_{i=1}^{4}\Delta x_i^2\right)^2\left(\sum_{i=5}^{8}\Delta x_i^2\right)^2$, which will equal to zero for some pair of codeword matrices if all the $8$ real variables are allowed to take values independently. Hence, it is not possible to obtain a full diversity single real symbol decodable STBC from the above LSTD. However, by entangling two real variables, as for example  $\left\{x_1,x_5\right\}$, $\left\{x_2,x_6\right\}$, $\left\{x_3,x_7\right\}$, $\left\{x_4,x_8\right\}$ and then allowing them to take values from a rotated QAM constellation (rotating a QAM constellation entangles the variables), a full diversity, single complex symbol ML decodable STBC can be obtained \cite{KhR}. The resulting STBC will be $4$-group ML decodable or $2$-real symbol ML decodable and its associated four constituent STBCs are given by $\mathbf{S_1}(\mathbf{s_1})=x_1\mathbf{A_1}+x_5\mathbf{A_5}$, $\mathbf{S_2}(\mathbf{s_2})=x_2\mathbf{A_2}+x_6\mathbf{A_6}$, $\mathbf{S_3}(\mathbf{s_3})=x_3\mathbf{A_3}+x_7\mathbf{A_7}$ and
$\mathbf{S_4}(\mathbf{s_4})=x_4\mathbf{A_4}+x_8\mathbf{A_8}$.
\end{eg}

Example \ref{eg_4x4_CIOD} shows that the requirement of full diversity can sometimes demand an increase in the encoding complexity and hence the decoding complexity even if the associated weight matrices satisfy condition \eqref{eqn_g_dec_cond} for $\lambda=1$. Thus, it is clear that full diversity and encoding/decoding complexity are inter-related and there exists a tradeoff between the two.

\subsection{Clifford Unitary Weight Designs and extended Clifford algebras}
\label{subsec2_2}

First, let us classify square LSTDs (as done in \cite{KaR1} for single complex symbol decodable codes). LSTDs can be broadly classified as unitary weight designs (UWDs) and non unitary weight designs (NUWDs). A UWD is one for which all the weight matrices are unitary and NUWDs are defined as those which are not UWDs. Clifford unitary weight designs (CUWDs) are a proper subclass of UWDs whose weight matrices satisfy certain sufficient conditions for $g$-group ML decodability. To state those sufficient conditions, let us list down the weight matrices of a CUWD in the form of an array as shown in Table \ref{table_CUWD}.

\begin{table}[h]
\caption{Structure of CUWDs}
\label{table_CUWD}
\begin{center}
\begin{tabular}{c|ccc}
$\mathbf{A_1}$ & $\mathbf{A_{\lambda+1}}$ & \dots & $\mathbf{A_{(g-1)\lambda+1}}$\\
\hline
$\mathbf{A_2}$ & $\mathbf{A_{\lambda+2}}$ & \dots & $\mathbf{A_{(g-1)\lambda+2}}$\\
$\vdots$ & $\vdots$ & $\ddots$ & $\vdots$\\
$\mathbf{A_{\lambda}}$ & $\mathbf{A_{2\lambda}}$ & \dots & $\mathbf{A_{K}}$
\end{tabular}
\end{center}
\end{table}

For simplicity, the grouping is assumed to be as follows: All the weight matrices in one column belong to one group. The weight matrices of CUWDs satisfy the following sufficient conditions for $g$-group ML decodability.

\begin{enumerate}
\item $\mathbf{A_1}=\mathbf{I}$.
\item The unitary matrices in the first row except $\mathbf{A_1}$ should form a Hurwitz-Radon family \cite{TJC,TiH1,Lia}. In other words, all the matrices in the first row except $\mathbf{A_1}$ should square to $-\mathbf{I}$ and should pair-wise anti-commute among themselves.
\item The unitary matrices in the first column should square to $\mathbf{I}$ and should commute with all the matrices in the first row and first column.
\item The unitary matrix in the $i$-th row and the $j$-th column is equal to $\pm\mathbf{A_iA_{(j-1)\lambda+1}}$.
\end{enumerate}

It can be checked that the above four conditions together imply that the necessary and sufficient condition for $g$-group ML decodability in \eqref{eqn_g_dec_cond} is satisfied and hence any CUWD is $g$-group ML decodable. Note that when $\lambda=1$, CUWDs become ODs \cite{TJC,TiH1,Lia}. Similarly the co-ordinate interleaved orthogonal designs proposed in \cite{KhR} are a proper subclass of NUWDs. The single complex symbol ML decodable STBCs in \cite{YGT} are also CUWDs \cite{KaR1}. Fig. \ref{fig_class_lstd} pictorially shows the broad classification of LSTDs.

\begin{figure}[htp]
\centering
\includegraphics{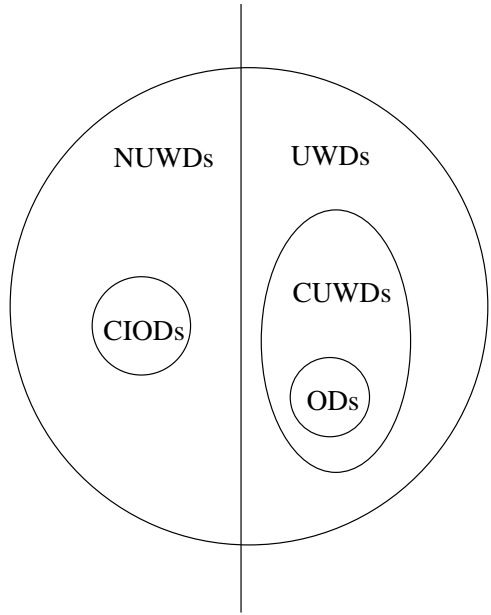}
\caption{Classification of LSTDs}
\label{fig_class_lstd}
\end{figure}

\subsubsection{Full diversity lattice constellations for Clifford Unitary Weight designs}
\label{subsubsec2_2_1}
~\\
An important advantage of CUWDs is that full diversity STBCs can be obtained from them without increasing the encoding/decoding complexity contrary to the case of CIODs (see Example \ref{eg_4x4_CIOD}) wherein real variables from different groups have to be entangled for full diversity. Moreover, explicit lattice constellations that optimize the coding gain can be obtained for CUWDs, thus admitting the use of a lattice/sphere decoder. In \cite{KaR1,KaR2,KaR3,KaR4}, few constructions of CUWDs are available  and the aspect of full diversity has been addressed in detail. In this paper, we only provide a brief outline of the basic idea (described below) and illustrate the procedure with an example. Later in the proof of Theorem \ref{thm_maxrate}, we also provide a new construction of CUWDs.

For a CUWD, $\mathrm{det}\left(\Delta \mathbf{S}(\mathbf{s})^H\Delta \mathbf{S}(\mathbf{s})\right)=\mathrm{det}\left(\sum_{i}^{g}\Delta \mathbf{S_i}(\mathbf{s_i})^H\Delta \mathbf{S_i}(\mathbf{s_i})\right)$. If $\mathrm{det}\left(\Delta \mathbf{S_i}(\mathbf{s_i})^H\Delta \mathbf{S_i}(\mathbf{s_i})\right)>0, \forall i=1,\dots,g$ then, we have
$$
\mathrm{det}\left(\Delta \mathbf{S}(\mathbf{s})^H\Delta \mathbf{S}(\mathbf{s})\right)\geq\sum_{i=1}^{g}\mathrm{det}\left(\Delta \mathbf{S_i}(\mathbf{s_i})^H\Delta \mathbf{S_i}(\mathbf{s_i})\right).
$$
Thus it is sufficient to construct full diversity lattice constellations independently for each of the constituent LSTDs, i.e., $\mathbf{S_i}(\mathbf{s_i})$'s and this will ensure $g$-group ML decodability. Note that\\
$\mathbf{S_i}(\mathbf{s_i})=\mathbf{A_{(i-1)\lambda+1}}\left(x_{(i-1)\lambda+1}\mathbf{I}+x_{(i-1)\lambda+2}\mathbf{A_2}+\dots+x_{i\lambda}\mathbf{A_\lambda}\right)=\mathbf{A_{(i-1)\lambda+1}}\mathbf{S_1}(\mathbf{s_1})$ which implies\\ \mbox{$\mathrm{det}\left(\Delta \mathbf{S_i}(\mathbf{s_i})^H\Delta \mathbf{S_i}(\mathbf{s_i})\right)=\mathrm{det}\left(\Delta \mathbf{S_1}(\mathbf{s_1})^H\Delta \mathbf{S_1}(\mathbf{s_1})\right)$}. Hence without loss of generality, we can consider the construction of full diversity lattice constellations for the LSTD $\mathbf{S_1}(\mathbf{s_1})$ since the same lattice constellation will ensure full diversity for the remaining constituent LSTDs $\mathbf{S_i}(\mathbf{s_i}),~ i=2,\dots,g$. We have $\mathbf{S_1}(\mathbf{s_1})=\sum_{i=1}^{\lambda}x_i\mathbf{A_i}$. Note that the matrices $\mathbf{A_i},~i=1,\dots,\lambda$ are unitary, square to $\mathbf{I}$ and pairwise commute among themselves. Hence they are simultaneously diagonalizable by some unitary matrix $\mathbf{U}$ to result in diagonal unitary matrices $\mathbf{D_1},\mathbf{D_2},\dots,\mathbf{D_\lambda}$. All these diagonal matrices will continue to be linearly independent over $\mathbb{R}$ and since all of them square to $\mathbf{I}$, the diagonal entries of $\mathbf{D_i},~i=1,\dots,\lambda$ are $\pm 1$ and $\mathbf{D_1}=\mathbf{I}$. Thus the LSTD $\mathbf{U}\mathbf{S_1}(\mathbf{s_1})\mathbf{U}^H=\sum_{i=1}^{\lambda}x_i\mathbf{D_i}$ becomes a diagonal matrix for which it is easy to compute the determinant and also find the lattice constellation that will provide full diversity. This procedure is illustrated in the following example.

\begin{eg}
Consider a CUWD for $N_T=8$, $\lambda=4$, $g=4$ and $K=\lambda g=16$ given by $\mathbf{S}(\mathbf{s})=\sum_{i=1}^{16}\mathbf{A_i}$ where,
$\mathbf{A_1}=\mathbf{I_{8}}$,
$\mathbf{A_2}=\left[\begin{array}{cccc}
\mathbf{0_2} & \mathbf{I_2} & \mathbf{0_2} & \mathbf{0_2}\\
\mathbf{I_2} & \mathbf{0_2} & \mathbf{0_2} & \mathbf{0_2}\\
\mathbf{0_2} & \mathbf{0_2} & \mathbf{0_2} & \mathbf{I_2}\\
\mathbf{0_2} & \mathbf{0_2} & \mathbf{I_2} & \mathbf{0_2}
 \end{array}\right]$, $\mathbf{A_3}=\left[\begin{array}{cccc}
\mathbf{0_2} & \mathbf{0_2} & \mathbf{I_2} & \mathbf{0_2}\\
\mathbf{0_2} & \mathbf{0_2} & \mathbf{0_2} & \mathbf{I_2}\\
\mathbf{I_2} & \mathbf{0_2} & \mathbf{0_2} & \mathbf{0_2}\\
\mathbf{0_2} & \mathbf{I_2} & \mathbf{0_2} & \mathbf{0_2}
\end{array}\right]$,\\  $\mathbf{A_4}=\mathbf{A_2}\mathbf{A_3}$,
$\mathbf{A_5}=\left[\begin{array}{cccccccc}
0 & i & 0 & 0 & 0 & 0 & 0 & 0\\
i & 0 & 0 & 0 & 0 & 0 & 0 & 0\\
0 & 0 & 0 & i & 0 & 0 & 0 & 0\\
0 & 0 & i & 0 & 0 & 0 & 0 & 0\\
0 & 0 & 0 & 0 & 0 & i & 0 & 0\\
0 & 0 & 0 & 0 & i & 0 & 0 & 0\\
0 & 0 & 0 & 0 & 0 & 0 & 0 & i\\
0 & 0 & 0 & 0 & 0 & 0 & i & 0
 \end{array}\right]$, $\mathbf{A_9}=\left[\begin{array}{rrrrrrrr}
0 & 1 & 0 & 0 & 0 & 0 & 0 & 0\\
-1 & 0 & 0 & 0 & 0 & 0 & 0 & 0\\
0 & 0 & 0 & 1 & 0 & 0 & 0 & 0\\
0 & 0 & -1 & 0 & 0 & 0 & 0 & 0\\
0 & 0 & 0 & 0 & 0 & 1 & 0 & 0\\
0 & 0 & 0 & 0 & -1 & 0 & 0 & 0\\
0 & 0 & 0 & 0 & 0 & 0 & 0 & 1\\
0 & 0 & 0 & 0 & 0 & 0 & -1 & 0
 \end{array}\right]$,\\ $\mathbf{A_{13}}=\left[\begin{array}{rrrrrrrr}
i & 0 & 0 & 0 & 0 & 0 & 0 & 0\\
0 & -i & 0 & 0 & 0 & 0 & 0 & 0\\
0 & 0 & i & 0 & 0 & 0 & 0 & 0\\
0 & 0 & 0 & -i & 0 & 0 & 0 & 0\\
0 & 0 & 0 & 0 & i & 0 & 0 & 0\\
0 & 0 & 0 & 0 & 0 & -i & 0 & 0\\
0 & 0 & 0 & 0 & 0 & 0 & i & 0\\
0 & 0 & 0 & 0 & 0 & 0 & 0 & -i
 \end{array}\right]$ and $\mathbf{A_{4i+j}}=\mathbf{A_{4i+1}}\mathbf{A_j},~i=1,2,3,~j=1,\dots,4$.\\

It can be checked that the above listed basis matrices satisfy all the requirements of a CUWD for $\lambda=4$, $g=4$. For the purpose of finding full diversity lattice constellations, it is enough to construct full diversity signal sets for the LSTD $\mathbf{S_1}(\mathbf{s_1})=\sum_{i=1}^{4}\mathbf{A_i}$. Since the matrices $\mathbf{A_i},~i=1,\dots,4$ mutually commute among themselves and square to $\mathbf{I_{8}}$, they can be simultaneously diagonalized by a unitary matrix $\mathbf{U}=\left[\begin{array}{rrrrrrrr}
1 & 1 & 1 & 1 & 1 & 1 & 1 & 1\\
1 & -1 & 1 & -1 & 1 & -1 & 1 & -1\\
1 & 1 & -1 & -1 & 1 & 1 & -1 & -1\\
1 & -1 & -1 & 1 & 1 & -1 & -1 & 1\\
1 & 1 & 1 & 1 & -1 & -1 & -1 & -1\\
1 & -1 & 1 & -1 & -1 & 1 & -1 & 1\\
1 & 1 & -1 & -1 & -1 & -1 & 1 & 1\\
1 & -1 & -1 & 1 & -1 & 1 & 1 & -1
 \end{array}\right]$ which in this case turns out to be the $8\times 8$ Hadamard matrix. Defining $\mathbf{D_i}=\mathbf{U}\mathbf{A_i}\mathbf{U}^H,~i=1,\dots,4$, we get $\mathbf{D_1}=\mathbf{I_8}$,\\ $\mathbf{D_2}=\mathrm{diag}\left\{\left[\begin{array}{rrrrrrrr}1 & 1 & -1 & -1 & 1 & 1 & -1 & -1\end{array}\right]\right\}$, $\mathbf{D_3}=\mathrm{diag}\left\{\left[\begin{array}{rrrrrrrr}1 & 1 & 1 & 1 & -1 & -1 & -1 & -1\end{array}\right]\right\}$ and $\mathbf{D_4}=\mathrm{diag}\left\{\left[\begin{array}{rrrrrrrr}1 & 1 & -1 & -1 & -1 & -1 & 1 & 1  \end{array}\right]\right\}$. Thus we have
$$
\mathrm{det}(\Delta \mathbf{S_1})=(\Delta q_1)^2(\Delta q_2)^2(\Delta q_3)^2(\Delta q_4)^2
$$
\noindent where, $\left[\begin{array}{cccc}\Delta q_1 & \Delta q_2 & \Delta q_3 & \Delta q_4\end{array}\right]^T=\mathbf{P}\left[\begin{array}{cccc}\Delta x_1 & \Delta x_2 & \Delta x_3 & \Delta x_4\end{array}\right]^T$ and $\mathbf{P}=\left[\begin{array}{rrrr}
1 & 1 & 1 & 1\\
1 & -1 & 1 & -1\\
1 & 1 & -1 & -1\\
1 & -1 & -1 & 1
 \end{array}\right]$. Thus full diversity will be achieved if $\Delta q_i \neq 0,~\forall~ i=1,\dots,4$. This can be guaranteed by letting $\mathbf{s_1}=\left[\begin{array}{cccc}x_1 & x_2 & x_3 & x_4\end{array}\right]^T$ take values from $\mathbf{P}^{-1}\mathbf{\mathcal{G}}\mathbb{Z}^4$ where, $\mathbf{\mathcal{G}}$ is the generator matrix of a lattice designed to maximize the product distance \cite{FOV,Vit}.
\end{eg}

\subsubsection{Extended Clifford Algebras}
\label{subsubsec2_2_2}
~\\
Towards constructing and studying CUWDs an algebraic framework of extended Clifford algebras is first established. Using this algebraic framework, the optimal tradeoff between rate and ML decoding complexity of CUWDs is obtained in Subsection \ref{subsec2_3}. Furthermore, algebraic descriptions for the ABBA construction \cite{TBH} and the tensor product based construction in \cite{KaR3} are provided using extended Clifford algebras.

First observe that in order to construct CUWDs it is sufficient to construct the weight matrices in the first row and first column (as discussed in Subsection \ref{subsec2_2} and Table \ref{table_CUWD}). Our methodology to construct the weight matrices in the first row and first column would be to fabricate an algebra in such a way that it contains elements satisfying the algebraic relations we need. Once we construct such an algebra, we then obtain the required CUWD by taking an appropriate matrix representation of the constructed algebra. Recall that an algebra over a field is simply a ring as well as a vector space with the addition operation being compatible to both the ring and the vector space structures. Let us recall certain basic definitions from algebra which will be useful in the sequel.

\begin{defn}
A nonempty set $\mathcal{B}$ equipped with two binary operations called addition and multiplication denoted by $+$ and $.$ is called a ring denoted by $(\mathcal{B},+,.)$ if
\begin{enumerate}
\item $(\mathcal{B},+)$ is a Abelian group
\item $(\mathcal{B},.)$ is a monoid with multiplicative identity $1$
\item $x.(y+z)=x.y+x.z,~\forall~x,y,z\in\mathcal{B}$
\item $(x+y).z=x.z+y.z,~\forall~x,y,z\in\mathcal{B}$
\end{enumerate}
\end{defn}
\begin{defn}
A nonempty set $\mathcal{A}$ equipped with two binary operations called addition and multiplication denoted by $+$ and $.$ is called a right module algebra over a ring $\cal{B}$ if
\begin{enumerate}
\item $(\mathcal{A},+,.)$ is a ring
\item There is a map $(x,\alpha)\to x\alpha$ of $\mathcal{A}\times\mathcal{B}$ into $\mathcal{A}$ satisfying the following for all $\alpha,\beta\in\mathbb{B}$ and $x,y\in\mathcal{A}$.
\begin{equation}
\begin{array}{c}
(x+y)\alpha=x\alpha+y\alpha\\
x(\alpha+\beta)=x\alpha+x\beta\\
x(\alpha\beta)=(x\alpha)\beta\\
x1=x
\end{array}
\end{equation}
\end{enumerate}
\end{defn}

Note that in the standard mathematical literature (for example \cite{Jac}), algebra is usually defined over a field. Since our definition differs from the definition in \cite{Jac}, we have given the name `right module algebra' in order to distinguish it from the concept of algebra over a field.

\begin{defn}\cite{TiH1}
The Clifford algebra, denoted by $Cliff_n$ is the algebra over $\mathbb{R}$ generated by $n$ objects $\gamma_k,\ k=1,\dots,n$ which are anti-commuting ($\gamma_k\gamma_j=-\gamma_j\gamma_k,\ \forall k\neq j$) and squaring to $-1$ ($\gamma_k^2=-1\ \forall k=1,\dots,n$).
\end{defn}
A natural basis for $Cliff_n$ seen as a vector space over $\mathbb{R}$~ is
\begin{equation}
\mathscr{B}_n=\left\{1\right\}\bigcup\left\{\gamma_i|i=1,\dots,n\right\} \bigcup_{m=2}^{n}\left\{\prod_{i=1}^{m}\gamma_{k_i}|1\leq k_i\leq k_{i+1}\leq n\right\}.
\end{equation}
The number of basis elements is $|\mathscr{B}_n|=2^n$.

\begin{eg}
$Cliff_0$ is nothing but the set of real numbers $\mathbb{R}$, $Cliff_1$ is the set of complex numbers $\mathbb{C}$~ and $Cliff_2$ is the Hamiltonian Quaternions denoted by $\mathbb{H}$.
\end{eg}

The reason we are interested in Clifford algebras is that the defining algebraic relations of the generators of a Clifford algebra resemble the algebraic relations which the matrices in the first row of a CUWD need to satisfy. Hence we can obtain the matrices in the first row by taking unitary matrix representations of the generators of a Clifford algebra. To obtain the matrices in the first column, we use a similar strategy. We introduce few new symbols in the Clifford algebra and define them to square to $1$, commute with the generators of the Clifford algebra and also commute among themselves. In other words, after introducing new symbols, multiplication in the algebra is appropriately defined in order to create a bigger algebra which contains Clifford algebra as a sub-algebra. Hence by taking a unitary matrix representation of these specific elements of the algebra, we get the weight matrices of the required CUWD. We give the name 'extended Clifford algebras' to the so constructed algebras:

\begin{defn}
\label{defn_eca}
Let $L=2^a,a\in \mathbb{N}$. An extended Clifford algebra denoted by $\mathbb{A}_n^L$ is the associative algebra over $\mathbb{R}$~ generated by $n+a$ objects $\gamma_k,\ k=1,\dots,n$ and $\delta_i,\ i=1,\dots,a$ which satisfy the following relations:
\begin{itemize}
\item $\gamma_k^2=-1,\ \forall\ k=1,\dots,n$
\item $\gamma_k\gamma_j=-\gamma_j\gamma_k,\ \forall\ k\neq j$
\item $\delta_k^2=1,\ \forall k=1,\dots,a$
\item $\delta_k\delta_j=\delta_j\delta_k,\ \forall\ 1\leq k,j\leq a$
\item $\delta_k\gamma_j=\gamma_j\delta_k,\ \forall\ 1\leq k\leq a, 1\leq j\leq n$
\end{itemize}
\end{defn}

From the above definition, it is clear that $Cliff_n$ (or $\mathbb{A}_n^1$) is a sub-algebra of $\mathbb{A}_n^L$. Let $\mathscr{B}_n$ be the natural $\mathbb{R}$~ basis for this sub-algebra $Cliff_n$. Then a natural $\mathbb{R}$~ basis for $\mathbb{A}_n^L$ is given by
\begin{equation}
\mathscr{B}_n^L=\mathscr{B}_n\cup\left\{\mathscr{B}_n\delta_i|i=1,\dots,a\right\} \bigcup_{m=2}^{a}\mathscr{B}_n\left\{\prod_{i=1}^{m}\delta_{k_i}|1\leq k_i\leq k_{i+1}\leq a\right\}.
\end{equation}
Thus the dimension of $\mathbb{A}_n^L$ seen as a vector space over $\mathbb{R}$ is $2^{n+a}$.

The algebra $\mathbb{A}_n^L$ over $\mathbb{R}$ can also be viewed as a right module algebra over the base ring $Cliff_n$. We will use this fact later in subsection \ref{subsec2_3}.

\begin{eg}
\label{eg_eca}
Let us take $n=2$, $a=1$. Hence $L=2$. Then
$$
\mathbb{A}_2^2=\left\{a_1+\gamma_1a_2+\delta_1a_3+\delta_1\gamma_1a_4|a_1,a_2,a_3,a_4\in\mathbb{R}\right\}.
$$

Addition in the algebra is defined to be component wise and multiplication is completely described by defining the multiplication between any two basis elements. The multiplication table can be easily generated using the defining algebraic relations of the generators and is given as follows.

\begin{center}
\begin{tabular}{|c|c|c|c|c|}
\hline
~ & $1$ & $\gamma_1$ & $\delta_1$ & $\delta_1\gamma_1$\\
\hline
$1$ & $1$ & $\gamma_1$ & $\delta_1$ & $\delta_1\gamma_1$\\
\hline
$\gamma_1$ & $\gamma_1$ & $-1$ & $\delta_1\gamma_1$ & $-\delta_1$\\
\hline
$\delta_1$ & $\delta_1$ & $\delta_1\gamma_1$ & $1$ & $\gamma_1$\\
\hline
$\delta_1\gamma_1$ & $\delta_1\gamma_1$ & $-\delta_1$ & $\gamma_1$ & $-1$\\
\hline
\end{tabular}
\end{center}

One can check from the multiplication table that the multiplication is indeed associative. Note that $\mathbb{A}_1^2$ can also be viewed as a vector space over $\mathbb{C}$ by viewing the symbol $\gamma_1$ as the complex number $i=\sqrt{-1}$. Then, we have $\mathbb{A}_1^2=\left\{z_1+\delta_1z_2|z_1,z_2\in\mathbb{C}\right\}$ where, $z_1=a_1+\gamma_1a_2$ and $z_2=a_3+\gamma_1a_4$.
\end{eg}

From the defining relations of the generators of the extended Clifford Algebra, it can be observed that the symbols $1$, $\gamma_1$, $\gamma_2$, $\dots$, $\gamma_n$ satisfy relations similar to that satisfied by the weight matrices that we need in the first row (squaring to $-1$ and anticommuting). Similarly, the symbols $\delta_k,k=1,\dots,a$, and $\bigcup_{m=2}^{a}\prod_{i=1}^{m}\delta_{k_i}$ for $1\leq k_i\leq k_{i+1}\leq a$ satisfy relations similar to that satisfied by the weight matrices that we need in the first column (squaring to $1$ and commuting with all other elements). Thus, for the case of $\lambda=2^a$, when the weight matrices of any CUWD are listed down in the array form as shown in Table \ref{table_CUWD}, the matrices in the first row will simply be matrix representations of the symbols $1$, $\gamma_1$, $\gamma_2$, $\dots$, $\gamma_n$ of an extended Clifford Algebra. Similarly, the matrices in the first column are nothing but matrix representation of the symbols $\delta_k,k=1,\dots,a$, and $\bigcup_{m=2}^{a}\prod_{i=1}^{m}\delta_{k_i}$ for $1\leq k_i\leq k_{i+1}\leq a$ of an extended Clifford Algebra.

\subsection{Optimal Rate-ML decoding complexity tradeoff of Clifford Unitary Weight codes}
\label{subsec2_3}

The maximum rate problem of CUWDs can be formally stated in many equivalent ways. Some of them are listed as follows.
\begin{enumerate}
\item Given $\lambda$ and $N_t$ what is the maximum rate?
\item Given $g$ and $N_t$ what is the maximum rate?
\item \label{question} Given $g$ and $\lambda$ what is the minimum value of $N_T$?
\end{enumerate}

For $\lambda=2$, the solution to the first question is reported in \cite{KaR1}. In this subsection, the solution to question number \ref{question}) for $\lambda=2^a,a\in\mathbb{N}$ is provided. Using the algebraic framework of extended Clifford algebras introduced in the previous subsection, the maximum rate problem can be restated in algebraic terms as follows.

\textit{What is the minimum matrix size $N_T$ in which the algebra $\mathbb{A}_{(g-1)}^\lambda$ has a non-trivial matrix representation?}

This problem appears to be difficult to solve directly. Hence, we take an alternate approach which is similar to the approach in \cite{TiH1} wherein matrix representations of Clifford algebras were obtained using matrix representations of the Clifford group. First, we find a finite group with respect to multiplication in the algebra $\mathbb{A}_{(g-1)}^{\lambda}$ such that it contains the elements of the natural $\mathbb{R}$-basis of $\mathbb{A}_{(g-1)}^{\lambda}$ denoted by $\mathscr{B}_{(g-1)}^{\lambda}$. Then, we find a suitable representation of this finite group such that it can be extended to a representation of the algebra.

\begin{prop}
\label{prop_group}

The set of elements $G=\mathscr{B}_{(g-1)}^{\lambda}\cup\left\{-b|b\in\mathscr{B}_{(g-1)}^{\lambda} \right\}$ is a finite group with respect to multiplication in $\mathbb{A}_{(g-1)}^{\lambda}$. Further, the group $G$ is a direct product of its subgroups $G_{\gamma}$ and $G_{\delta}$, where
\begin{equation}
\begin{array}{rcl}
G_{\gamma}&=&\mathscr{B}_{(g-1)}\cup\left\{-b|b\in\mathscr{B}_{(g-1)}\right\},\\
G_{\delta}&=&G_{\delta_1}\times G_{\delta_2}\times\dots\times G_{\delta_a}
\end{array}
\end{equation}
\noindent and $G_{delta_i}=\left\{1,\delta_i\right\},~i=1,\dots,a$.
\end{prop}
\begin{proof}
The multiplication in $G$ is associative and the unit is $1$. The inverse of the element $\pm\prod_{i=1}^{m}\gamma_{k_i}$ is $\pm(-1)^{\lceil\frac{m}{2}\rceil}\prod_{i=1}^{m}\gamma_{k_i}$. The inverse of the element $\prod_{i=1}^{m}\delta_{k_i}$ is itself. Similarly, it is easy to find the inverse of the other elements. The set $G_{\gamma}$ is nothing but the well known Clifford group \cite{TiH1}. The set $G_{\delta_i}$ is the cyclic group of order two (denoted by $C_2$) with generator $\delta_i$.  The set $G_\delta$ is a group since it is the $a$ times direct product of $C_2$. The group $G$ is a direct product of $G_{\gamma}$ and $G_{\delta}$ because:
\begin{enumerate}
\item Each $s\in G$ can be written uniquely in the form $s=s_1s_2$ with $s_1\in G_{\gamma}$ and $s_2\in G_{\delta}$.
\item For all $s_1\in G_{\gamma}$ and $s_2\in G_{\delta}$, we have $s_1s_2=s_2s_1$.
\end{enumerate}
\end{proof}

Thus the problem is simplified to finding the matrix representations of this finite group $G$. Towards that end, we quickly recall some basic concepts in linear representation of finite groups. We refer the readers to \cite{Serre} for a formal introduction.

\begin{defn}\cite{Serre}
Let $G$ be a finite group with identity element $1$ and let $V$ be a finite dimensional vector space over $\mathbb{C}$. A linear representation of $G$ in $V$ is a group homomorphism $\rho$ from $G$ into the group $GL(V)$. The dimension of $V$ is called the degree of the representation.
\end{defn}

Few basic results in representation theory are as listed below.

\begin{enumerate}
\item[[R1]]: Irreducible representations are representations with no invariant subspaces.
\item[[R2]]: Every representation is a direct sum of irreducible representations. They are equivalent to block-diagonal representations, with irreducible representation matrices on the block diagonal.
\item[[R3]]: Two representations $R$ and $R'$ of $G$ are equivalent, if there exists a similarity transform $U$ so that
$$
R'(x)=U^{-1}R(x)U,~\forall~x\in G
$$
\item[[R4]]: Unitary representations are representations in terms of unitary matrices
\item[[R5]]: Every representation is equivalent to a unitary representation
\end{enumerate}

\begin{thm}\cite{Serre}
\label{thm_abelian}
All the irreducible representations of an Abelian group have degree $1$.
\end{thm}

\begin{lem}\cite{Serre}
\label{lem_tensor}
Let $\rho_1:G_1\to GL(V_1)$ and $\rho_2:G_2\to GL(V_2)$ be linear representations of groups $G_1$ and $G_2$ in vector spaces $V_1$ and $V_2$ respectively. Then $\rho_1\otimes\rho_2$ is a linear representation of $G_1\times G_2$ into $V_1\otimes V_2$.
\end{lem}

\begin{thm}\cite{Serre}
\label{thm_tensor}
\begin{enumerate}
\item If $\rho_1$ and $\rho_2$ are irreducible, then $\rho_1\otimes\rho_2$ is an irreducible representation of $G_1\times G_2$.
\item Each irreducible representation of $G_1\times G_2$ is equivalent to a representation $\rho_1\otimes\rho_2$, where $\rho_i$ is an irreducible representation of $G_i,~i=1,2$.
\end{enumerate}
\end{thm}

Now, having introduced the necessary tools, the problem is to find unitary matrix representations of the finite group $G$. Before we proceed, note that when $G$ is interpreted as a finite group, the representation of $-1$ does not necessarily have anything to do with $-1$ times identity matrix and similarly for a generic $-b,b\in\mathscr{B}_{(g-1)}^{\lambda}$. Such a representation $\rho$, where $\rho(-1)\neq-\rho(1)$ is said to be a degenerate representation. Degenerate representations are not representations of the algebra $\mathbb{A}_{(g-1)}^{\lambda}$. Thus we are interested in the smallest degree non-degenerate unitary representation $\rho$ of the finite group $G$ such that the representation matrices of the required elements of $G$ are linearly independent over $\mathbb{R}$. The following lemma will help to prove the linear independence of complex matrices over $\mathbb{R}$.

\begin{lem}
\label{lem_lid}
A set of complex matrices $\mathbf{A_i}, i=1,\dots,K \in\mathbb{C}^{T\times N_T}$ are linearly independent over $\mathbb{R}$ if $\mathrm{Tr}(\mathbf{A_i}^H\mathbf{A_j}+\mathbf{A_j}^H\mathbf{A_i})=0,~\forall~i\neq j$.
\end{lem}

\begin{proof}
The linear map from $\mathbb{C}^{T\times N_T}\times\mathbb{C}^{T\times N_T}\mapsto\mathbb{R}$ given by $\frac{1}{2}\mathrm{Tr}(\mathbf{A}^H\mathbf{B}+\mathbf{B}^H\mathbf{A})$ for $\mathbf{A},\mathbf{B}\in\mathbb{C}^{T\times N_T}$ is an inner product. The statement of the lemma then follows.
\end{proof}

\begin{thm}
\label{thm_maxrate}
The maximum rate of a CUWD for $\lambda=2^a,a\in\mathbb{N}$ and arbitrary $g$ is equal to $\frac{g}{2^{\lfloor\frac{(g-1)}{2}\rfloor}}$ dpcu.
\end{thm}

\begin{proof}
Proof is by induction on $a$. The proof proceeds to find the smallest degree non-degenerate unitary representation $\rho$ of $G$ such that the following condition is satisfied.

\begin{equation}
\label{eqn_cond_rep}
\rho(x)\neq\pm\rho(y),~\forall~x\neq y\in G
\end{equation}

The above condition is required, since otherwise, the representation matrices will be linearly dependent over $\mathbb{R}$. However, even if the above condition is satisfied, linear independence is still not guaranteed. Therefore, we can only obtain an upper bound on the rate but we shall see that a representation meeting the upper bound actually provides us with linear independence as well.

For $a=0$, CUWDs become ODs and the maximum rate for square ODs is well known \cite{TiH1} and the theorem holds true. For $a=1$, $\lambda=2$ and the group $G=G_\gamma\times G_{\delta_1}$ where $G_{\delta_1}=\left\{1,\delta_1 \right\}$. Since we are interested in the smallest degree representation of $G$, let us first study the irreducible representations of $G$. From Theorem \ref{thm_tensor}, all irreducible representations of $G$ are obtained as a tensor product of irreducible representations of $G_{\gamma}$ and $G_{\delta_1}$. All irreducible representations of $G_{\gamma}$ have been studied in \cite{TiH1}. There are $2$ non-degenerate irreducible representations of $G_\gamma$ in dimension $2^{\lfloor\frac{g-1}{2}\rfloor}$. The representation matrices of the $(g-1)$ generators of $G_\gamma$ are given as follows \cite{TiH1}:

\begin{equation*}
\begin{array}{rclcc}
R(\gamma_{2k})&=&\underbrace{\mathbf{I_2}\otimes\dots\otimes\mathbf{I_2}}\otimes\mathbf{\sigma_1}\otimes&\underbrace{\mathbf{\sigma_3}\otimes\dots\otimes\mathbf{\sigma_3}},& k=1,2,\dots,(K-1)\\
&&~~K-1-k&k-1\\
R(\gamma_{2k})&=&\underbrace{\mathbf{I_2}\otimes\dots\otimes\mathbf{I_2}}\otimes\mathbf{\sigma_2}\otimes&\underbrace{\mathbf{\sigma_3}\otimes\dots\otimes\mathbf{\sigma_3}},&k=1,2,\dots,(K-1)\\
&&~~K-1-k&k-1\\
R(\gamma_{1})&=&\pm i \underbrace{\mathbf{\sigma_3}\otimes\dots\otimes\mathbf{\sigma_3}}\\
&&~~~~~~~~K-1&
\end{array}
\end{equation*}

\noindent where, $\mathbf{\sigma_1}=\left[\begin{array}{cc}0 & 1\\-1 & 0\end{array}\right]$, $\mathbf{\sigma_2}=\left[\begin{array}{cc}0 & i\\ i & 0 \end{array}\right]$, $\mathbf{\sigma_3}=i\mathbf{\sigma_1}\mathbf{\sigma_2}=\left[\begin{array}{cc}1 & 0\\0 & -1 \end{array}\right]$ and $K=\left\{\begin{array}{cc} \frac{g}{2} & \mathrm{if}~g~\mathrm{is}~\mathrm{even}\\ \frac{g+1}{2} & \mathrm{if}~g~\mathrm{is}~\mathrm{odd}  \end{array}\right.$. The notation $R(.)$ is used to denote the representation matrix. Also note that $\mathrm{Tr}(\mathbf{\sigma_i})=0,~i=1,2,3$.

Since the two non-degenerate representations are in the same dimension, without loss of generality, let us consider one of them and denote it by $\rho_0$. By Theorem \ref{thm_abelian}, all the irreducible representations of $G_{\delta_1}$ are in dimension $1$ since the group $G_{\delta_1}$ is Abelian. Recall that $G_{\delta_1}$ is nothing but the cyclic group $C_2$ of order two. Apart from the trivial representation (all elements are mapped to $1$), the only other irreducible representation of the order two cyclic group $G_{\delta_1}$ is given by: $R(1)=1$, $R(\delta_1)=-1$. Note that $(-1)^2=1$ and hence $-1$ is the generator. Thus we get two non-degenerate irreducible representations of $G$ in dimension $2^{\lfloor\frac{g-1}{2}\rfloor}$ denoted by $R_1$ and $R_2$ respectively and they are given by:
\begin{enumerate}
\item $R_1(\gamma_i)=\rho_0(\gamma_i), i=1,\dots,(g-1)$, $R_1(\delta_1)=\mathbf{I_m}$, $R_1(\delta_1\gamma_i)=\rho_0(\gamma_i), i=1,\dots,(g-1)$
\item $R_2(\gamma_i)=\rho_0(\gamma_i), i=1,\dots,(g-1)$, $R_2(\delta_1)=-\mathbf{I_m}$, $R_2(\delta_1\gamma_i)=-\rho_0(\gamma_i), i=1,\dots,(g-1)$
\end{enumerate}

\noindent where, $m=2^{\lfloor\frac{g-1}{2}\rfloor}$. But both the non-degenerate irreducible representations of $G$ fail to satisfy condition \eqref{eqn_cond_rep}. Thus we seek non-degenerate reducible representations of $G$ that satisfy \eqref{eqn_cond_rep}. From property [R2], we have that reducible representations can be easily obtained by placing irreducible representations as blocks on the diagonal. If degenerate irreducible representations are placed as blocks on the diagonal then it is easy to check that the resulting representation will also be degenerate. Thus only non-degenerate irreducible representations can be placed as blocks on the diagonal to construct non-degenerate reducible representations of $G$. It then follows that the smallest degree non-degenerate representation $\rho_1$ satisfying \eqref{eqn_cond_rep} for $a=1$ is $2(2^{\lfloor\frac{g-1}{2}\rfloor})$ and the corresponding basis matrices we need are explicitly given as follows:

$\mathbf{A_1}=\mathbf{I}_{2m}$, $\mathbf{A_{2i+1}}=\left[\begin{array}{cc}\rho_0(\gamma_{i}) & \mathbf{0}\\ \mathbf{0} & \rho_0(\gamma_{i})\end{array}\right], i=1,2,\dots,(g-1)$, $\mathbf{A_{2}}=\left[\begin{array}{cc}\mathbf{I_m} & \mathbf{0}\\ \mathbf{0} & -\mathbf{I_m} \end{array}\right]$.

Now using the identity $\mathrm{Tr}(\mathbf{A}\otimes\mathbf{B})=\mathrm{Tr}(\mathbf{A})\times\mathrm{Tr}(\mathbf{B})$, it can be easily checked that the above basis matrices are trace orthogonal, i.e.,  $\mathrm{Tr}(\mathbf{A}_i^H\mathbf{A}_j+\mathbf{A}_j^H\mathbf{A}_i)=0,~\forall~i\neq j$ and hence by Lemma \ref{lem_lid} they are linearly independent over $\mathbb{R}$. Thus the theorem is true for $a=1.$  Now let us assume that the theorem is true for $a=n-1$ and prove the theorem for $a=n$.

For the case of $a=n$, note that the corresponding $G$ can be expressed as $G=G_{n-1}\times G_{\delta_n}$ where,\\ $G_{n-1}=G_{\gamma}\times G_{\delta_1}\times G_{\delta_2}\times\dots\times G_{\delta_{n-1}}$ and $G_{\delta_i}=\left\{1,\delta_i \right\}, i=1,\dots,n$. Once again invoking Theorem \ref{thm_tensor}, we have that the irreducible representations of $G$ are a tensor product of irreducible representations of $G_{n-1}$ and $G_{\delta_n}$. Now using Theorem \ref{thm_abelian}, the non-degenerate irreducible representations of $G$ are in dimension $2^{\lfloor\frac{(g-1)}{2}\rfloor}$. Since they do not satisfy \eqref{eqn_cond_rep}, we look for non-degenerate reducible representations whose degree has to be a multiple of $2^{\lfloor\frac{g-1}{2}\rfloor}$. By induction hypothesis, the smallest degree non-degenerate representation which results in linearly independent basis matrices for $a=n-1$ is $2^{n-1}(2^{\lfloor\frac{g-1}{2}\rfloor})$. Let it be denoted by $\rho_{n-1}$. Since the representation $\rho_{n-1}$ is also a representation of $G_{n-1}$, using analogous arguments as made for $a=1$ it follows that the smallest degree non-degenerate representation $\rho_n$ satisfying \eqref{eqn_cond_rep} for $a=n$ is in dimension $2^n(2^{\lfloor\frac{(g-1)}{2}\rfloor})$ and the corresponding basis matrices are given by:
$A_1=\left[\begin{array}{cc}\rho_{n-1}(1) & \mathbf{0}\\ \mathbf{0} & \rho_{n-1}(1)\end{array}\right]$, $A_{2^ni+1}=\left[\begin{array}{cc}\rho_{n-1}(\gamma_i) & \mathbf{0}\\
\mathbf{0} & \rho_{n-1}(\gamma_i)\end{array}\right], i=1,\dots,(g-1)$,\\ $A_i=\left[\begin{array}{cc}\rho_{n-1}(\delta_{i-1}) & \mathbf{0}\\ \mathbf{0} & \rho_{n-1}(\delta_{i-1})\end{array}\right], i=2,\dots,(n-1)$, $A_n=\left[\begin{array}{cc}\rho_{n-1}(1) & \mathbf{0}\\ \mathbf{0} & -\rho_{n-1}(1)\end{array}\right]$.

Once again it can be shown that the above basis matrices are linearly independent over $\mathbb{R}$ by using Lemma \ref{lem_lid}.
\end{proof}

Theorem \ref{thm_maxrate} essentially answers the question: For a CUWD, given $g$ and $\lambda$, a power of two, what is the minimum matrix size $N_T$ that it can have? The answer to this question is given by $\lambda\left(2^{\lceil\frac{g-1}{2}\rceil}\right)$. The following example highlights the fact that the maximum rate expression of a CUWD given in Theorem \ref{thm_maxrate} does not depend on $\lambda$.

\begin{eg}
\label{eg_rate_lambda}
For $g=4$, let us study CUWDs for two cases $\lambda=1$ and $\lambda=2$.\\
Case 1: $\lambda=1$, $g=4$\\
The minimum possible dimension in which a CUWD with these parameters exists is given by Theorem 6 which is equal to $2$. The corresponding CUWD is nothing but the well known Alamouti LSTD.\\
Case 2: $\lambda=2$, $g=4$\\
The minimum possible dimension in which a CUWD with $\lambda=2$, $g=4$ exists as per Theorem 6 is $4$ and the corresponding CUWD is given by:
$$
\left[\begin{array}{rrrr}
x_1+x_2+i(x_3+x_4) & -x_5-x_6+i(x_7+x_8) & 0 & 0\\
x_5+x_6+i(x_7+x_8) & x_1+x_2-i(x_3+x_4) & 0 & 0\\
0 & 0 & x_1-x_2+i(x_3-x_4) & -x_5+x_6+i(x_7-x_8)\\
0 & 0 & x_5-x_6+i(x_7-x_8) & x_1-x_2-i(x_3-x_4)
 \end{array}\right]
$$
\noindent where, the grouping of real variables are $\left\{x_1,x_2\right\}$, $\left\{x_3,x_4\right\}$, $\left\{x_5,x_6\right\}$ and $\left\{x_7,x_8\right\}$. Note that when $\lambda$ increases from one to two, the minimum dimension in which $K=\lambda g$ matrices with required properties exist, also increases. In fact, Theorem \ref{thm_maxrate} explicitly tells that $N_T$ increases linearly with $\lambda$. This makes the rate, which is $\frac{\lambda g}{N_T}$ independent of $\lambda$.
\end{eg}

\subsubsection{Algebraic description for ABBA construction}
\label{subsubsec2_3_1}
~\\
As stated in Section \ref{sec3} (also see Example \ref{eg_eca}), the algebra $\mathbb{A}_n^{L}$ over $\mathbb{R}$ can also be viewed as a finitely generated right module algebra over $Cliff_n$. A general element $x$ of the algebra $\mathbb{A}_n^L$ can be written as follows:
\begin{equation}
x=c_1+\delta_1c_2+\dots+\delta_ac_{a+1}+\delta_1\delta_2c_{a+2}+\dots+(\prod_{i=1}^{a}\delta_i)c_{L}
\end{equation}
where $c_i,i=1,\dots,L~\in Cliff_n$. There is a natural embedding of $\mathbb{A}_n^L$ into $End_{Cliff_n}(\mathbb{A}_n^L)$ given by left multiplication as shown below:

\begin{equation}
\begin{array}{l}
\phi:\mathbb{A}_n^L\mapsto \mathrm{End}_{Cliff_n}(\mathbb{A}_n^L),\\
\phi(x)=L_x:y\mapsto xy.
\end{array}
\end{equation}

It is easy to check that the map $L_x$ is $Cliff_n$ linear and the map $\phi$ is a ring homomorphism. Also, it can be proved that the map $\phi$ is injective as follows.\\
Let $\phi(x_1)=L_{x_1}$ and $\phi(x_2)=L_{x_2}$. If $\phi(x_1)=\phi(x_2)$, then
$$
\begin{array}{rcl}
L_{x_1}(y)&=&L_{x_2}(y) ~ \forall~ y\\
x_1y&=&x_2y ~ \forall ~y\\
(x_1-x_2)y&=&0 ~ \forall~ y
\end{array}
$$
\noindent which implies $x_1-x_2=0$ or equivalently $x_1=x_2$. Hence, we can represent the algebra $\mathbb{A}_n^L$ by matrices with entries from Clifford algebra. However, we are only interested in matrix representations with entries from the complex field. But this can be easily obtained by simply replacing each Clifford algebra element by its matrix representation over $\mathbb{C}$. This is possible because the matrix representation of $Cliff_n$ over $\mathbb{C}$ is well known and is explicitly given in \cite{TiH1}. The resulting weight matrices are guaranteed to be linearly independent since $\phi$ is injective. We now illustrate this construction with an example.
\begin{eg}
Consider $\mathbb{A}_n^{2^a}$ for $a=2$. Thus $\lambda=4$ , $g=n+1$ and $K=4(n+1)$. A general element $x\in \mathbb{A}_n^4$ can be expressed as follows.
$$
x=c_1+\delta_1c_2+\delta_2c_3+\delta_1\delta_2c_4
$$
where, $c_i,i=1,\dots,4 \in Cliff_n$. Let us obtain a matrix representation over $Cliff_n$ for the map $L_x$. We have,
\begin{equation}
\begin{array}{rcl}
L_x(1)&=&c_1+\delta_1c_2+\delta_2c_3+\delta_1\delta_2c_4\\
L_x(\delta_1)&=&(c_1+\delta_1c_2+\delta_2c_3+\delta_1\delta_2c_4)\delta_1\\
&=&\delta_1c_1+c_2+\delta_1\delta_2c_3+\delta_2c_4\\
L_x(\delta_2)&=&(c_1+\delta_1c_2+\delta_2c_3+\delta_1\delta_2c_4)\delta_2\\
&=&\delta_2c_1+\delta_1\delta_2c_2+c_3+\delta_1c_4\\
L_x(\delta_1\delta_2)&=&(c_1+\delta_1c_2+\delta_2c_3+\delta_1\delta_2c_4)\delta_1\delta_2\\
&=&\delta_1\delta_2c_1+\delta_2c_2+\delta_1c_3+c_4.
\end{array}
\end{equation}
\end{eg}
The map $L_x$ can be represented by the matrix
\begin{equation}
\label{eqn_abba}
\left[\begin{array}{cccc}
c_1 & c_2 & c_3 & c_4\\
c_2 & c_1 & c_4 & c_3\\
c_3 & c_4 & c_1 & c_2\\
c_4 & c_3 & c_2 & c_1
\end{array}\right]
\end{equation}
\noindent  where, $c_1,c_2,c_3,c_4\in Cliff_n$. In order to get a matrix representation over $\mathbb{C}$, we simply replace each $c_i,i=1,\dots,4$ by their matrix representations over $\mathbb{C}$. However, we are interested only in a $4$-group ML decodable LSTD which can be obtained by using the matrix representation of the specific elements $1$, $\gamma_i,~i=1,\dots,n$, $\delta_1$, $\delta_1\gamma_i,~i=1,\dots,n$, $\delta_2$, $\delta_2\gamma_i,~i=1,\dots,n$, $\delta_1\delta_2$, $\delta_1\delta_2\gamma_i,~i=1,\dots,n$ as weight matrices. This is done by restricting the representation of the algebra to the subspace over $\mathbb{R}$ spanned by the required elements of the algebra. In other words, we substitute zero for the coefficients corresponding to the terms not required (terms involving product of Clifford algebra generators like $\gamma_1\gamma_2$ are omitted). To be precise, each $c_i\in Cliff_n$ is restricted to be of the form:
$$
c_i=x_{(i-1)(n+1)+1}+\sum_{j=2}^{n}x_{(i-1)(n+1)+j}\gamma_j
$$
\noindent for some $x_i\in\mathbb{R},~ i=1,\dots,K$ by forcing the coefficients of the remaining terms as zero. In terms of the corresponding matrix representation, this is equivalent to simply replacing $c_i, i=1,\dots,4$ by ODs in \eqref{eqn_abba}. Therefore, the above method results in a $4$-real symbol ML decodable CUW STBC with maximal rate. It turns out that the above construction is precisely the ABBA construction proposed by Tirkkonen et al in \cite{TBH}.

As a consequence of this result, it follows that the $4$ transmit antenna LSTD based on ABBA construction given by $\left[\begin{array}{crcr}
x_1+ix_2 & -x_3+ix_4 & x_5+ix_6 & -x_7+ix_8\\
x_3+ix_4 & x_1-ix_2 & x_7+ix_8 & x_5-ix_6\\
x_5+ix_6 & -x_7+ix_8 & x_1+ix_2 & -x_3+ix_4\\
x_7+ix_8 & x_5-ix_6 & x_3+ix_4 & x_1-ix_2\end{array}\right]$ has to be $2$-real symbol ML decodable. Though the same LSTD was proposed earlier in \cite{TBH}, the authors of \cite{TBH} chose the following pairing of real variables in a group which essentially resulted in a $4$-real symbol ML decodable STBC.
\begin{enumerate}
\item First group $\left\{x_{1},x_{2}\right\}$
\item Second group $\left\{x_{3},x_{4}\right\}$
\item Third group $\left\{x_{5},x_{6}\right\}$
\item Fourth group $\left\{x_{7},x_{8}\right\}$.
\end{enumerate}
However, if we form the following partition of real variables, we can obtain a single complex symbol ML decodable STBC.
\begin{enumerate}
\item First group $\left\{x_{1},x_{5}\right\}$
\item Second group $\left\{x_{2},x_{6}\right\}$
\item Third group $\left\{x_{3},x_{7}\right\}$
\item Fourth group $\left\{x_{4},x_{8}\right\}$.
\end{enumerate}

Thus we see that the ML decoding complexity of STBCs obtained from linear designs can vary dramatically depending on the choice of multidimensional signal set $\mathscr{A}$.

\subsubsection{Algebraic description for Tensor product construction}
\label{subsubsec2_3_2}
~\\
In \cite{KaR2}, a construction of CUW STBCs based on tensor products was provided without giving any reasoning for the mathematical source of such a construction. With the algebraic background that we have now developed, the tensor product construction in \cite{KaR2} can be easily explained. Since the group $G$ is a direct product of $G_\gamma$ and $G_{\delta}$, from Lemma \ref{lem_tensor}, a representation of $G$ can be obtained as a tensor product of a representation of $G_{\gamma}$ and that of $G_{\delta}$. Unitary matrix representations of $G_{\gamma}$ are available in \cite{TiH1}. The unitary  matrices representing $G_{\delta}$ should commute and also square to $I$. Such matrices are simultaneously diagonalizable and their eigen values are equal to $\pm 1$ (squaring to $I$). So a simple method to construct $\lambda$-real symbol decodable STBCs would be to take $\lambda$ linearly independent diagonal matrices of size $\lambda\times\lambda$ having $\pm 1$ as entries and then tensor them with representation matrices of the generators of $G_\gamma$. The construction suggested in \cite{KaR2} is precisely based on this principle. One advantage of the construction in \cite{KaR3} is that it provides CUWDs for all even number of transmit antennas.

\subsubsection{On the maximal rate of non-CUW STBCs}
\label{subsubsec2_3_3}
~\\
It is important to note that Theorem \ref{thm_maxrate} provides the optimal rate-ML decoding complexity tradeoff only within the class of CUW STBCs. The rate of a general $g$-group ML decodable STBC can in fact be more that of a CUWD. An example of a such a LSTD is the recently found high rate quasi-orthogonal STBC in \cite{YGT2}. This LSTD for $4$ transmit antennas which was found by an exhaustive computer search has a rate of $2.5$ dpcu and is $2$-group ML decodable. This solitary example for $4$ transmit antennas shows that there is a lot of room for further work in the direction of increasing transmission rate of $g$-group ML decodable STBCs.

\section{Multi-group ML decodable Distributed STBCs}
\label{sec3}

In this section, multi-group ML decodable distributed space time block codes (DSTBCs) are discussed. In Subsection \ref{subsec3_1}, we present a generalization of the distributed space-time coding strategy proposed in \cite{JiH} and derive a code design criteria for full diversity. In Subsection \ref{subsec3_2}, the necessary and sufficient conditions for multi-group ML decoding of DSTBCs are provided. Three new classes of four group decodable DSTBCs are constructed in Subsection \ref{subsec3_3}.

\subsection{Distributed Space-Time Coding}
\label{subsec3_1}

Consider a network consisting of a source node, a destination node and $R$ relay nodes which aid the source in communicating information to the destination. All the nodes are assumed to be equipped only with a single antenna and are half duplex constrained, i.e., a node cannot transmit and receive simultaneously in the same frequency. The wireless channels between the terminals are assumed to quasi-static and flat fading. The channel fading gains from the source to the $i$-th relay, $f_i$ and from the $i$-th relay to the destination $g_i$ are all assumed to be independent and identically distributed complex Gaussian random variables with zero mean and unit variance. Symbol synchronization and carrier frequency synchronization are assumed among all the nodes. Moreover the destination is assumed to have perfect knowledge of all the channel fading gains.

Every transmission cycle from the source to the destination comprises of two phases-broadcast phase and cooperation phase. In the broadcast phase, the source transmits a $T(T\geq R)$ length vector $\sqrt{\pi_1P}\mathbf{z}$ which the relays receive. Here, $P$ denotes the total average power spent by all the relays and the source. The fraction of total power $P$ spent by the source is denoted by $\pi_1$. The vector $\mathbf{z}$ satisfies $\mathrm{E}[\mathbf{z}^H\mathbf{z}]=T$ and represents the information that the source intends to communicate. The received vector at the $j$-th relay node is then given by $\mathbf{r_j}=\sqrt{\pi_1P}f_j\mathbf{z}+\mathbf{v_j},\ \mathrm{where}\ \mathbf{v_j}\sim\mathcal{CN}(0,I_T)$. During the cooperation phase, all the relay nodes are scheduled to transmit together. The relays are allowed to only linearly process the received vector $\mathbf{r_j}$ or its conjugate $\mathbf{r_j}^*$. To be precise, the $j$-th relay node is equipped with a $T\times T$ matrix $\mathbf{B_j}$ (called relay matrix) satisfying $\parallel\mathbf{B_j}\parallel_F^2=T$ and it transmits  $\mathbf{t_j}=\sqrt{\frac{\pi_2P}{\pi_1P+1}}\mathbf{B_jr_j}$ or $\mathbf{t_j}=\sqrt{\frac{\pi_2P}{\pi_1P+1}}\mathbf{B_jr_j}^*$. Here, $\pi_2$ denotes the fraction of total power $P$ spent by a relay. Without loss of generality, we may assume that the first $M$ relays linearly process $\mathbf{r_j}$ and the remaining $R-M$ relays linearly process $\mathbf{r_j}^*$. If the quasi-static duration of the channel is much greater than $2T$ time slots, then the received vector at the destination is given by
\begin{equation}
\mathbf{y}=\sum_{j=1}^{R}g_j\mathbf{t_j}+\mathbf{w}=\sqrt{\frac{\pi_1\pi_2P^2}{\pi_1P+1}}\mathbf{Xh}+\mathbf{n}
\end{equation}
where,
\begin{eqnarray}
\mathbf{X}&=&\left[\begin{array}{ccccccc}\mathbf{B_1z} & \dots & \mathbf{B_Mz} &
\label{eqn_dstbc}
\mathbf{B_{M+1}z}^* & \dots & \mathbf{B_Rz}^*\end{array}\right]\\
\mathbf{h}&=&\left[\begin{array}{ccccccc}f_1g_1 & f_2g_2 & \dots f_Mg_M & f_{M+1}^*g_{M+1} & \dots & f_R^*g_R\end{array}\right]^T,\\
\mathbf{n}&=&\sqrt{\frac{\pi_2P}{\pi_1P+1}}\left(\sum_{j=1}^{M}g_j\mathbf{B_jv_j}+\sum_{j=M+1}^{R}g_j\mathbf{B_jv_j}^*\right)+w,
\end{eqnarray}
\noindent
and $\mathbf{w}\sim\mathcal{CN}(0,I_T)$ represents the additive noise at the destination. The power allocation factors $\pi_1$ and $\pi_2$ must satisfy $\pi_1PT+R\pi_2PT=P(2T)$. Throughout this paper, we choose $\pi_1=1$ and $\pi_2=\frac{1}{R}$ as suggested in \cite{JiH}. Let $\mathbf{\Gamma}$ denote the covariance matrix of $\mathbf{n}$. We have,
\begin{equation}
\label{eqn_cov_matrix}
\mathbf{\Gamma}=\mathrm{E}[\mathbf{n}\mathbf{n}^H]=\mathbf{I_T}+\frac{\pi_2P}{\pi_1P+1}(\sum_{i=1}^{R}|g_i|^2\mathbf{B_i}\mathbf{B_i}^H).
\end{equation}

The vector $\mathbf{z}$ transmitted by the source is taken from a finite subset of $\mathbb{C}^T$ which then defines a collection of matrices when substituted for in $\mathbf{X}$ as given in \eqref{eqn_dstbc}. This finite set of matrices is called a DSTBC since each column of a codeword matrix is transmitted by geographically distributed relay nodes. The destination node performs ML decoding as follows:
\begin{equation}
\label{eqn_ml_dstbc}
\mathbf{\hat{X}}=\arg\min_{\mathbf{X}\in\mathscr{C}}\parallel\mathbf{\Gamma}^{-\frac{1}{2}}(\mathbf{y}-\sqrt{\frac{\pi_1\pi_2P^2}{\pi_1P+1}}\mathbf{X}\mathbf{h})\parallel_F^2.
\end{equation}

Observe that if the entries of $\mathbf{z}$ are treated as complex variables, then the DSTBC $\mathscr{C}$ can be viewed as being obtained from certain special LSTDs having the form of $\eqref{eqn_dstbc}$. Note that such LSTDs have the property that any column has linear functions of either only the complex variables or only their conjugates respectively. We refer to LSTDs with this property as `conjugate LSTDs'. The following theorem provides sufficient conditions under which the DSTBC $\mathscr{C}$ achieves full cooperative diversity equal to $R$ under ML decoding.

\begin{thm}
\label{thm_codinggain}
Assume that $T\geq R$, $\pi_1=1$ and $\pi_2=\frac{1}{R}$. If $\mathbf{B_i}\mathbf{B_i}^H$ is a diagonal matrix $\forall\ i=1,\ldots,R$ and if $\Delta \mathbf{X}=\mathbf{X_i}-\mathbf{X_j}$ has full rank for all pairs of distinct codewords $\mathbf{X_i},\mathbf{X_j}\in\mathscr{C}$, then the DSTBC $\mathscr{C}$ achieves full cooperative diversity equal to $R$ under ML decoding.
\end{thm}

\begin{proof}
Let $\mathbf{X_i}$ be the transmitted codeword, $\mathbf{X_j}$ be some other codeword and let $\Delta \mathbf{X}=\mathbf{X_i}-\mathbf{X_j}$. We have,

$$
\mathrm{P}(\mathbf{y}|\mathbf{X_i})=\frac{e^{-(\mathbf{y}-\sqrt{\frac{\pi_1\pi_2P^2}{\pi_1P+1}}\mathbf{X}\mathbf{h})^H\mathbf{\Gamma}^{-1}(\mathbf{y}-\sqrt{\frac{\pi_1\pi_2P^2}{\pi_1P+1}}\mathbf{X}\mathbf{h})}}{\pi^T|\Gamma|}.
$$

On applying the Chernoff bound, we can upper bound the pairwise error probability (PEP) that a ML decoder decodes wrongly to $\mathbf{X_j}$ as follows:

\begin{equation}
\label{eqn_chernoff}
\mathrm{PEP}\leq\begin{array}{c}\mathrm{E}\\\left\{f_i\right\}\left\{g_i\right\}\end{array}e^{-\frac{\pi_1\pi_2P^2}{4(\pi_1P+1)}\mathbf{h}^H(\Delta \mathbf{X})^H\mathbf{\Gamma}^{-1}(\Delta \mathbf{X})\mathbf{h}}.
\end{equation}

Since $\mathbf{B_i}\mathbf{B_i}^H$ is a diagonal matrix $\forall i=1,\dots,R$, let $\mu$ denote the maximum of the diagonal entries of $\mathbf{B_i}\mathbf{B_i^H}$ over all $i=1,\dots,R$. Let $\mathbf{D}=\left(1+\frac{\mu \pi_2PR}{\pi_1P+1}\sum_{i=1}^{R}|g_i|^2\right)\mathbf{I_{T}}$. Now by replacing $\mathbf{\Gamma}$ by $\mathbf{D}$, we can further upper bound the PEP expression in \eqref{eqn_chernoff} since this is essentially equivalent to assuming more noise variance at the destination than what is actually present and hence results in an upper bound. Thus, we have

\begin{equation*}
\mathrm{PEP}\leq\begin{array}{c}\mathrm{E}\\\left\{f_i\right\}\left\{g_i\right\}\end{array}e^{-\frac{\pi_1\pi_2P^2}{4(\pi_1P+1)}\mathbf{h}^H(\Delta \mathbf{X})^H\mathbf{D}^{-1}(\Delta \mathbf{X})\mathbf{h}}.
\end{equation*}

On integrating over the $f_i$'s as done in Appendix I of \cite{JiH}, we get

\begin{equation*}
\mathrm{PEP}\leq\begin{array}{c}\mathrm{E}\\ \left\{g_i\right\}\end{array}|\mathbf{I_R}+\frac{\pi_1\pi_2P^2}{4(\pi_1P+1)}(\Delta \mathbf{X})^H\mathbf{D}^{-1}(\Delta \mathbf{X})\mathrm{diag}\left\{|g_1|^2,|g_2|^2,\dots,|g_R|^2\right\}|^{-1}.
\end{equation*}

For the power allocation $\pi_1=1$, $\pi_2=\frac{1}{R}$ and for large $P$, we can approximate the above expression as follows:

\begin{equation}
\mathrm{PEP}\lesssim \begin{array}{c}\mathrm{E}\\ \left\{g_i\right\}\end{array}|\mathbf{I_R}+\frac{P}{4(R+\mu\sum_{i=1}^{R}|g_i|^2)}(\Delta \mathbf{X})^H(\Delta \mathbf{X})\mathrm{diag}\left\{|g_1|^2,|g_2|^2,\dots,|g_R|^2\right\}|^{-1}.
\end{equation}

Now proceeding as in Appendix II of \cite{JiH}, it can be shown that the above expectation can be further upper bounded to result in:

$$
\mathrm{PEP}\lesssim\left(\frac{c|(\Delta \mathbf{X})^H(\Delta \mathbf{X})|^{\frac{1}{R}}}{4R}\right)^{-R}P^{-R\left(1-\frac{\log\log P}{\log P}\right)}
$$

\noindent where, $c$ is some constant independent of $P$. This completes the proof.
\end{proof}

Theorem \ref{thm_codinggain} generalizes the results of \cite{JiH} (wherein only unitary relay matrices were permitted) to allow row orthogonal relay matrices ($\mathbf{B_i}\mathbf{B_i}^H$ is a diagonal matrix). Note that the transmission protocol assumed in this paper does not involve communication using the direct link between the source and the destination. An even more general transmission protocol called as `GNAF protocol' which employs the direct link and also allows a general form of linear processing at the relays along with unequal duration of broadcast phase and cooperation phase is discussed in \cite{RaR1}. Note that if the direct link is also employed, then a maximum diversity of $R+1$ can be achieved \cite{RaR1}. However, for the purposes of this paper, the results of Theorem \ref{thm_codinggain} are sufficient. We shall see in the sequel that relaxing $\mathbf{B_i}$ to row orthogonal matrices paves the way to obtain DSTBCs with low ML decoding complexity. Hence, for constructing DSTBCs we need conjugate LSTDs whose relay matrices have orthogonal rows. This is one of the major differences between collocated STBCs and DSTBCs.

\subsection{Conditions for Multi-group ML decoding of DSTBCs}
\label{subsec3_2}

The ML decoding complexity of DSTBCs becomes an important issue especially when $R$ is large. This provides a good motivation to study multi-group decodable DSTBCs. The following theorem provides necessary and sufficient conditions for multi-group ML decoding of DSTBCs.

\begin{thm}
\label{thm_multigroup_dstbc}
A DSTBC $\mathscr{C}$ is $g$-group decodable if and only if the following two conditions are satisfied.
\begin{enumerate}
\item $\mathscr{C}$ is $g$-group encodable
\item The associated basic matrices $\mathbf{A_i}, i=1,\dots,K$ of $\mathscr{C}$ satisfy:
\begin{equation}
\label{eqn_multigroup_cond_dstbc}
\mathbf{A_i^H}\mathbf{\Gamma}^{-1}\mathbf{A_j}+\mathbf{A_j^H}\mathbf{\Gamma}^{-1}\mathbf{A_i}=\mathbf{0}
\end{equation}
\noindent whenever $\mathbf{A_i}$ and $\mathbf{A_j}$ belong to different groups.
\end{enumerate}
\end{thm}
\begin{proof}
Let $\mathbf{\tilde{y}}=\mathbf{\Gamma}^{-\frac{1}{2}}\mathbf{y}$. Then from \eqref{eqn_ml_dstbc}, the ML decoding metric is given by\\ \mbox{$\parallel \mathbf{\tilde{y}}-\sqrt{\frac{\pi_1\pi_2P^2}{\pi_1P+1}}(\mathbf{\Gamma}^{-\frac{1}{2}}\mathbf{X})\mathbf{h}\parallel_F^2$}. Compared to the collocated case, the difference here is the term involving $\mathbf{\Gamma}^{-\frac{1}{2}}$. The effect of pre-multiplication by $\mathbf{\Gamma}^{-\frac{1}{2}}$ can be captured by considering $\mathbf{\Gamma}^{-\frac{1}{2}}\mathbf{X}$ as a LSTD whose basis matrices are given by $\mathbf{\Gamma}^{-\frac{1}{2}}\mathbf{A_i}, i=1,\dots,K$. Now applying the conditions for $g$-group ML decoding of collocated STBCs, we get the condition for $g$-group ML decoding of DSTBCs to be that: (1) $\mathscr{C}$ should be $g$-group encodable and (2) whenever $\mathbf{A_i}$ and $\mathbf{A_j}$ belong to different groups, they should satisfy
$$
(\mathbf{\Gamma}^{-\frac{1}{2}}\mathbf{A_i})^H(\mathbf{\Gamma}^{-\frac{1}{2}}\mathbf{A_j})+(\mathbf{\Gamma}^{-\frac{1}{2}}\mathbf{A_j})^H(\mathbf{\Gamma}^{-\frac{1}{2}}\mathbf{A_i})=\mathbf{0}
$$

\noindent which, on simplification gives \eqref{eqn_multigroup_cond_dstbc}.
\end{proof}

Note from \eqref{eqn_cov_matrix} that if all the relay matrices are restricted to be unitary as in \cite{JiH}, then $\Gamma$ becomes a scaled identity matrix which in turn makes the condition in \eqref{eqn_multigroup_cond_dstbc} coincide with that for collocated STBCs. To summarize, a $g$-group ML decodable collocated STBC qualifies to become a $g$-group ML decodable DSTBC if it satisfies the below three conditions.
\begin{enumerate}
\item The associated LSTD $\mathbf{X}=\sum_{i=1}^{K}x_i\mathbf{A_i}$ is a conjugate LSTD.
\item The associated relay matrices are row orthogonal, i.e., $\mathbf{B_i}\mathbf{B_i}^H$ is a diagonal matrix.
\item Equation (21) is satisfied by the associated basis matrices $\mathbf{A_i},~i=1,\dots,K$.
\end{enumerate}
\begin{eg}
\label{eg_4x4OD}
Consider the $4\times 4$ single real symbol ML decodable (6-group ML decodable ) STBC from $4\times 4$ orthogonal design given by
$\left[\begin{array}{rrrr}
z_1 & -z_2^* & -z_3^* & 0\\
z_2 & z_1^* & 0 & -z_3^*\\
z_3 & 0 & z_1^* & -z_2\\
0 & z_3 & z_2^* & z_1 \end{array}\right]$. Note that it is not a conjugate LSTD and hence does not qualify as a DSTBC.
\end{eg}
Example \ref{eg_4x4OD} demonstrates that though orthogonal designs and hence single real symbol ML decodable collocated STBCs are well known in the literature, the transition to distributed case is not straightforward. Thus it is clear that it is more difficult and challenging to construct multi-group ML decodable DSTBCs compared to multi-group ML decodable collocated STBCs.

\subsection{Four group decodable DSTBCs from Precoded CIODs}
\label{subsec3_3}

Towards constructing four-group decodable DSTBCs, consider the following example.

\begin{eg}
Consider the $4\times 4$ CIOD \cite{KhR} shown below
\label{eg_CIOD}
$$\mathbf{X}_{CIOD}=\sqrt{2}\left[\begin{array}{cccc}
z_1 & -z_2^* & 0 & 0\\
z_2 & z_1^* & 0 & 0\\
0 & 0 & z_3 & -z_4^*\\
0 & 0 & z_4 & z_3^*
\end{array}\right]
$$
where, $z_1=x_1+ix_2$, $z_2=x_3+ix_4$, $z_3=x_5+ix_6$ and $z_4=x_7+ix_8$ are complex variables. It is clear that $\mathbf{X}_{CIOD}$ is a conjugate LSTD. We shall now see how $\mathbf{X}_{CIOD}$ is actually a $4$-group decodable DSTBC. Let the number of relays $R=4$. In the broadcast phase let the source  transmit the vector $\sqrt{\pi_1P}\left[\begin{array}{cccc}z_1 & z_2 & z_3 & z_4\end{array}\right]^T$, where the information symbols $\left\{x_1,x_5\right\}$, $\left\{x_{2},x_{6}\right\}$, $\left\{x_{3},x_{7}\right\}$, $\left\{x_{4},x_{8}\right\}$ are each taken from a rotated QAM constellation as given below
\begin{equation*}
\left[\begin{array}{c}x_{i}\\x_{i+4}\end{array}\right]=\left[\begin{array}{cc}\cos\theta & -\sin\theta\\ \sin\theta & \cos\theta\end{array}\right]\left[\begin{array}{c}y_{iI}\\y_{iQ}\end{array}\right],i=1,\dots,4
\end{equation*}
where, $\left[\begin{array}{c}y_{iI}\\y_{iQ}\end{array}\right],i=1,\dots 4,$ take values from a QAM constellation and $\theta$ is an appropriately chosen rotation angle \cite{KhR} so that the resulting DSTBC satisfies the rank criterion for full diversity according to Theorem \ref{thm_codinggain}. For $\mathbf{X}_{CIOD}$, the value of $M=2$ and the corresponding $4$ relay matrices are:
\begin{equation*}
\begin{array}{l}
\mathbf{B_1}=\left[\begin{array}{cccc}
\sqrt{2} & 0 & 0 & 0\\
0 & \sqrt{2} & 0 & 0\\
0 & 0 & 0 & 0\\
0 & 0 & 0 & 0
\end{array}\right], \mathbf{B_2}=\left[\begin{array}{cccc}
0 & -\sqrt{2} & 0 & 0\\
\sqrt{2} & 0 & 0 & 0\\
0 & 0 & 0 & 0\\
0 & 0 & 0 & 0
\end{array}\right]\\
\mathbf{B_3}=\left[\begin{array}{cccc}
0 & 0 & 0 & 0\\
0 & 0 & 0 & 0\\
0 & 0 & \sqrt{2} & 0\\
0 & 0 & 0 & \sqrt{2}
\end{array}\right], \mathbf{B_4}=\left[\begin{array}{cccc}
0 & 0 & 0 & 0\\
0 & 0 & 0 & 0\\
0 & 0 & 0 & -\sqrt{2}\\
0 & 0 & \sqrt{2} & 0
\end{array}\right].
\end{array}
\end{equation*}
The corresponding matrix $\mathbf{\Gamma}$ is given by
\begin{equation*}
\mathbf{\Gamma}=\mathbf{I_4}+\frac{\pi_2P}{\pi_1P+1}\left[\begin{array}{cc}2\left(\vert g_1\vert^2+\vert g_2\vert^2\right)I_{2} & 0\\
0 & 2\left(\vert g_3\vert^2+\vert g_4\vert^2\right)I_{2}
\end{array}\right].
\end{equation*}
The weight matrices for $\mathbf{X}_{CIOD}$ are given as follows:

\begin{equation*}
\begin{array}{l}
\mathbf{A_1}=\sqrt{2}\left[\begin{array}{cccc}
1 & 0 & 0 & 0\\
0 & 1 & 0 & 0\\
0 & 0 & 0 & 0\\
0 & 0 & 0 & 0
\end{array}\right], \mathbf{A_2}=\sqrt{2}\left[\begin{array}{cccc}
i & 0 & 0 & 0\\
0 & -i & 0 & 0\\
0 & 0 & 0 & 0\\
0 & 0 & 0 & 0
\end{array}\right], \mathbf{A_3}=\sqrt{2}\left[\begin{array}{cccc}
0 & -1 & 0 & 0\\
1 & 0 & 0 & 0\\
0 & 0 & 0 & 0\\
0 & 0 & 0 & 0
\end{array}\right],\\
\mathbf{A_4}=\sqrt{2}\left[\begin{array}{cccc}
0 & i & 0 & 0\\
i & 0 & 0 & 0\\
0 & 0 & 0 & 0\\
0 & 0 & 0 & 0
\end{array}\right], \mathbf{A_5}=\sqrt{2}\left[\begin{array}{cccc}
0 & 0 & 0 & 0\\
0 & 0 & 0 & 0\\
0 & 0 & 1 & 0\\
0 & 0 & 0 & 1
\end{array}\right], \mathbf{A_6}=\sqrt{2}\left[\begin{array}{cccc}
0 & 0 & 0 & 0\\
0 & 0 & 0 & 0\\
0 & 0 & i & 0\\
0 & 0 & 0 & -i
\end{array}\right],\\
\mathbf{A_7}=\sqrt{2}\left[\begin{array}{cccc}
0 & 0 & 0 & 0\\
0 & 0 & 0 & 0\\
0 & 0 & 0 & -1\\
0 & 0 & 1 & 0
\end{array}\right], \mathbf{A_8}=\sqrt{2}\left[\begin{array}{cccc}
0 & 0 & 0 & 0\\
0 & 0 & 0 & 0\\
0 & 0 & 0 & i\\
0 & 0 & i & 0
\end{array}\right].
\end{array}
\end{equation*}

It is easy to check that all the weight matrices are row orthogonal and they satisfy \eqref{eqn_multigroup_cond_dstbc} for $\lambda=1$. This is because of the special block diagonal structure of $\mathbf{X}_{CIOD}$ with each block being a replica of the Alamouti LSTD. The resulting DSTBC will achieve full cooperative diversity and is $4$-group ML decodable or equivalently one complex symbol ML decodable.
\end{eg}

We now generalize the LSTD $\mathbf{X}_{CIOD}$ given in Example \ref{eg_CIOD} for any number of relays having the special feature of $4$-group decodability. We call these LSTDs as `Precoded CIODs' (PCIODs).

\noindent
\textbf{Construction of Precoded CIOD for even number of relays:}\\
Given $R$ an even number, the $R\times R$ PCIOD $\mathbf{X}_{PCIOD}$ for $R$ relays is given by \eqref{eqn_PCIOD_cons}.

{\tiny
\begin{equation}
\label{eqn_PCIOD_cons}
\mathbf{X}_{CIOD}= \mathrm{diag}\left\{\left[\begin{array}{cc}
x_1+ix_2 & -x_3+ix_4\\
x_3+ix_4 & x_1-ix_2\end{array}\right],\dots,\left[\begin{array}{cc}
x_{k}+jx_{k+1} & -x_{k+2}+jx_{k+3}\\
x_{k+2}+jx_{k+3} & x_{k}-jx_{k+1}\end{array}\right],\dots,\left[\begin{array}{cc}
x_{2R-3}+ix_{2R-2} & -x_{2R-1}+ix_{2R}\\
x_{2R-1}+ix_{2R} & x_{2R-3}-ix_{2R-2}
\end{array}\right]\right\}
\end{equation}
}
There are totally $2R$ real variables $x_1,x_2,\dots,x_{2R}$ in the conjugate LSTD $\mathbf{X}_{PCIOD}$. The following expression shows that $\mathbf{X}_{PCIOD}$ is not fully diverse for arbitrary signal sets.
\begin{equation*}
\vert(\Delta \mathbf{X}_{PCIOD})^H(\Delta \mathbf{X}_{PCIOD})\vert=\left(\sum_{i=1}^{4}\Delta x_i^2\right)^2\dots\left(\sum_{i=k}^{k+3}\Delta x_i^2\right)^2\dots\left(\sum_{i=2R-3}^{2R}\Delta x_i^2\right)^2
\end{equation*}
However, constellation precoding can be done to achieve full diversity. Precoding is to be done in the following manner. The $2R$ real variables are first partitioned into 4 groups as follows:
\begin{itemize}
\item First group: $\left\{x_{1+4k}|k=0,1,\dots,\frac{2R-4}{4}\right\}$
\item Second group: $\left\{x_{2+4k}|k=0,1,\dots,\frac{2R-3}{4}\right\}$
\item Third group: $\left\{x_{3+4k}|k=0,1,\dots,\frac{2R-2}{4}\right\}$
\item Fourth group: $\left\{x_{4+4k}|k=0,1,\dots,\frac{2R-1}{4}\right\}$.
\end{itemize}
There are $\frac{R}{2}$ real variables in each group. Now let $\mathbf{X_i}, i=1,\dots,4$ denote the LSTDs corresponding to only the real variables in the $i$-th group respectively. Now $\mathbf{X}_{PCIOD}=\sum_{i=1}^{4}\mathbf{X_i}$. Also it can be checked that

$$
(\Delta \mathbf{X}_{PCIOD})^H(\Delta \mathbf{X}_{PCIOD})=\sum_{i=1}^{4}\Delta \mathbf{X}_{i}^H\Delta \mathbf{X}_{i}.
$$

Supposing the constituent STBCs corresponding to LSTDs $\mathbf{X}_i,~i=1,\dots,4$ are fully diverse, then $\vert(\Delta \mathbf{X}_{i})^H(\Delta \mathbf{X}_{i})\vert\geq 0~\forall~ i=1,\dots,4$ and on application of Corollary 4.3.3 in \cite{HoJ}, we get

$$
\vert(\Delta \mathbf{X}_{PCIOD})^H(\Delta \mathbf{X}_{PCIOD})\vert\geq\min_{i=1,\dots,4}\left\{\vert(\Delta \mathbf{X}_{i}^H)(\Delta \mathbf{X}_{i})\vert\right\}.
$$

Thus we see that if the constituent STBCs are fully diverse, then the resulting STBC from PCIOD will also be fully diverse. Note that $\vert(\Delta \mathbf{X}_{i})^H(\Delta \mathbf{X}_{i})\vert=(\prod_{j=0}^{\frac{R}{2}-1}\Delta x_{i+4j})^2$ which is nothing but the product distance. Hence, if we let the $\frac{R}{2}$ real variables in a group to take values from a rotated $\mathbb{Z}^{\frac{R}{2}}$ lattice constellation which is designed to maximize the minimum product distance then full diversity is guaranteed. Algebraic number theory provides effective means to construct rotated $\mathbb{Z}^n$ lattices with large minimum product distance \cite{FOV,Vit} for any $n\in\mathbb{N}$ and the corresponding lattice generator matrices can be explicitly obtained from \cite{Vit} for dimensions upto $30$. Due to the block diagonal nature of $\mathbf{X}_{PCIOD}$ with replicas of Alamouti designs on the blocks, the resulting DSTBC will be a full diversity $4$-group decodable DSTBC. The following example illustrates the construction procedure for $R=6$.
\begin{eg}
The PCIOD for $6$ relays is as shown below:

\begin{equation*}
\mathrm{diag}\left\{\left[\begin{array}{cc}
x_1+ix_2 & -x_3+ix_4\\
x_3+ix_4 & x_1-ix_2\end{array}\right],\left[\begin{array}{cc}
x_5+ix_6 & -x_7+ix_8\\
x_7+ix_8 & x_5-ix_6\end{array}\right],\left[\begin{array}{cc}
x_9+ix_{10} & -x_{11}+ix_{12}\\
x_{10}+ix_{11} & x_9-ix_{10}\end{array}\right]\right\}
\end{equation*}

\noindent where,

\begin{equation*}
\left[\begin{array}{c}x_{i}\\x_{i+4}\\x_{i+8}\end{array}\right]=\mathbf{\mathcal{G}}\left[\begin{array}{c}y_{i}\\y_{i+4}\\y_{i+8}\end{array}\right], i=1,\dots,4\\
\end{equation*}
\noindent and the vectors $\left[\begin{array}{c}y_i\\y_{i+4}\\y_{i+8}\end{array}\right],i=1,2,3,4$ take values from a subset of $\mathbb{Z}^3$. The $3\times 3$ lattice generator matrix $\mathbf{\mathcal{G}}$ can be taken from \cite{Vit}. At the destination, the ML decoding of the real variables $\left\{x_i,x_{i+4},x_{i+8}\right\}$ has to be done jointly for each $i=1,2,3,4$ separately. Thus the resulting DSTBC is $4$-group decodable or $3$-real symbol decodable.
\end{eg}

\noindent
\textbf{Construction of Precoded CIOD for odd number of relays:}\\
If $R$ is odd, then construct a PCIOD for $R+1$ relays and drop the last column to get a $(R+1)\times R$ LSTD. For example, a single complex symbol decodable code for $3$ relays can be obtained from Example \ref{eg_CIOD} by dropping the last column. This is shown in the following example.
\begin{eg}
\begin{equation*}
\left[\begin{array}{rcl}
x_{1}+ix_{2} & x_{3}+ix_{4} & 0\\
-x_{3}+ix_{4} & x_{1}-ix_{2} & 0\\
0 & 0 & x_{5}+ix_{6}\\
0 & 0 & -x_{7}+ix_{8}
\end{array}\right]
\end{equation*}
\end{eg}

\subsubsection{Encoding complexity at the relays for PCIODs}
\label{subsubsec3_3_1}
~\\
By observing the structure of the relay matrices of PCIOD one can see that it has zeros everywhere except for a single non-zero $2\times 2$ sub-matrix which is a scaled version of either identity matrix or $\left[\begin{array}{cc}0 & -1\\1 & 0\end{array}\right]$. Thus having received two complex numbers, say $\left[\begin{array}{c}a+ib\\c+id\end{array}\right]$, a relay should be capable of generating and transmitting one of the following: $\left[\begin{array}{c}a+ib\\c+id\end{array}\right]$ or $\left[\begin{array}{c}-a+ib\\c-id\end{array}\right]$ both of which require significantly less complexity as compared to multiplying the received vector by an arbitrary complex matrix.

\subsubsection{Resistance to relay node failures}
\label{subsubsec3_3_2}
~\\
Note that any two columns of the PCIOD are orthogonal. This leads to the property that even if any column of the design is dropped, it continues to satisfy the full rank condition. This property is important since even if certain relay nodes fail, which is equivalent to dropping few columns of the LSTD, the residual diversity benefits are still guaranteed and that too with no additional increase in ML decoding complexity.

Thus we have a constructed a class of $4$-group decodable DSTBCs for any number of relays having the following salient features:
\begin{enumerate}
\item Transmission rate of the source is $0.5$ complex symbols per channel use
\item Full diversity
\item Four group ML decodable
\item Low encoding complexity at the relays
\item Resistance to relay node failures
\end{enumerate}

\subsection{Four group decodable DSTBCs from extended Clifford algebras}
\label{subsec3_4}

In the previous subsection, a class of $4$-group decodable DSTBCs was constructed for arbitrary number of relays from PCIODs. Amidst many advantages, PCIODs do have a drawback that the power distribution among the relays is not uniform across time slots which is due to the large number of zeros in the LSTD. This leads to a large peak to average power ratio (PAPR) problem at the relays which is undesirable since it demands the use of larger power amplifiers at the relays. Moreover since the relay matrices of PCIODs are not unitary, this forces the destination to perform additional processing to make the noise covariance matrix a scaled identity matrix, i.e., pre-multiplying the received vector by $\mathbf{\Gamma}^{-\frac{1}{2}}$. Above all, the construction of PCIODs was not obtained from a systematic algebraic procedure targeting the requirements for $4$-group decodable DSTBCs. Hence it is natural to ask whether there exists a systematic algebraic construction of $4$-group decodable DSTBCs with unitary relay matrices and uniform power distribution across the relays and in time.

In this subsection, using the algebraic framework of extended Clifford algebras introduced in Subsection \ref{subsec2_2}, two new classes of $4$-group decodable DSTBCs with unitary relay matrices as well as unitary weight matrices for power of two number of relays are constructed.

As discussed in Subsection \ref{subsec2_2}, to construct $4$-group decodable DSTBCs, we need $4$ matrices (including identity matrix) in the first row (as shown in Table \ref{table_CUWD}). One way to obtain such matrices is to take the matrix representation of $\mathbb{A}_3^L$ for $L=2^a,a\in\mathbb{N}$. The matrix representation of the symbols $1,\gamma_1,\gamma_2,\gamma_3$ respectively can be used to fill up the first row. Interestingly there is yet another way of obtaining such matrices. Let us look at $\mathbb{A}_2^L$ for $L=2^a,a\in \mathbb{N}$. The symbols $\gamma_1$ and $\gamma_2$ square to $-1$ and anti-commute. However, note that

\begin{equation}
\begin{array}{rcl}
(\gamma_2\gamma_1)^2&=&-1\\
(\gamma_2\gamma_1)\gamma_1&=&-\gamma_1(\gamma_2\gamma_1)\\
(\gamma_2\gamma_1)\gamma_2&=&-\gamma_2(\gamma_2\gamma_1)
\end{array}.
\end{equation}

Thus the symbol $\gamma_2\gamma_1$ also squares to $-1$ and anti-commutes with the symbols $\gamma_1$ and $\gamma_2$. Thus we can fill up the first row with the matrix representations of the symbols $1,\gamma_1,\gamma_2,\gamma_2\gamma_1$ respectively. Thus we can get two classes of $4$-group decodable DSTBCs, one from $\mathbb{A}_3^L$ and the other from $\mathbb{A}_2^L$ if the problem of conjugate linearity property and unitary relay matrices are also taken care.

\subsubsection{Matrix Representation}
\label{subsubsec3_4_1}
~\\
There are several ways to obtain a matrix representation of an algebra. We need to take an appropriate matrix representation such that the following conditions are met.

\begin{enumerate}
\item The symbols $1$, $\gamma_1$, $\gamma_2$, $\dots$, $\gamma_n$, $\delta_k,$ for $k=1,\dots,a$, and $\bigcup_{m=2}^{a}\prod_{i=1}^{m}\delta_{k_i}$ for $1\leq k_i\leq k_{i+1}\leq a$ should be represented by unitary matrices.
\item The resulting LSTD should be a conjugate LSTD.
\item All the relay matrices should be unitary.
\end{enumerate}

Such matrices are naturally provided by the left regular representation of the associative algebra $\mathbb{A}_n^L$. Left regular representation is an easy way to obtain the matrix representation for any finite dimensional associative algebra \cite{Jac}. Such techniques have been previously used in \cite{SRS1,SRS2} to obtain the matrix representation of division algebras and crossed product algebras. The first requirement of unitary matrix representation is met because the natural basis elements of $\mathbb{A}_n^L$ over $\mathbb{R}$ together with their negatives form a finite group under multiplication (see Proposition \ref{prop_group}). This fact in conjunction with the properties of left regular representation guarantee a unitary matrix representation for the required symbols. We shall prove the other properties after illustrating the construction procedure for both the codes from $\mathbb{A}_2^L$ and $\mathbb{A}_3^L$.

\subsubsection{DSTBCs from $\mathbb{A}_2^L$}
\label{subsubsec3_4_2}
~\\
We first view $\mathbb{A}_2^L$ as a vector space over $\mathbb{C}$ by thinking of $\gamma_1$ as the complex number $i=\sqrt{-1}$. A natural $\mathbb{C}$~ basis for $\mathbb{A}_2^L$ is given by

\begin{equation}
\mathcal{B}_n^L=\left\{1,\gamma_2\right\}\cup\left\{\left\{1,\gamma_2\right\}\delta_i|i=1,\dots,a\right\}  \bigcup_{m=2}^{a}\left\{1,\gamma_2\right\}\left\{\prod_{i=1}^{m}\delta_{k_i}|1\leq k_i\leq k_{i+1}\leq a\right\}.
\end{equation}
\noindent
Thus the dimension of $\mathbb{A}_2^L$ seen as a vector space over $\mathbb{C}$ is $2^{n+a-1}$. We have a natural map from $\mathbb{A}_2^L$~ to $\mathrm{End}_{\mathbb{C}}(\mathbb{A}_2^L)$ given by left multiplication \cite{SRS1,SRS2,Jac} as shown below.

\begin{equation}
\begin{array}{l}
\phi:\mathbb{A}_2^L\mapsto \mathrm{End}_{\mathbb{C}}(\mathbb{A}_2^L)\\
\phi(x)=L_x:y\mapsto xy
\end{array}
\end{equation}

Since the map $L_x$ is $\mathbb{C}$ linear, we can write down a matrix representation of $L_x$ with respect to the natural $\mathbb{C}$~ basis $\mathcal{B}_n^L$. Thus we obtain a LSTD satisfying the requirements of \eqref{eqn_g_dec_cond} for $g=4$.

\begin{eg}
Let us begin with the simplest case of $R=2^1$ relays. Let $n=2$. Then equating $n+a-1=1$, we get $a=0$ and hence $L=1$. But the algebra $\mathbb{A}_2^1$ is same as $Cliff(2)$ which is nothing but the Hamiltonian Quaternions $\mathbb{H}$. It is well known \cite{SRS1} that the left regular matrix representation of $\mathbb{H}$ yields the popular Alamouti design. Thus we see that our algebraic code construction which was driven by the need for low ML decoding complexity naturally leads to the Alamouti design.
\end{eg}

\begin{eg}
Suppose we want a LSTD for $R=8=2^3$ relays. Let $n=2$. Then we need $n+a-1=3$. Thus $a=2$ and $L=4$. A general element of the algebra $\mathbb{A}_2^4$ looks like

$$
x=z_1+\delta_1z_2+\delta_2z_3+\delta_1\delta_2z_4+\gamma_2z_5+\delta_1\gamma_2z_6+\delta_2\gamma_2z_7+\delta_1\delta_2\gamma_2z_8
$$
\noindent
where, $z_i\in\mathbb{C},\forall i=1,\dots,8$. We shall now find the matrix representation of $L_x$ by finding out the image of the basis $\mathcal{B}_2^4$ under the map $L_x$ which is shown below.

\begin{equation}
\label{eqn_matrix}
\begin{array}{rcl}
L_x(1)&=&z_1+\delta_1z_2+\delta_2z_3+\delta_1\delta_2z_4+\gamma_2z_5+\delta_1\gamma_2z_6+\delta_2\gamma_2z_7+\delta_1\delta_2\gamma_2z_8\\
L_x(\delta_1)&=&\delta_1z_1+z_2+\delta_1\delta_2z_3+\delta_2z_4+\delta_1\gamma_2z_5+\gamma_2z_6+\delta_1\delta_2\gamma_2z_7+\delta_2\gamma_2z_8\\
L_x(\delta_2)&=&\delta_2z_1+\delta_1\delta_2z_2+z_3+\delta_1z_4+\delta_2\gamma_2z_5+\delta_1\delta_2\gamma_2z_6+\gamma_2z_7+\delta_1\gamma_2z_8\\
L_x(\delta_1\delta_2)&=&\delta_2\delta_2z_1+\delta_2z_2+\delta_1z_3+z_4+\delta_1\delta_2\gamma_2z_5+\delta_2\gamma_2z_6+\delta_1\gamma_2z_7+\gamma_2z_8\\
L_x(\gamma_2)&=&\left(z_1+\delta_1z_2+\delta_2z_3+\delta_1\delta_2z_4+\gamma_2z_5+\delta_1\gamma_2z_6+\delta_2\gamma_2z_7+\delta_1\delta_2\gamma_2z_8\right)\gamma_2\\
&=&\gamma_2z_1^*+\delta_1\gamma_2z_2^*+\delta_2\gamma_2z_3^*+\delta_1\delta_2\gamma_2z_4^*-z_5^*-\delta_1z_6^*-\delta_2z_7^*-\delta_1\delta_2z_8^*\\
L_x(\delta_1\gamma_2)&=&\delta_1\gamma_2z_1^*+\gamma_2z_2^*+\delta_1\delta_2\gamma_2z_3^*+\delta_2\gamma_2z_4^*-\delta_1z_5^*-z_6^*-\delta_1\delta_2z_7^*-\delta_2z_8^*\\
L_x(\delta_2\gamma_2)&=&\delta_2\gamma_2z_1^*+\delta_1\delta_2\gamma_2z_2^*+\gamma_2z_3^*+\delta_1\gamma_2z_4^*-\delta_2z_5^*-\delta_2z_6^*-z_7^*-\delta_1z_8^*\\
L_x(\delta_1\delta_2\gamma_2)&=&\delta_1\delta_2\gamma_2z_1^*+\delta_2\gamma_2z_2^*+\delta_1\gamma_2z_3^*+\gamma_2z_4^*-\delta_1\delta_2z_5^*-\delta_2z_6^*-\delta_1z_7^*-z_8^*
\end{array}
\end{equation}

The matrix representation of $L_x$ is thus given by:
\begin{equation}
\label{eqn_design8}
[L_x]=\left[\begin{array}{ccccrrrr}
z_1 & z_2 & z_3 & z_4 & -z_5^* & -z_6^* & -z_7^* & -z_8^*\\
z_2 & z_1 & z_4 & z_3 & -z_6^* & -z_5^* & -z_8^* & -z_7^*\\
z_3 & z_4 & z_1 & z_2 & -z_7^* & -z_8^* & -z_5^* & -z_6^*\\
z_4 & z_3 & z_2 & z_1 & -z_8^* & -z_7^* & -z_6^* & -z_5^*\\
z_5 & z_6 & z_7 & z_8 & z_1^* & z_2^* & z_3^* & z_4^*\\
z_6 & z_5 & z_8 & z_7 & z_2^* & z_1^* & z_4^* & z_3^*\\
z_7 & z_8 & z_5 & z_6 & z_3^* & z_4^* & z_1^* & z_2^*\\
z_8 & z_7 & z_6 & z_5 & z_4^* & z_3^* & z_2^* & z_1^*
\end{array}\right].
\end{equation}

Also, we have that
\begin{equation}
\begin{array}{rcl}
x&=&z_{1I}+\gamma_1z_{1Q}+\delta_1z_{2I}+\delta_1\gamma_1z_{2Q}+\delta_2z_{3I}+\delta_2\gamma_1z_{3Q}+\delta_1\delta_2z_{4I}+\delta_1\delta_2\gamma_1z_{4Q}+\gamma_2z_{5I}\\
&&+\gamma_2\gamma_1z_{5Q}+\delta_1\gamma_2z_{6I}+\delta_1\gamma_2\gamma_1z_{6Q}+\delta_2\gamma_2z_{7I}+\delta_2\gamma_2\gamma_1z_{7Q}+\delta_1\delta_2\gamma_2z_{8I}+\delta_1\delta_2\gamma_2\gamma_1z_{8Q}.
\end{array}
\end{equation}

Since the map $\phi$ is a ring homomorphism, we have
\begin{equation}
\label{eqn_wtmat}
\begin{array}{rl}
L_x=&\phi(1)z_{1I}\phi(1)+\phi(\gamma_1)z_{1Q}+\phi(\delta_1)z_{2I}+\phi(\delta_1\gamma_1)z_{2Q}+\phi(\delta_2)z_{3I}+\phi(\delta_2\gamma_1)z_{3Q}\\
&+\phi(\delta_1\delta_2)z_{4I}+\phi(\delta_1\delta_2\gamma_1)z_{4Q}+\phi(\gamma_2)z_{5I}+\phi(\gamma_2\gamma_1)z_{5Q}+\phi(\delta_1\gamma_2)z_{6I}+\phi(\delta_1\gamma_2\gamma_1)z_{6Q}\\
&+\phi(\delta_2\gamma_2)z_{7I}+\phi(\delta_2\gamma_2\gamma_1)z_{7Q}\\
&+\phi(\delta_1\delta_2\gamma_2)z_{8I}+\phi(\delta_1\delta_2\gamma_2\gamma_1)z_{8Q}.
\end{array}
\end{equation}

The equation \eqref{eqn_wtmat} explicitly gives the design $[L_x]$ in terms of its weight matrices. Because of our algebraic construction, the weight  matrices can be partitioned into four groups such that \eqref{eqn_g_dec_cond} is satisfied. The four groups are
\begin{enumerate}
\item $\left\{\phi(1),\phi(\delta_1),\phi(\delta_2),\phi(\delta_1\delta_2)\right\}$
\item $\left\{\phi(\gamma_1),\phi(\delta_1\gamma_1),\phi(\delta_2\gamma_1),\phi(\delta_1\delta_2\gamma_1)\right\}$
\item $\left\{\phi(\gamma_2),\phi(\delta_1\gamma_2),\phi(\delta_2\gamma_2),\phi(\delta_1\delta_2\gamma_2)\right\}$
\item $\left\{\phi(\gamma_2\gamma_1),\phi(\delta_1\gamma_2\gamma_1),\phi(\delta_2\gamma_2\gamma_1),\phi(\delta_1\delta_2\gamma_2\gamma_1)\right\}$
\end{enumerate}
respectively. Expressing the real variables of the resulting design and their corresponding weight matrices in the form of a tabular column as shown in Table \ref{table_CUWD}, we get

\begin{center}
\begin{tabular}{|c|c|c|c|}
\hline
$\phi(1)$ & $\phi(\gamma_1)$ & $\phi(\gamma_2)$ & $\phi(\gamma_2\gamma_1)$\\
$z_{1I}$ & $z_{1Q}$ & $z_{5I}$ & $z_{5Q}$\\
\hline
$\phi(\delta_1)$ & $\phi(\delta_1)\phi(\gamma_1)$ & $\phi(\delta_1)\phi(\gamma_2)$ & $\phi(\delta_1)\phi(\gamma_2\gamma_1)$\\
$z_{2I}$ & $z_{2Q}$ & $z_{6I}$ & $z_{6Q}$\\
\hline
$\phi(\delta_2)$ & $\phi(\delta_2)\phi(\gamma_1)$ & $\phi(\delta_2)\phi(\gamma_2)$ & $\phi(\delta_2)\phi(\gamma_2\gamma_1)$\\
$z_{3I}$ & $z_{3Q}$ & $z_{7I}$ & $z_{7Q}$\\
\hline
$\phi(\delta_1\delta_2)$ & $\phi(\delta_1\delta_2)\phi(\gamma_1)$ & $\phi(\delta_1\delta_2)\phi(\gamma_2)$ & $\phi(\delta_1\delta_2)\phi(\gamma_2\gamma_1)$\\
$z_{4I}$ & $z_{4Q}$ & $z_{8I}$ & $z_{8Q}$\\
\hline
\end{tabular}
\end{center}
\end{eg}

In general for $R=2^m$ relays we can take the left regular representation of $\mathbb{A}_2^{2^{m-1}}$ to obtain a $4$-group decodable LSTD. These LSTDs were first obtained using a non-algebraic iterative construction procedure in \cite{RaR3}. The algebraic framework presented in this paper provides an interesting algebraic description for codes in \cite{RaR3}.

\begin{remark}
Note that in general to represent $L_x$ as a matrix one could have chosen any basis for $\mathbb{A}_2^L$ instead of the natural basis $\mathcal{B}_2^L$. However, only the natural basis will lead to a design with the low decoding complexity requirements, although a different basis will also give a representation of the same algebra. This shows that although two designs can be algebraically isomorphic, the choice of basis is crucial and only certain basis admit low decoding complexity. Further, even changing the ordering of the natural basis can result in designs which apparently look very different. But this is same as simply applying a permutation to the rows and columns.
\end{remark}

\subsubsection{DSTBCs from $\mathbb{A}_3^L$}
\label{subsubsec3_4_3}
~\\
We use a slightly different approach to obtain codes from $\mathbb{A}_3^L$. Let us first consider the algebra, $\mathbb{A}_3^1$ which is nothing but $Cliff_3$. A general element of $Cliff_3$ looks like

\begin{equation}
\begin{array}{rcl}
x&=&\hat{a}_1+\gamma_1\hat{a}_2+\gamma_2\hat{a}_3+\gamma_3\hat{a}_4+\gamma_1\gamma_2\hat{a}_5+\gamma_2\gamma_3\hat{a}_6+\gamma_1\gamma_3\hat{a}_7+\gamma_1\gamma_2\gamma_3\hat{a}_8
\end{array}
\end{equation}
\noindent for some $\hat{a}_i\in\mathbb{R},i=1,\dots,8$. Note that we have used the natural $\mathbb{R}$ basis of $Cliff_3$ to represent an element of $Cliff_3$. The element $\gamma_1\gamma_2\gamma_3$ satisfies the following properties.

\begin{equation}
\begin{array}{rcl}
(\gamma_1\gamma_2\gamma_3)^2&=&1\\
(\gamma_1)(\gamma_1\gamma_2\gamma_3)&=&(\gamma_1\gamma_2\gamma_3)(\gamma_1)\\
(\gamma_2)(\gamma_1\gamma_2\gamma_3)&=&(\gamma_1\gamma_2\gamma_3)(\gamma_2)\\
(\gamma_3)(\gamma_1\gamma_2\gamma_3)&=&(\gamma_1\gamma_2\gamma_3)(\gamma_3)
\end{array}
\end{equation}
Thus the element $\gamma_1\gamma_2\gamma_3$ squares to $1$ and commutes with all the generators of $Cliff_3$. Hence the matrix representation of the element $\gamma_1\gamma_2\gamma_3$ can be used as a candidate to fill up the first column. Since we have now filled up two matrices (including identity matrix) in the first column, it should be possible to get a $2$-real symbol decodable STBC using matrix representation of $Cliff_3$. From Subsection \ref{subsec2_2}, we know that the remaining weight matrices are simply the product of matrices in the first row and those in the first column. We have,

\begin{equation}
(\gamma_1)(\gamma_1\gamma_2\gamma_3)=-\gamma_2\gamma_3,~~~
(\gamma_2)(\gamma_1\gamma_2\gamma_3)=\gamma_1\gamma_3 \mbox{  and  }
(\gamma_3)(\gamma_1\gamma_2\gamma_3)=-\gamma_1\gamma_2.
\end{equation}

It so turns out that the elements of $\left\{1, \gamma_1, \gamma_2, \gamma_3, -\gamma_1\gamma_2, -\gamma_2\gamma_3, \gamma_1\gamma_3, \gamma_1\gamma_2\gamma_3\right\}$ also form a basis for $Cliff_3$ over $\mathbb{R}$. Thus a general element of $Cliff_3$ can be expressed as
\begin{equation}
\begin{array}{rcl}
x&=&a_1+\gamma_1a_2+\gamma_2a_3+\gamma_3a_4+(-\gamma_1\gamma_2)a_5+(-\gamma_2\gamma_3)a_6+(\gamma_1\gamma_3)a_7+\gamma_1\gamma_2\gamma_3a_8
\end{array}
\end{equation}
\noindent
\noindent for some $a_i\in\mathbb{R},i=1,\dots,8$. By thinking of the element $\gamma_1$ as the complex number $i=\sqrt{-1}$, we can view $Cliff_3$ as a vector space over $\mathbb{C}$. To be precise,
\begin{equation}
\begin{array}{rcl}
x&=&(a_1+\gamma_1a_2)+\gamma_2(a_3+\gamma_1a_5)+\gamma_3(a_4-\gamma_1a_7)+\gamma_2\gamma_3(-a_6+\gamma_1a_8)\\
&=&z_1+\gamma_2z_2+\gamma_3z_3+\gamma_2\gamma_3z_4
\end{array}
\end{equation}
\noindent where, $z_i\in\mathbb{C},i=1,\dots,4$ and are given by $z_1=(a_1+\gamma_1a_2)$, $z_2=(a_3+\gamma_1a_5)$, $z_3=(a_4-\gamma_1a_7)$ and $z_4=(-a_6+\gamma_1a_8)$.

Now using left regular representation as in the case of codes from $\mathbb{A}_2^L$, we get
\begin{equation}
\begin{array}{rcl}
L_x(1)&=&z_1+\gamma_2z_2+\gamma_3z_3+\gamma_2\gamma_3z_4\\
L_x(\gamma_2)&=&\gamma_2z_1^*-z_2^*-\gamma_2\gamma_3z_3^*+\gamma_3z_4^*\\
L_x(\gamma_3)&=&\gamma_3z_1^*+\gamma_2\gamma_3z_2^*-z_3^*-\gamma_2z_4^*\\
L_x(\gamma_2\gamma_3)&=&\gamma_2\gamma_3z_1+\gamma_2^2\gamma_3z_2+\gamma_3\gamma_2\gamma_3z_3-z_4\\
&=&\gamma_2\gamma_3z_1-\gamma_3z_2+\gamma_2z_3-z_4
\end{array}
\end{equation}
Hence we obtain the following LSTD $[L_x]$

\begin{equation}
[L_x]=\left[\begin{array}{rrrr}
z_1 & -z_2^* & -z_3^* & -z_4\\
z_2 & z_1^* & -z_4^* & z_3\\
z_3 & z_4^* & z_1^* & -z_2\\
z_4 & -z_3^* & z_2^* & z_1
\end{array}\right]
\end{equation}

By construction, the weight matrices and the real variables of the LSTD $[L_x]$ can be partitioned into four groups for decoding purposes which is illustrated in the following table.

\begin{center}
\begin{tabular}{|c|c|c|c|}
\hline
$\phi(1)$ & $\phi(\gamma_1)$ & $\phi(\gamma_2)$ & $\phi(\gamma_3)$\\
$z_{1I}$ & $z_{1Q}$ & $z_{2I}$ & $z_{3I}$\\
\hline
$\phi(\gamma_1\gamma_2\gamma_3)$ & $\phi(-\gamma_2\gamma_3)$ & $\phi(\gamma_1\gamma_3)$ & $\phi(-\gamma_1\gamma_2)$\\
$z_{4Q}$ & $z_{4I}$ & $z_{3Q}$ & $z_{2Q}$\\
\hline
\end{tabular}
\end{center}

In general, for $R=2^m$ relays we can take the left regular representation of $\mathbb{A}_3^{2^{m-2}}$.

\begin{eg}
Suppose we want a design for $R=8=2^3$ relays. Then, we have $m=3$. Consider the algebra $\mathbb{A}_3^2$. A typical element $x$ can be expressed as
$$
x=z_1+\gamma_2z_2+\gamma_3z_3+\gamma_2\gamma_3z_4+\delta_1z_5+\delta_1\gamma_2z_6\delta_1\gamma_3z_7+\delta_1\gamma_2\gamma_3z_8
$$
where,$z_1=(z_{1I}+\gamma_1z_{1Q})$, $z_2=(z_{2I}+\gamma_1z_{2Q})$, $z_3=(z_{3I}-\gamma_1z_{3Q})$, $z_4=(-z_{4I}+\gamma_1z_{4Q})$, $z_5=(z_{5I}+\gamma_1z_{5Q})$, $z_6=(z_{6I}+\gamma_1z_{6Q})$, $z_7=(z_{7I}-\gamma_1z_{7Q})$, $z_8=(-z_{8I}+\gamma_1z_{8Q})$
and $z_{iI},z_{iQ}\in\mathbb{R},i=1,\dots,8$. Using left regular representation, we get the following LSTD
\begin{equation}
\left[\begin{array}{rrrrrrrr}
z_1 & -z_2^* & -z_3^* & -z_4 & z_5 & -z_6^* & -z_7^* & -z_8\\
z_2 & z_1^* & -z_4^* & z_3 & z_6 & z_5^* & -z_8^* & z_7\\
z_3 & z_4^* & z_1^* & -z_2 & z_7 & z_8^* & z_5^* & -z_6\\
z_4 & -z_3^* & z_2^* & z_1 & z_8 & -z_7^* & z_6^* & z_5\\
z_5 & -z_6^* & -z_7^* & -z_8 & z_1 & -z_2^* & -z_3^* & -z_4\\
z_6 & z_5^* & -z_8^* & z_7 & z_2 & z_1^* & -z_4^* & z_3\\
z_7 & z_8^* & z_5^* & -z_6 & z_3 & z_4^* & z_1^* & -z_2\\
z_8 & -z_7^* & z_6^* & z_5 & z_4 & -z_3^* & z_2^* & z_1
\end{array}\right].
\end{equation}
The corresponding $4$ groups of real variables are $\left\{z_{1I},z_{4Q},z_{5I},z_{8Q}\right\}$, $\left\{z_{1Q},z_{4I},z_{5Q},z_{8I}\right\}$,\\ $\left\{z_{2I},z_{3Q},z_{6I},z_{7Q}\right\}$ and $\left\{z_{3I},z_{2Q},z_{7I},z_{6Q}\right\}$ respectively.
\end{eg}

\subsubsection{Features of DSTBCs from extended Clifford algebras}
\label{subsubsec3_4_4}
Note that both the LSTDs from $\mathbb{A}_2^L$ and $\mathbb{A}_3^L$ are conjugate linear. This is by virtue of the properties of left regular representation. While taking the left regular matrix representation, recall that we viewed the algebra as a vector space over $\mathbb{C}$ by thinking of the element $\gamma_1$ as the analogue of the complex number $i=\sqrt{-1}$. Any column of the design was then obtained as the image of a few elements of the natural basis of the algebra under the map $L_x$. All the elements of the natural basis of $\mathbb{A}_n^L$ have the property that they either commute with $\gamma_1$ or anti-commute with $\gamma_1$. When we find the image of a basis element say $\alpha$ under the map $L_x$, recall that we moved $\alpha$ past a complex number $z_i$. If $\alpha$ commutes with $\gamma_1$, then it leaves the complex number intact. If $\alpha$ anti-commutes with $\gamma_1$, then it inflicts conjugation while moving past the complex number. This property leads to conjugate LSTDs. This fact can be clearly observed in for instance \eqref{eqn_matrix}. Moreover it can be easily observed that all the relay matrices of the resulting designs are unitary. This is because the number of complex variables in the design is equal to the size of the matrix and by virtue of the left regular representation any complex variable appears only once in any column and also they appear in different positions in every column. Full diversity is guaranteed for both these classes of LSTDs because they are CUWDs and full diversity aspects for general CUWDs have been discussed in Subsection \ref{subsubsec2_2_1}. Also note that both these constructions meet the optimal rate-ML decoding complexity as stated in Theorem \ref{thm_maxrate} for $g=4$. The constructed DSTBCs can be easily described and they have a very nice structure which is due to the algebraic approach.

\section{OFDM based Distributed Space Time Coding for Asynchronous Relay Networks}
\label{sec4}

In this section, we consider symbol asynchronism among the relays and propose an OFDM based transmission scheme that can achieve full cooperative diversity in asynchronous relay networks. This transmission scheme is a generalization of the Li-Xia transmission scheme in \cite{LiX}. We briefly review the Li-Xia \cite{LiX} transmission scheme in Subsection \ref{subsec4_1} and in Subsection \ref{subsec4_2}, we describe the proposed transmission scheme and also provide code constructions based on the four group decodable DSTBCs constructed in the previous section. Finally, in Subsection \ref{subsec4_3}, it is shown how differential encoding at the source node can be combined with the proposed transmission scheme in Subsection \ref{subsec4_2} to arrive at a transmission scheme for non-coherent asynchronous relay networks.

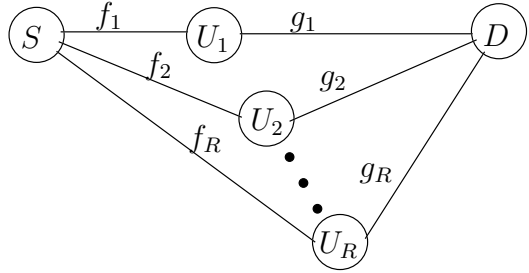
\begin{figure}[h]
\centering
\input{network.pstex_t}
\caption{Asynchronous wireless relay network}
\label{fig_network}
\end{figure}

An asynchronous wireless relay network is depicted in Fig. \ref{fig_network}. The overall relative timing error of the signals arrived at the destination node from the $i$-th relay node is denoted by $\tau_i$. Without loss of generality, it is assumed that $\tau_1=0$, $\tau_{i+1}\geq \tau_i,i=1,\dots,R-1$. The relay nodes are assumed to have perfect carrier synchronization. The destination node is assumed to have the knowledge of all the channel fading gains $f_i,g_i, i=1,\dots,R$ and the relative timing errors $\tau_i,i=1,\dots,R$. All the other assumptions are same as that made for the synchronous wireless relay network case.

In this OFDM based transmission scheme, the transmission of information from the source node to the destination node takes place in two phases. In the first phase, the source broadcasts the information to the relay nodes using OFDM. The relay nodes receive the faded and noise corrupted OFDM symbols, process them and transmit them to the destination.

\subsection{Li-Xia transmission scheme\cite{LiX}}
\label{subsec4_1}

The source takes $2N$ complex symbols $z_{i,j},$ $0\leq i\leq N-1, j=1,2$ and forms two blocks of data denoted by $\mathbf{a_j}=\left[\begin{array}{cccc}z_{0,j} & z_{1,j} & \dots & z_{N-1,j}\end{array}\right]^T, j=1,2$. The first block $\mathbf{a_1}$ is modulated by $N$-point Inverse Discrete Fourier Transform (IDFT) and $\mathbf{a_2}$ is modulated by $N$-point Discrete Fourier Transform (DFT). Then a cyclic prefix (CP) of length $l_{cp}$ is added to each block, where $l_{cp}$ is not less than the maximum of the overall relative timing errors of the signals arrived at the destination node from the relay nodes. The resulting two OFDM symbols denoted by $\mathbf{\bar{a}_1}$ and $\mathbf{\bar{a}_2}$ consisting of $L_s=N+l_{cp}$ complex numbers are broadcasted to the two relays using a fraction $\pi_1$ of the total average power $P$ consumed by the source and the relay nodes together.

If the channel fade gains $f_i,g_i, i=1,\dots,R$ are assumed to be constant for $4$ OFDM symbol intervals, the received signals at the $i$-th relay during the $j$-th OFDM symbol duration is given by
$$
\mathbf{r_{i,j}}=f_i\mathbf{\bar{a}_j}+\mathbf{\bar{v}_{i,j}}
$$

\noindent where, $\mathbf{\bar{v}_{i,j}}$ is the additive white Gaussian noise (AWGN) at the $i$-th relay node during the $j$-th OFDM symbol duration. The two relay nodes then process and transmit the resulting signals as shown in Table \ref{table_alamouti} using a fraction $\pi_2$ of the total power $P$. The notation $\zeta(.)$ denotes the time reversal operation, i.e., $\zeta(\mathbf{r}(n))\triangleq \mathbf{r}(L_s-n)$.

\begin{table}[h]
\caption{Alamouti code based transmission scheme}
\label{table_alamouti}
\begin{center}
\begin{tabular}{|c|c|c|}
\hline
OFDM Symbol & $U_1$ & $U_2$\\
\hline
$1$ & $\sqrt{\frac{\pi_2P}{\pi_1P+1}}\mathbf{r_{1,1}}$ & $-\sqrt{\frac{\pi_2P}{\pi_1P+1}}\mathbf{r_{2,2}}^*$\\
\hline
$2$ & $\sqrt{\frac{\pi_2P}{\pi_1P+1}}\zeta(\mathbf{r_{1,2}})$ & $\sqrt{\frac{\pi_2P}{\pi_1P+1}}\zeta(\mathbf{r_{2,1}}^*)$\\
\hline
\end{tabular}
\end{center}
\end{table}

The destination removes the CP for the first OFDM symbol and implements the following for the second OFDM symbol:

\begin{enumerate}
\item Remove the CP to get a $N$-point vector
\item Shift the last $l_{cp}$ samples of the $N$-point vector as the first $l_{cp}$ samples.
\end{enumerate}

DFT is then applied on the resulting two vectors. Since $l_{cp}\geq\tau_2$, the orthogonality between the sub carriers is still maintained. The delay in the time domain then translates to a corresponding phase change of $e^{-\frac{i2\pi k}{N}}$ in the $k$-th sub carrier. Let $d^{\tau_2}$ denote $\left[\begin{array}{cccc}1 & e^{-\frac{i2\pi\tau_2}{N}} & \dots & e^{-\frac{i2\pi\tau_2(N-1)}{N}} \end{array}\right]^T$. Then the received signals for two consecutive OFDM blocks after CP removal and DFT transformation denoted by $\mathbf{y_1}=\left[\begin{array}{cccc}y_{0,1} & y_{1,1} & \dots & y_{N-1,1}\end{array}\right]^T$ and $\mathbf{y_2}=\left[\begin{array}{cccc}y_{0,2} & y_{1,2} & \dots & y_{N-1,2}\end{array}\right]^T$ can be expressed as:

$$
\begin{array}{rl}
\mathbf{y_1}=& \sqrt{\frac{\pi_1\pi_2P^2}{\pi_1P+1}}(\mathrm{DFT}(\mathrm{IDFT}(\mathbf{a_1}))f_1g_1+\mathrm{DFT}(-(\mathrm{DFT}(\mathbf{a_2}))^*)\circ d^{\tau_2}f_2^*g_2)\\
&+\sqrt{\frac{\pi_2P}{\pi_1P+1}}(\mathbf{v_{1,1}}g_1-\mathbf{v_{2,2}}^*\circ d^{\tau_2}g_2)+\mathbf{w_1}\\
\mathbf{y_2}=& \sqrt{\frac{\pi_1\pi_2P^2}{\pi_1P+1}}(\mathrm{DFT}(\zeta(\mathrm{DFT}(\mathbf{a_2})))f_1g_1+\mathrm{DFT}(\zeta((\mathrm{IDFT}(\mathbf{a_1}))^*))\circ d^{\tau_2}f_2^*g_2)\\
&+\sqrt{\frac{\pi_2P}{\pi_1P+1}}(\mathbf{v_{1,2}}g_1+\mathbf{v_{2,1}}^*\circ d^{\tau_2}g_2)+\mathbf{w_2}
\end{array}
$$

\noindent where, $\circ$ denotes Hadamard product, $\mathbf{w_i}=(w_{k,i}), i=1,2$ is the AWGN at the destination and $\mathbf{v_{i,j}}$ denotes the DFT of $\mathbf{\bar{v}_{i,j}}$. Now using the identities

\begin{equation}
\label{eqn_identities}
\begin{array}{lcl}
(\mathrm{DFT}(\mathbf{x}))^*&=&\mathrm{IDFT}(\mathbf{x}^*)\\ (\mathrm{IDFT}(\mathbf{x}))^*&=&\mathrm{DFT}(\mathbf{x}^*)\\  \mathrm{DFT}(\zeta(\mathrm{DFT}(\mathbf{x})))&=&\mathbf{x}
\end{array}
\end{equation}

\noindent we get the Alamouti code form in each sub carrier $k,0\leq k\leq N-1$ as shown below:

$$
\begin{array}{ll}
\left[\begin{array}{c}y_{k,1}\\y_{k,2}\end{array}\right]=&\sqrt{\frac{\pi_1\pi_2P^2}{\pi_1P+1}}\left[\begin{array}{cc}
z_{k,1} & -z_{k,2}^*\\
 z_{k,2} & z_{k,1}^*
\end{array}\right]\left[\begin{array}{c}f_1g_1\\e^{-\frac{i2\pi k\tau_2}{N}}f_2^*g_2\end{array}\right]\\
&+\sqrt{\frac{\pi_2P}{\pi_1P+1}}\left[\begin{array}{c}\mathbf{v_{1,1}}(k)g_1-\mathbf{v_{2,2}}^*(k)e^{-\frac{i2\pi k\tau_2}{N}}g_2\\
\mathbf{v_{1,2}}(k)g_1+\mathbf{v_{2,1}}^*(k)e^{-\frac{i2\pi k\tau_2}{N}}g_2\end{array}\right]+\left[\begin{array}{c}w_{k,1}\\w_{k,2}\end{array}\right].
\end{array}
$$

With the power allocation $\pi_1=1$, $\pi_2=\frac{1}{R}$ and because of the Alamouti code form, diversity order of two can be achieved along with symbol-by-symbol ML decoding.

\subsection{Proposed Transmission Scheme}
\label{subsec4_2}

In this section, we extend the Li-Xia transmission scheme to a general transmission scheme that can achieve full asynchronous cooperative diversity for arbitrary number of relays. This nontrivial extension is based on analyzing the sufficient conditions required on the structure of STBCs which admit application in the Li-Xia transmission scheme.

\subsubsection{Transmission by the source node}
\label{subsubsec4_2_1}
~\\
The source takes $RN$ complex symbols $z_{i,j},$ $0\leq i\leq N-1, j=1,2,\dots,R,$ and forms $R$ blocks of data denoted by $\mathbf{a_j}=\left[\begin{array}{cccc}z_{0,j} & z_{1,j} & \dots & z_{N-1,j}\end{array}\right]^T, j=1,2,\dots,R$. Of these $R$ blocks, $M$ of them are modulated by $N$-point IDFT and the remaining $R-M$ blocks are modulated by $N$-point DFT. Without loss of generality, let us assume that the first $M$ blocks are modulated by $N$-point IDFT. Then a CP of length $l_{cp}$ is added to each block, where $l_{cp}$ is not less than the maximum of the overall relative timing errors of the signals arrived at the destination node from all the relay nodes. The resulting $R$ OFDM symbols denoted by $\mathbf{\bar{a}_1},\mathbf{\bar{a}_2},\dots,\mathbf{\bar{a}_R}$ each consisting of $L_s=N+l_{cp}$ complex numbers are broadcasted to the $R$ relays using a fraction $\pi_1$ of the total average power $P$.

\subsubsection{Processing at the relay nodes}
\label{subsubsec4_2_2}
~\\
If the channel fade gains are assumed to be constant for $2R$ OFDM symbol intervals, the received signals at the $i$-th relay during the $j$-th OFDM symbol duration is given by
$$
\mathbf{r_{i,j}}=f_i\mathbf{\bar{a}_j}+\mathbf{\bar{v}_{i,j}}
$$

\noindent where, $\mathbf{\bar{v}_{i,j}}$ is the AWGN at the $i$-th relay node during the $j$-th OFDM symbol duration. The relay nodes process and transmit the received noisy signals as shown in Table \ref{table_proposed} using a fraction $\pi_2$ of total power $P$. Note from Table \ref{table_proposed} that time reversal is done during the last $R-M$ OFDM symbol durations. We would like to emphasize that in general time reversal could be implemented in any $R-M$ of the total $R$ OFDM symbol durations. Now, $\mathbf{t_{i,j}}\in\left\{\mathbf{0},\pm \mathbf{r_{i,j}}, j=1,\dots,R\right\}$ with the constraint that the $i$-th relay should not be allowed to transmit the following:

$$
\begin{array}{l}
\left\{\pm\mathbf{r_{i,j}}^*, j=1,\dots,M\right\}
\cup\left\{\pm\zeta(\mathbf{r_{i,j}}),j=1,\dots,M\right\}\\
\cup\left\{\pm\mathbf{r_{i,j}},j=M+1,\dots,R\right\}
\cup\left\{\pm\zeta(\mathbf{r_{i,j}}^*), j=M+1,\dots,R\right\}.
\end{array}
$$

\begin{remark}
If the $i$-th relay is permitted to transmit elements belonging to the above set, then after CP removal and DFT transformation at the destination node, we would end up with the following vectors corresponding to each of the four subsets in the above set respectively:
$$
\begin{array}{l}
\pm\mathrm{DFT}((\mathrm{IDFT}(\mathbf{a_j}))^*)=\pm\mathrm{DFT}(\mathrm{DFT}(\mathbf{a_j}^*)),\ j=1,\dots,M\\
\pm\mathrm{DFT}(\zeta(\mathrm{IDFT}(\mathbf{a_j}))),\ j=1,\dots,M\\
\pm\mathrm{DFT}(\mathrm{DFT}(\mathbf{a_j})),\ j=M+1,\dots,R\\
\pm\mathrm{DFT}(\zeta((\mathrm{DFT}(\mathbf{a_j}))^*)),\ j=M+1,\dots,R
\end{array}
$$

\noindent from any of which it is not possible to recover any of $\pm\mathbf{a_j},\pm\mathbf{a_j}^*,j=1,2,\dots,R$. However, if the destination node is allowed to apply DFT to some of the received OFDM symbols and IDFT to the remaining OFDM symbols, then possibly the above restrictions can be removed, which is a scope for further work.
\end{remark}

\begin{table}
\caption{Proposed transmission scheme}
\label{table_proposed}
{\footnotesize
\begin{center}
\begin{tabular}{|c|c|c|c|c|c|c|}
\hline
OFDM Symbol & $U_1$ & $\dots$ & $U_{M}$ & $U_{M+1}$ & $\dots$ & $U_R$\\
\hline
$1$ & $\mathbf{t_{1,1}}$ & $\dots$ & $\mathbf{t_{M,1}}$ & $\mathbf{t_{M+1,1}}^*$ & $\dots$ & $\mathbf{t_{R,1}}^*$\\
\hline
$\vdots$ & $\vdots$ & $\vdots$ & $\vdots$ & $\vdots$ & $\vdots$ & $\vdots$\\
\hline
$M$ & $\mathbf{t_{1,M}}$ & $\dots$ & $\mathbf{t_{M,M}}$ & $\mathbf{t_{M+1,M}}^*$ & $\dots$ & $\mathbf{t_{R,M}}^*$\\
\hline
$M+1$ & $\zeta(\mathbf{t_{1,M+1}})$ & $\dots$ & $\zeta(\mathbf{t_{M,M+1}})$ & $\zeta(\mathbf{t_{M+1,M+1}}^*)$ & $\dots$ & $\zeta(\mathbf{t_{R,M+1}}^*)$\\
\hline
$\vdots$ & $\vdots$ & $\vdots$ & $\vdots$ & $\vdots$ & $\vdots$ & $\vdots$\\
\hline
$R$ & $\zeta(\mathbf{t_{1,R}})$ & $\dots$ & $\zeta(\mathbf{t_{M,R}})$ & $\zeta(\mathbf{t_{M,R}}^*)$ & $\dots$ & $\zeta(\mathbf{t_{R,R}}^*)$\\
\hline
\end{tabular}
\end{center}
}
\end{table}

\subsubsection{Decoding at the destination}
\label{subsubsec4_2_3}
~\\
The destination removes the CP for the first $M$ OFDM symbols and implements the following for the remaining OFDM symbols:

\begin{enumerate}
\item Remove the CP to get a $N$-point vector
\item Shift the last $l_{cp}$ samples of the $N$-point vector as the first $l_{cp}$ samples.
\end{enumerate}

DFT is then applied on the resulting $R$ vectors. Let the received signals for $R$ consecutive OFDM blocks after CP removal and DFT transformation be denoted by\\ $\mathbf{y_j}=\left[\begin{array}{cccc}y_{0,j} & y_{1,j} & \dots & y_{N-1,j}\end{array}\right]^T, j=1,2,\dots,R$. Let $\mathbf{w_i}=(w_{k,i}), i=1,\dots,R$ represent the AWGN at the destination node and let $\mathbf{v_{i,j}}$ denote the DFT of $\mathbf{\bar{v}_{i,j}}$. Let \\$\mathbf{z_k}=\left[\begin{array}{cccc}z_{k,1} & z_{k,2} & \dots & z_{k,R}\end{array}\right]^T, k=0,1,\dots,N-1$.

Now using \eqref{eqn_identities}, we get in each sub carrier $k,0\leq k\leq N-1$:

\begin{equation}
\label{eqn_sys_model}
\mathbf{y_{k}}=\left[\begin{array}{cccc}y_{k,1} & y_{k,2} & \dots & y_{k,R}\end{array}\right]^T=\sqrt{\frac{\pi_1\pi_2P^2}{\pi_1P+1}}\mathbf{X_kh_k}+\mathbf{n_k}
\end{equation}

\noindent where,
\begin{equation}
\label{eqn_conj_linear}
\mathbf{X_k}=\left[\begin{array}{cccccc}\mathbf{B_1z_k} & \dots & \mathbf{B_{M}z_k} & \mathbf{B_{M+1}}\mathbf{z_k}^* & \dots \mathbf{B_{R}}\mathbf{z_k}^* \end{array}\right]
\end{equation}

\noindent for some square real matrices $\mathbf{B_i}, i=1,\dots,R$ having the property that any row of $\mathbf{B_i}$ has only one nonzero entry. If $u_k^{\tau_i}=e^{-\frac{i2\pi k\tau_i}{N}}$, then
\begin{equation}
\label{eqn_channel}
\mathbf{h_k}=\left[\begin{array}{ccccccc} f_1g_1 & u_k^{\tau_2}f_2g_2 & \dots &
u_k^{\tau_{M}}f_{M}g_{M} &
u_k^{\tau_{M+1}}f_{M+1}^*g_{M+1} & \dots &
u_k^{\tau_R}f_R^*g_R\end{array}\right]^T
\end{equation}
\noindent is the equivalent channel matrix for the $k$-th sub carrier. The equivalent noise vector $\mathbf{n_k}$ is given by:

$$
\begin{array}{rl}
\mathbf{n_k}=&\sqrt{\frac{\pi_2P}{\pi_1P+1}}\left[\begin{array}{c}\beta_1\sum_{i=1}^{R}sgn(\mathbf{t_{i,1}})\mathbf{\hat{v}_{i,1}}(k)g_iu_k^{\tau_i}\\
\beta_2\sum_{i=1}^{R}sgn(\mathbf{t_{i,2}})\mathbf{\hat{v}_{i,2}}(k)g_iu_k^{\tau_i}\\
\vdots\\
\beta_R\sum_{i=1}^{R}sgn(\mathbf{t_{i,R}})\mathbf{\hat{v}_{i,R}}(k)g_iu_k^{\tau_i}
\end{array}\right]+\left[\begin{array}{c}w_{k,1}\\ w_{k,2}\\ \dots\\  w_{k,R}\end{array}\right]
\end{array}
$$

\noindent where, $sgn(\mathbf{t_{i,j}})=\left\{\begin{array}{l} 1~~\mathrm{if}~ \mathbf{t_{i,j}}\in\left\{\mathbf{r_{i,j}}, j=1,\dots,R\right\}\\
-1~~\mathrm{if}~\mathbf{t_{i,j}}\in\left\{-\mathbf{r_{i,j}}, j=1,\dots,R\right\}\\
0~~\mathrm{if}~\mathbf{t_{i,j}}=\mathbf{0}\end{array}\right.$ and\\ $\mathbf{\hat{v}_{i,m}}=\left\{\begin{array}{l} \pm\mathbf{v_{i,j}}~\mathrm{if}~i\leq M~\mathrm{and}~\mathbf{t_{i,m}}=\pm\mathbf{r_{i,j}}\\
\pm\mathbf{v_{i,j}}^*~\mathrm{if}~i>M~\mathrm{and}~\mathbf{t_{i,m}}=\pm\mathbf{r_{i,j}}\end{array}\right.$. The $\beta_i$'s are simply scaling factors to account for the correct noise variance due to some zeros in the transmission.

ML decoding of $\mathbf{X_k}$ can be done from \eqref{eqn_sys_model} by choosing that codeword which minimizes $\parallel\Omega^{-\frac{1}{2}}(\mathbf{y_k}-\mathbf{X_k}\mathbf{h_k})\parallel_F^2$, where $\Omega$ is the covariance matrix of $\mathbf{n_k}$. Essentially, the proposed transmission scheme implements a space time code having a special structure in each sub carrier.

\subsubsection{Full diversity four group decodable DSTBCs}
\label{subsubsec4_2_4}
~\\
In this subsection, we analyze the structure of the STBC required for implementing in the proposed transmission scheme. Then we observe that the DSTBCs constructed in Section \ref{sec3} have this structure and hence are applicable in this setting as well. Note from \eqref{eqn_conj_linear} that the conjugate linearity property is required. But conjugate linearity alone is not enough for a space time code to qualify for implementation in the proposed transmission scheme. Note from Table \ref{table_proposed} that time reversal is implemented for certain OFDM symbol durations by all the relay nodes. In other words if one relay node implements time reversal during a particular OFDM symbol duration, then all the other relay nodes should necessarily implement time reversal during that OFDM symbol duration. Observe that this is a property connected with the row structure of a space time code. We now provide a set of sufficient conditions that are required on the row structure of conjugate linear space time codes. First let us partition the complex symbols appearing in the $i$-th row into two sets- one set $P_i$ containing those complex symbols which appear without conjugation and another set $P_i^c$ which contains those complex symbols which appear with conjugation in the $i$-th row. Any conjugate linear STBC satisfying the following sufficient conditions can be implemented in the proposed OFDM based transmission scheme described in the previous subsection.

\begin{equation}
\label{eqn_row_property}
\begin{array}{c}
P_i\cap P_i^c=\phi,~\forall~i=1,\dots,R\\
|P_i|=|P_i^c|,~ \forall~ i=1,\dots,R\\
P_i\cap P_j\in\left\{\phi, P_i, P_j\right\},~\forall~ i\neq j.
\end{array}
\end{equation}

To understand what happens if the above condition is not met, let us see an example of a conjugate linear STBC which cannot be employed in the proposed transmission scheme.

\begin{eg}
Consider the conjugate linear STBC given by
$$
\left[\begin{array}{cccc}
z_{k,1} & z_{k,2} & -z_{k,3}^* & -z_{k,4}^*\\
z_{k,2} & z_{k,3} & -z_{k,4}^* & -z_{k,1}^*\\
z_{k,3} & z_{k,4} & z_{k,1}^* & z_{k,2}^*\\
z_{k,4} & z_{k,1} & z_{k,2}^* & z_{k,3}^*
\end{array}\right]
$$
\noindent for which $P_1=P_3^c=\left\{z_{k,1},z_{k,2}\right\}$, $P_1^c=P_3=\left\{z_{k,3},z_{k,4}\right\}$, $P_2=P_4^c=\left\{z_{k,2},z_{k,3}\right\}$, $P_2^c=P_4=\left\{z_{k,4},z_{k,1}\right\}$. It can be checked that there is no assignment of time reversal OFDM symbol durations together with an appropriate choice of $M$ and relay node processing such that the above conjugate linear STBC form is obtained in every sub carrier at the destination node. This is because the conditions in \eqref{eqn_row_property} are not met by this conjugate linear STBC.
\end{eg}

For the case of the Alamouti code, $P_1=P_2^c=\left\{z_{k,1}\right\}$, $P_2=P_1^c=\left\{z_{k,2}\right\}$ and hence it satisfies the conditions in \eqref{eqn_row_property}.

It is easy to observe that the four group ML decodable DSTBCs constructed in Section \ref{sec3} satisfy all the required conditions as stated in \eqref{eqn_row_property} and are thus suitable for application in the proposed transmission scheme. This is illustrated using the following two examples.

\begin{eg}
\label{eg_asynchronous_4relay}

Let us consider $R=4$ and the DSTBC from extended Clifford algebra $\mathbb{A}_2^2$ for this case has the following structure
$$
\left[\begin{array}{cccc}
z_{k,1} & z_{k,2} & -z_{k,3}^* & -z_{k,4}^*\\
z_{k,2} & z_{k,1} & -z_{k,4}^* & -z_{k,3}^*\\
z_{k,3} & z_{k,4} & z_{k,1}^* & z_{k,2}^*\\
z_{k,4} & z_{k,3} & z_{k,2}^* & z_{k,1}^*
\end{array}\right]
$$
\noindent for which $M=2$, $P_1=P_2=P_3^c=P_4^c=\left\{z_{k,1},z_{k,2}\right\}$ and $P_3=P_4=P_1^c=P_2^c=\left\{z_{k,3},z_{k,4}\right\}$. To arrive at the above structure in every sub carrier, encoding and processing at the relays are done as follows:
$\mathbf{\bar{a}_1}=\mathrm{IDFT}(\mathbf{a_1})$, $\mathbf{\bar{a}_2}=\mathrm{IDFT}(\mathbf{a_2})$, $\mathbf{\bar{a}_3}=\mathrm{DFT}(\mathbf{a_3})$ and $\mathbf{\bar{a}_4}=\mathrm{DFT}(\mathbf{a_4})$.

\begin{table}[h]
\caption{Transmission scheme for $4$ relays}
\label{table_4relay}
\begin{center}
\begin{tabular}{|c|c|c|c|c|}
\hline
OFDM & $U_1$ & $U_2$ & $U_3$ & $U_4$\\
Symbol & & & &\\
\hline
$1$ & $\mathbf{r_{1,1}}$ & $\mathbf{r_{2,2}}$ & $-\mathbf{r_{3,3}}^*$ & $-\mathbf{r_{4,4}}^*$\\
\hline
$2$ & $\mathbf{r_{1,2}}$ & $\mathbf{r_{2,1}}$ &
$-\mathbf{r_{3,4}}^*$ & $-\mathbf{r_{4,3}}^*$\\
\hline
$3$ & $\zeta(\mathbf{r_{1,3}})$ & $\zeta(\mathbf{r_{2,4}})$ &
$\zeta(\mathbf{r_{3,1}}^*)$ & $\zeta(\mathbf{r_{4,2}}^*)$\\
\hline
$4$ & $\zeta(\mathbf{r_{1,4}})$ & $\zeta(\mathbf{r_{2,3}})$ &
$-\zeta(\mathbf{r_{3,2}}^*)$ & $-\zeta(\mathbf{r_{4,1}}^*)$\\
\hline
\end{tabular}
\end{center}
\end{table}

As discussed in Section \ref{sec3}, this DSTBC is single complex symbol decodable and achieves full diversity for appropriately chosen signals sets  .
\end{eg}

\begin{eg}
\label{eg_5relay}
Let us take $R=5$, for which the DSTBC is obtained by taking a DSTBC from PCIOD for $6$ relays and dropping one column. It is given by
$$
\left[\begin{array}{ccccc}
z_{k,1} & -z_{k,2}^* & 0 & 0 & 0\\
z_{k,2} & z_{k,1}^* & 0 & 0 & 0\\
0 & 0 & z_{k,3} & -z_{k,4}^* & 0\\
0 & 0 & z_{k,4} & z_{k,3}^* & 0\\
0 & 0 & 0 & 0 & z_{k_5}\\
0 & 0 & 0 & 0 & z_{k,6}
\end{array}\right]
$$
\noindent for which $P_1=P_2^c=\left\{z_{k,1}\right\}$, $P_2=P_1^c=\left\{z_{k,2}\right\}$, $P_3=P_4^c=\left\{z_{k,3}\right\}$, $P_4=P_3^c=\left\{z_{k,4}\right\}$, $P_5=\left\{z_{k,5}\right\}$, $P_6=\left\{z_{k,6}\right\}$ and $P_5^c=P_6^c=\phi$. At the source, we choose $\mathbf{\bar{a}_1}=\mathrm{IDFT}(\mathbf{a_1})$, $\mathbf{\bar{a}_2}=\mathrm{DFT}(\mathbf{a_2})$,
$\mathbf{\bar{a}_3}=\mathrm{IDFT}(\mathbf{a_3})$,
$\mathbf{\bar{a}_4}=\mathrm{DFT}(\mathbf{a_4})$,
$\mathbf{\bar{a}_5}=\mathrm{IDFT}(\mathbf{a_5})$ and
$\mathbf{\bar{a}_6}=\mathrm{DFT}(\mathbf{a_6})$. The $5$ relays process the received OFDM symbols as shown in Table \ref{table_5relay}.

\begin{table}[h]
\caption{Transmission scheme for $5$ relays}
\label{table_5relay}
\begin{center}
{\footnotesize
\begin{tabular}{|c|c|c|c|c|c|}
\hline
OFDM & $U_1$ & $U_2$ & $U_3$ & $U_4$ & $U_5$\\
Symbol & & & & &\\
\hline
$1$ & $\mathbf{r_{1,1}}$ & $-\mathbf{r_{2,2}}^*$ & $\mathbf{0}$ & $-\mathbf{0}$ & $\mathbf{0}$\\
\hline
$2$ & $\zeta(\mathbf{r_{1,2}})$ & $\zeta(\mathbf{r_{2,1}}^*)$ &
$-\mathbf{0}$ & $-\mathbf{0}$ & $\mathbf{0}$\\
\hline
$3$ & $\mathbf{0}$ & $\mathbf{0}$ & $\mathbf{r_{3,3}}$ & $-\mathbf{r_{4,4}}^*$ & $\mathbf{0}$\\
\hline
$4$ & $\mathbf{0}$ & $\mathbf{0}$ &
$-\zeta(\mathbf{r_{3,2}})$ & $\zeta(\mathbf{r_{4,1}}^*)$ & $\mathbf{0}$\\
\hline
$5$ & $\mathbf{0}$ & $\mathbf{0}$ & $\mathbf{0}$ & $\mathbf{0}$ & $\mathbf{r_{5,5}}$\\
\hline
$6$ & $\mathbf{0}$ & $\mathbf{0}$ & $\mathbf{0}$ & $\mathbf{0}$ &
$-\zeta(\mathbf{r_{5,6}})$\\
\hline
\end{tabular}
}
\end{center}
\end{table}

This code is $3$ real symbol decodable and achieves full diversity for appropriately chosen signal sets.
\end{eg}

Example \ref{eg_5relay} illustrates how the proposed transmission scheme can be employed for odd number of relays as well. Note that the ML decoding complexity of the proposed codes for asynchronous relay networks is significantly less compared to all other distributed space time codes for asynchronous cooperative diversity known in the literature.

Note that the full diversity, $2$-group ML decodable DSTBCs in \cite{KiR2} and the full diversity, $1$-group ML decodable DSTBCs in \cite{EOK,OgH} also satisfy the conditions in \eqref{eqn_row_property} and are applicable in the proposed OFDM based transmission scheme.

\subsection{Transmission Scheme for Noncoherent Asynchronous Relay Networks}
\label{subsec4_3}

In this subsection, it is shown how differential encoding can be combined with the proposed transmission scheme described in Subsection \ref{subsec4_2} and the distributed differential space time codes for noncoherent synchronous relay networks in \cite{RaR6} are proposed for application in this setting.

For the transmission scheme described in Subsection \ref{subsec4_2}, at the end of one transmission frame,  we have in the $k$-th sub carrier $\mathbf{y_k}=\sqrt{\frac{\pi_1\pi_2P^2}{\pi_1P+1}}\mathbf{X_k}\mathbf{h_k}+\mathbf{n_k}$. Note that the channel matrix $\mathbf{h_k}$ depends on $f_i,g_i,\tau_i, i=1,\dots,R$. Thus the destination node needs to have the knowledge of all these values in order to perform ML decoding.

Now using differential encoding ideas which were proposed in \cite{KiR,OgH2,JiJ} for non-coherent communication in synchronous relay networks, we propose to combine them with the proposed asynchronous transmission scheme. Supposing the channel remains approximately constant for two transmission frames, then differential encoding can be done at the source node in each sub carrier $0\leq k\leq N-1$ as follows:
$$
\mathbf{a_k^0}=\left[\begin{array}{cccc}\sqrt{R} & 0 & \dots & 0\end{array}\right]^T,~ \mathbf{a_k^t}=\frac{1}{b_t-1}\mathbf{C_t}\mathbf{a_k^{t-1}}, \mathbf{C_t}\in\mathscr{C}
$$
where, $\mathbf{a_k^{i}}$ denotes the vector of complex symbols transmitted by the source during the $i$-th transmission frame and $\mathscr{C}$ is the codebook used by the source which consists of scaled unitary matrices $\mathbf{C}^H\mathbf{C_t}=b_t^2\mathbf{I}$ such that $\mathrm{E}[b_t^2]=1$. If \mbox{$\mathbf{C}\mathbf{B_i}=\mathbf{B_i}\mathbf{C}, i=1,\dots,M$} and \mbox{$\mathbf{C}\mathbf{B_i}=\mathbf{B_i}\mathbf{C}^*, i=M+1,\dots,R$} for all $\mathbf{C\in\mathscr{C}}$, then we can write
\begin{equation}
\mathbf{y_k^t}=\frac{1}{b_{t-1}}\mathbf{C_t}\mathbf{y_k^{t-1}}+(\mathbf{n_k^t}-\frac{1}{b_t-1}\mathbf{C_{t}}\mathbf{n_k^{t-1}})
\end{equation}
from which $\mathbf{C_t}$ can be decoded as \mbox{$\mathbf{\hat{C}_t}=\arg\min_{\mathbf{C_t\in\mathscr{C}}}\parallel\mathbf{y_k^t}-\frac{1}{b_{t-1}}\mathbf{C_t}\mathbf{y_k^{t-1}}\parallel_F^2$} in each sub carrier $0\leq k\leq N-1$.

Note that this decoder does not require the knowledge of $f_i,g_i,\tau_i, i=1,\dots R$ at the destination although this transmission strategy assumes the knowledge of the maximum of the $\tau$'s since it is needed to decide the length of CP. It turns out that the four group decodable distributed differential space time codes constructed in \cite{RaR6} for synchronous relay networks with power of two number of relays meet all the requirements for use in the proposed transmission scheme as well. Let us see an example to illustrate this fact.

\begin{eg}
\label{eg_4relay_diff}
Let $R=4$. The codebook at the source is given by\\ $\mathscr{C}=\left\{\sqrt{\frac{1}{4}}\left[\begin{array}{cccc}
z_1 & z_2 & -z_3^* & -z_4^*\\
z_2 & z_1 & -z_4^* & -z_3^*\\
z_3 & z_4 & z_1^* & z_2^*\\
z_4 & z_3 & z_2^* & z_1^*
 \end{array}\right]\right\}$ where $\left\{z_{1I},z_{2I}\right\},\left\{z_{1Q},z_{2Q}\right\},\left\{z_{3I},z_{4I}\right\},\left\{z_{3Q},z_{4Q}\right\}\in\mathbb{S}$ and \mbox{$\mathbb{S}=\left\{\left[\begin{array}{c}\frac{1}{\sqrt{3}}\\0\end{array}\right],\left[\begin{array}{c}-\frac{1}{\sqrt{3}}\\0\end{array}\right],\left[\begin{array}{c}0\\ \sqrt{\frac{5}{3}}\end{array}\right],\left[\begin{array}{c}0\\ -\sqrt{\frac{5}{3}}\end{array}\right] \right\}$}. Differential encoding is done at the source node for each sub carrier $0\leq k\leq N-1$ as follows:
$$
\mathbf{a_k^0}=\left[\begin{array}{cccc}\sqrt{R} & 0 & \dots & 0\end{array}\right]^T,~ \mathbf{a_k^t}=\frac{1}{b_t-1}\mathbf{C_t}\mathbf{a_k^{t-1}}, \mathbf{C_t}\in\mathscr{C}.
$$
Once we get $\mathbf{a_k^t},k=0,\dots,N-1$ from the above equation, the $N$ length vectors $\mathbf{z_i}, i=1,\dots,R$ can be obtained. Then IDFT/DFT is applied on these vectors and broadcasted to the relay nodes as shown below: $\mathbf{\bar{a}_1}=\mathrm{IDFT}(\mathbf{a_1})$, $\mathbf{\bar{a}_2}=\mathrm{IDFT}(\mathbf{a_2})$, $\mathbf{\bar{a}_3}=\mathrm{DFT}(\mathbf{a_3})$ and $\mathbf{\bar{a}_4}=\mathrm{DFT}(\mathbf{a_4})$. The relay nodes process the received OFDM symbols as given below:

$$
\begin{array}{|c|c|c|c|c|}
\hline
\mathrm{OFDM} & U_1 & U_2 & U_3 & U_4\\
\mathrm{Symbol} & & & &\\
\hline
1 & \mathbf{r_{1,1}} & \mathbf{r_{2,2}} & -\mathbf{r_{3,3}}^* & -\mathbf{r_{4,4}}^*\\
\hline
2 & \mathbf{r_{1,2}} & \mathbf{r_{2,1}} &
-\mathbf{r_{3,4}}^* & -\mathbf{r_{4,3}}^*\\
\hline
3 & \zeta(\mathbf{r_{1,3}}) & \zeta(\mathbf{r_{2,4}}) &
\zeta(\mathbf{r_{3,1}}^*) & \zeta(\mathbf{r_{4,2}}^*)\\
\hline
4 & \zeta(\mathbf{r_{1,4}}) & \zeta(\mathbf{r_{2,3}}) &
-\zeta(\mathbf{r_{3,2}}^*) & -\zeta(\mathbf{r_{4,1}}^*)\\
\hline
\end{array}
$$

for which $M=2$, $\mathbf{B_1}=\mathbf{I_4}$, \mbox{$\mathbf{B_2}=\left[\begin{array}{cccc}0 & 1 & 0 & 0\\
1 & 0 & 0 & 0\\
0 & 0 & 0 & 1\\
0 & 0 & 1 & 0\end{array}\right]$}, $\mathbf{B_3}=\left[\begin{array}{ccrr}0 & 0 & -1 & 0\\
0 & 0 & 0 & -1\\
1 & 0 & 0 & 0\\
0 & 1 & 0 & 0\end{array}\right]$ and\\
\mbox{$\mathbf{B_4}=\left[\begin{array}{ccrr}0 & 0 & 0 & -1\\
0 & 0 & -1 & 0\\
0 & 1 & 0 & 0\\
1 & 0 & 0 & 0\end{array}\right]$}. It has been proved in \cite{RaR6} that \mbox{$\mathbf{C}\mathbf{B_i}=\mathbf{B_i}\mathbf{C}, i=1,2$} and\\ \mbox{$\mathbf{C}\mathbf{B_i}=\mathbf{B_i}\mathbf{C}^*, i=3,4$} for all $\mathbf{C\in\mathscr{C}}$. At the destination node, decoding for $\left\{z_{1I},z_{2I}\right\}$, $\left\{z_{1Q},z_{2Q}\right\}$, $\left\{z_{3I},z_{4I}\right\}$ and $\left\{z_{3Q},z_{4Q}\right\}$ can be done separately in every sub carrier due to the four group decodable structure of $\mathscr{C}$.
\end{eg}

\section{Simulation Results}
\label{sec5}

In this section, we study the error performance of the DSTBCs proposed in this paper using simulations.  We consider both the synchronous case and the asynchronous case.

For the synchronous case, we compare the performance of the newly proposed DSTBCs from extended Clifford algebras and PCIODs with the DSTBCs from field extensions \cite{OgH,EOK} for a $4$ relay network. The PCIOD taken for simulations is given by
$\left[\begin{array}{rrrr}
z_1 & -z_2^* & 0 & 0\\
z_2 & z_1^* & 0 & 0\\
0 & 0 & z_3 & -z_4^*\\
0 & 0 & z_4 & z_3^*
\end{array}\right]$ where, $\left\{z_{1I},z_{3I}\right\}$, $\left\{z_{1Q},z_{3Q}\right\}$, $\left\{z_{2I},z_{4I}\right\}$, $\left\{z_{2Q},z_{4Q}\right\}$ are allowed to take values from QAM constellation rotated by $31.718^\circ$. The DSTBC from extended Clifford algebras (ECA) is obtained from $\mathbb{A}_2^2$ and is given by $\left[\begin{array}{ccrr}
z_1 & z_2 & -z_3^* & -z_4^*\\
z_2 & z_1 & -z_4^* & -z_3^*\\
z_3 & z_4 & z_1^* & z_2^*\\
z_4 & z_3 & z_2^* & z_1^*
\end{array}\right]$ where, $\left\{z_{1I},z_{2I}\right\}$, $\left\{z_{1Q},z_{2Q}\right\}$, $\left\{z_{3I},z_{3Q}\right\}$, $\left\{z_{4I},z_{4Q}\right\}$ are allowed to take values from QAM constellation rotated by $166.71^\circ$. The DSTBC from field extensions \cite{OgH,EOK} is given by $\left[\begin{array}{cccc}
z_1 & iz_4 & iz_3 & iz_2\\
z_2 & z_1 & iz_4 & iz_3\\
z_3 & z_2 & z_1 & iz_4\\
z_4 & z_3 & z_2 & z_1
\end{array}\right]$ where $z_i,~i=1,2,3,4$ are allowed to take values from regular QAM constellation.

\begin{figure}
\centering
\includegraphics[width=5 in]{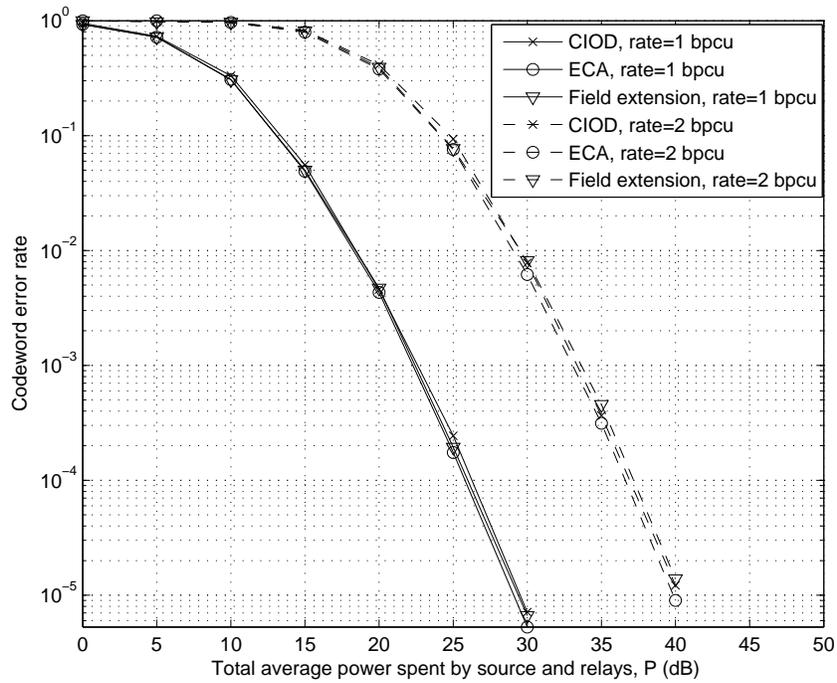}
\caption{Performance comparison of DSTBCs from PCIOD, ECA and field extension}
\label{fig_simulation_syn}
\end{figure}

Fig.\ref{fig_simulation_syn} shows the codeword error rate performance of the proposed DSTBCs ($4$ relays) in comparison to those from field extensions \cite{OgH,EOK} for transmission rates of $1$ bit per channel use (bpcu) and $2$ bpcu. We observe that the error performance of the proposed codes are very similar to the $1$-group ML decodable DSTBC from field extensions \cite{OgH,EOK}. Thus the proposed codes enjoy a good error performance along with reduced ML decoding complexity.

\begin{figure}[h]
\centering
\includegraphics[width=5 in]{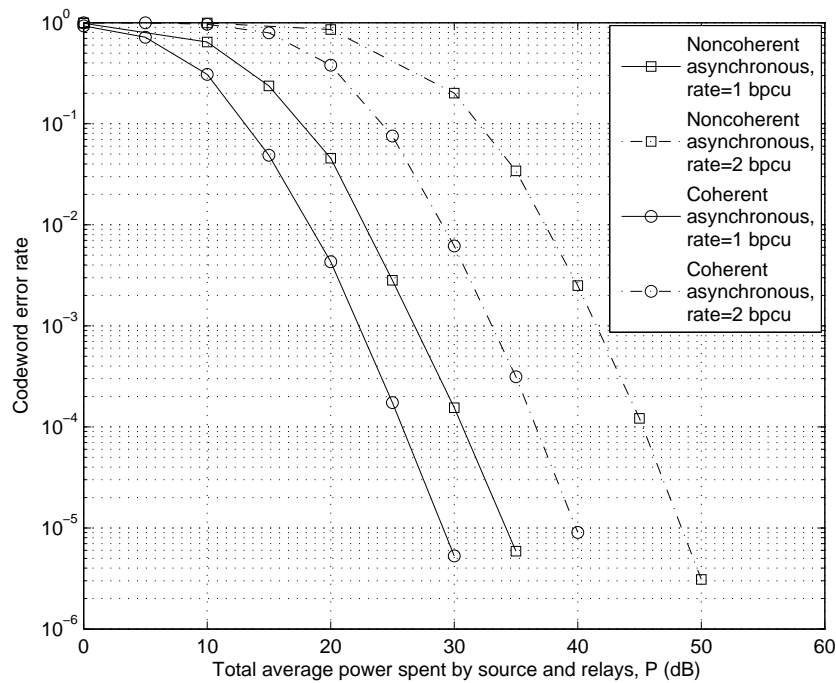}
\caption{Error performance for an asynchronous $4$ relay network with and without channel knowledge}
\label{fig_simulation_asyn}
\end{figure}
For the asynchronous case, we take $R=4$, $N=64$ and the length of CP as $16$. The delay $\tau_i$ at each relay is chosen randomly between $0$ to $15$ with uniform distribution. Two cases are considered for simulation: (i) with channel knowledge at the destination and (ii) without channel knowledge at any node. When channel knowledge is available at the destination, the processing at relay nodes is done as described in Example \ref{eg_asynchronous_4relay} and $166.71^\circ$ rotated QPSK is used as the signal set. Coherent detection is done at the destination in every sub carrier. For the case of no channel knowledge, differential encoding is done at the source as described in Example \ref{eg_4relay_diff} and a modified signal set (as explained in \cite{RaR6} for scaled unitary codewords)  is employed. Also, differential detection is done at the destination in every sub carrier. Simulations are done for transmission rates (neglecting rate loss due to CP) of $1$ bpcu and $2$ bpcu in both the cases.

The error performance curves for both the cases are shown in Fig. \ref{fig_simulation_asyn}. It can be observed that the error performance for the no channel knowledge case performs approximately $5$ dB worse and $8$ dB worse compared to that with channel knowledge for transmission rates of $1$ bpcu and $2$ bpcu respectively. This is partly due to the differential transmission/reception technique (which increases the effective noise seen by destination) and also in part, because of the change in signal set from rotated QAM to some other signal set \cite{RaR6} for scaled unitary codeword matrices. The change in signal set for the sake of scaled unitary codeword matrices results in a reduction of the coding gain.

\section{Discussion}
\label{sec6}

In this paper, the problem of optimal rate-ML decoding complexity within the framework of multi-group decodable STBCs was posed. Then an algebraic framework for studying CUW STBCs was introduced using which the optimal rate-ML decoding complexity tradeoff of CUW STBCs was obtained and several optimal code constructions were also provided. Then the paper focused on constructing multi-group decodable DSTBCs and three new classes of four group decodable DSTBCs were constructed. The OFDM based transmission scheme in \cite{LiX} was extended to a more general transmission scheme for arbitrary number of relays that can achieve full cooperative diversity in the presence of timing errors at the relay nodes. It was then pointed out that the four group decodable DSTBCs constructed in this paper can be applied in the proposed transmission scheme for any number of relay nodes. A drawback of the proposed transmission scheme is that it requires a large coherence interval spanning over multiple OFDM symbol durations. Moreover there is a rate loss due to the use of CP, however this loss can be made negligible by choosing a large enough $N$. Finally, it was shown how differential encoding at the source node can be combined with this OFDM based transmission strategy to arrive at a new transmission strategy than can achieve full cooperative diversity in asynchronous relay networks with no channel knowledge as well as no timing error knowledge. The distributed differential STBCs in \cite{RaR6} were then proposed for application in this setting for power of two number of relays.

Some of the interesting directions for further work are as listed below:

\begin{enumerate}
\item The CUW STBCs are based on sufficient conditions for $g$-group ML decodability. An algebraic framework for $g$-group ML decodable STBCs based on the necessary and sufficient conditions and the optimal rate-ML decoding complexity tradeoff of general $g$-group ML decodable STBCs is an important open problem.
\item How to construct $g$-group decodable DSTBCs for $g\neq 4$? In particular, constructing single symbol decodable DSTBCs for the synchronous as well as asynchronous cases is worth investigating. Some initial results in this direction have been reported in \cite{YiK,HaR}.
\item In the results pertaining to asynchronous relay networks, we have assumed that there are no frequency offsets at the relay nodes. Extending this work to asynchronous relay networks with frequency offsets is an interesting direction for further work. This problem has been addressed in \cite{LiX2} for the case of two relay nodes.
\item In this work, we have constructed DSTBCs with low ML decoding complexity mainly for the two phase amplify and forward based transmission protocols \cite{JiH,EVAK,RaR1}. Constructing low ML decoding complexity codes for the other cooperative diversity protocols in the literature is also an interesting problem.
\end{enumerate}

\section*{Acknowledgement}
The authors thank Prof. Hamid Jafarkhani and Prof. Xiang Gen Xia for sending us preprints of their recent works \cite{JiJ1,LiX,GuX,JiJ,LiX2}. The authors are grateful to the anonymous reviewers for their valuable comments which greatly helped to improve the organization and presentation of the contents.

\bibliographystyle{IEEEtran}


\end{document}

%% file: stbc_lstd.pstex_t
\begin{picture}(0,0)%
\includegraphics{stbc_lstd.pstex}%
\end{picture}%
\setlength{\unitlength}{3947sp}%
\begingroup\makeatletter\ifx\SetFigFont\undefined%
\gdef\SetFigFont#1#2#3#4#5{%
  \reset@font\fontsize{#1}{#2pt}%
  \fontfamily{#3}\fontseries{#4}\fontshape{#5}%
  \selectfont}%
\fi\endgroup%
\begin{picture}(3390,3764)(4793,-3068)
\put(6466,-1891){\makebox(0,0)[lb]{\smash{{\SetFigFont{14}{16.8}{\familydefault}{\mddefault}{\updefault}{$\mathscr{C}$}%
}}}}
\put(5416,-1051){\makebox(0,0)[lb]{\smash{{\SetFigFont{14}{16.8}{\familydefault}{\mddefault}{\updefault}{$\langle\mathscr{C}\rangle=\langle \mathbf{A_1},\dots,\mathbf{A_K}\rangle$}%
}}}}
\put(6316,194){\makebox(0,0)[lb]{\smash{{\SetFigFont{14}{16.8}{\familydefault}{\mddefault}{\updefault}{$\mathbb{C}^{T\times N_T}$}%
}}}}
\end{picture}%

%% file: g_enc.pstex_t
\begin{picture}(0,0)%
\includegraphics{g_enc.pstex}%
\end{picture}%
\setlength{\unitlength}{3947sp}%
\begingroup\makeatletter\ifx\SetFigFont\undefined%
\gdef\SetFigFont#1#2#3#4#5{%
  \reset@font\fontsize{#1}{#2pt}%
  \fontfamily{#3}\fontseries{#4}\fontshape{#5}%
  \selectfont}%
\fi\endgroup%
\begin{picture}(7010,1811)(153,-4523)
\put(1376,-3061){\makebox(0,0)[lb]{\smash{{\SetFigFont{12}{14.4}{\rmdefault}{\mddefault}{\updefault}{Bits--$>\mathbf{s_1}$}%
}}}}
\put(1364,-3486){\makebox(0,0)[lb]{\smash{{\SetFigFont{12}{14.4}{\rmdefault}{\mddefault}{\updefault}{Bits--$>\mathbf{s_2}$}%
}}}}
\put(2689,-3049){\makebox(0,0)[lb]{\smash{{\SetFigFont{12}{14.4}{\rmdefault}{\mddefault}{\updefault}{$\mathbf{S_1}(\mathbf{s_1})=\sum_{i=1}^{\lambda}x_i\mathbf{A_i}$}%
}}}}
\put(2676,-3499){\makebox(0,0)[lb]{\smash{{\SetFigFont{12}{14.4}{\rmdefault}{\mddefault}{\updefault}{$\mathbf{S_2}(\mathbf{s_2})=\sum_{i=\lambda+1}^{2\lambda}x_i\mathbf{A_i}$}%
}}}}
\put(2676,-4361){\makebox(0,0)[lb]{\smash{{\SetFigFont{12}{14.4}{\rmdefault}{\mddefault}{\updefault}{$\mathbf{S_g}(\mathbf{s_g})=\sum_{i=(g-1)\lambda+1}^{K}x_i\mathbf{A_i}$}%
}}}}
\put(1389,-4361){\makebox(0,0)[lb]{\smash{{\SetFigFont{12}{14.4}{\rmdefault}{\mddefault}{\updefault}{Bits--$>\mathbf{s_g}$}%
}}}}
\put(5189,-3674){\makebox(0,0)[lb]{\smash{{\SetFigFont{12}{14.4}{\rmdefault}{\mddefault}{\updefault}{$\mathbf{S}(\mathbf{s})=\sum_{i=1}^{g}\mathbf{S_i}(\mathbf{s_i})$}%
}}}}
\end{picture}%

%% file: g_dec.pstex_t
\begin{picture}(0,0)%
\includegraphics{g_dec.pstex}%
\end{picture}%
\setlength{\unitlength}{3947sp}%
\begingroup\makeatletter\ifx\SetFigFont\undefined%
\gdef\SetFigFont#1#2#3#4#5{%
  \reset@font\fontsize{#1}{#2pt}%
  \fontfamily{#3}\fontseries{#4}\fontshape{#5}%
  \selectfont}%
\fi\endgroup%
\begin{picture}(6089,1786)(340,-4173)
\put(1701,-2661){\makebox(0,0)[lb]{\smash{{\SetFigFont{12}{14.4}{\rmdefault}{\mddefault}{\updefault}{$\mathbf{\hat{X}_1}=\arg\min_{\mathbf{X_1}\in\mathscr{C}_1}\parallel \mathbf{Y}-\mathbf{X_1}\mathbf{H}\parallel_F^2$}%
}}}}
\put(1039,-3261){\makebox(0,0)[lb]{\smash{{\SetFigFont{12}{14.4}{\rmdefault}{\mddefault}{\updefault}{$\mathbf{Y}$}%
}}}}
\put(1689,-3161){\makebox(0,0)[lb]{\smash{{\SetFigFont{12}{14.4}{\rmdefault}{\mddefault}{\updefault}{$\mathbf{\hat{X}_2}=\arg\min_{\mathbf{X_2}\in\mathscr{C}_2}\parallel \mathbf{Y}-\mathbf{X_2}\mathbf{H}\parallel_F^2$}%
}}}}
\put(1701,-4049){\makebox(0,0)[lb]{\smash{{\SetFigFont{12}{14.4}{\rmdefault}{\mddefault}{\updefault}{$\mathbf{\hat{X}_g}=\arg\min_{\mathbf{X_g}\in\mathscr{C}_g}\parallel \mathbf{Y}-\mathbf{X_g}\mathbf{H}\parallel_F^2$}%
}}}}
\put(6414,-3399){\makebox(0,0)[lb]{\smash{{\SetFigFont{12}{14.4}{\rmdefault}{\mddefault}{\updefault}{$\mathbf{\hat{X}}$}%
}}}}
\end{picture}%

%% file: network.pstex_t
\begin{picture}(0,0)%
\includegraphics{network.pstex}%
\end{picture}%
\setlength{\unitlength}{3947sp}%
\begingroup\makeatletter\ifx\SetFigFont\undefined%
\gdef\SetFigFont#1#2#3#4#5{%
  \reset@font\fontsize{#1}{#2pt}%
  \fontfamily{#3}\fontseries{#4}\fontshape{#5}%
  \selectfont}%
\fi\endgroup%
\begin{picture}(3266,1727)(1119,-2851)
\put(1201,-1449){\makebox(0,0)[lb]{\smash{{\SetFigFont{12}{14.4}{\rmdefault}{\mddefault}{\updefault}{$S$}%
}}}}
\put(2300,-1447){\makebox(0,0)[lb]{\smash{{\SetFigFont{12}{14.4}{\rmdefault}{\mddefault}{\updefault}{$U_1$}%
}}}}
\put(4102,-1436){\makebox(0,0)[lb]{\smash{{\SetFigFont{12}{14.4}{\rmdefault}{\mddefault}{\updefault}{$D$}%
}}}}
\put(2642,-1966){\makebox(0,0)[lb]{\smash{{\SetFigFont{12}{14.4}{\rmdefault}{\mddefault}{\updefault}{$U_2$}%
}}}}
\put(3072,-2740){\makebox(0,0)[lb]{\smash{{\SetFigFont{12}{14.4}{\rmdefault}{\mddefault}{\updefault}{$U_R$}%
}}}}
\put(1664,-1299){\makebox(0,0)[lb]{\smash{{\SetFigFont{12}{14.4}{\rmdefault}{\mddefault}{\updefault}{$f_1$}%
}}}}
\put(2889,-1311){\makebox(0,0)[lb]{\smash{{\SetFigFont{12}{14.4}{\rmdefault}{\mddefault}{\updefault}{$g_1$}%
}}}}
\put(3076,-1686){\makebox(0,0)[lb]{\smash{{\SetFigFont{12}{14.4}{\rmdefault}{\mddefault}{\updefault}{$g_2$}%
}}}}
\put(1976,-1611){\makebox(0,0)[lb]{\smash{{\SetFigFont{12}{14.4}{\rmdefault}{\mddefault}{\updefault}{$f_2$}%
}}}}
\put(2251,-2074){\makebox(0,0)[lb]{\smash{{\SetFigFont{12}{14.4}{\rmdefault}{\mddefault}{\updefault}{$f_R$}%
}}}}
\put(3326,-2261){\makebox(0,0)[lb]{\smash{{\SetFigFont{12}{14.4}{\rmdefault}{\mddefault}{\updefault}{$g_R$}%
}}}}
\end{picture}%